\documentclass[letterpaper]{article}
\usepackage[top=3cm,bottom=2.5cm,hmargin=2.5cm]{geometry}
\usepackage{amsmath,amssymb,amsthm,mathtools}
\usepackage{mathrsfs,bm,color}
\usepackage{hyperref}
\usepackage{graphicx}
\usepackage{subcaption}
\usepackage{algorithm}
\usepackage[noend]{algpseudocode}
\usepackage{indentfirst}

\newtheorem{theorem}{Theorem}[section]
\newtheorem{corollary}[theorem]{Corollary}
\newtheorem{lemma}[theorem]{Lemma}

\newtheorem{definition}[theorem]{Definition}
\newtheorem{observation}[theorem]{Observation}

\newtheorem{property}{Property}

\title{A Lock-free Binary Trie}
\author{Jeremy Ko, University of Toronto, jerko@cs.toronto.edu}
\algdef{SE}[DOWHILE]{Do}{EndDo}{\algorithmicdo}[1]{\algorithmicwhile\ #1}%
\algdef{SE}[SUBALG]{Indent}{EndIndent}{}{\algorithmicend\ }%
\algtext*{Indent}
\algtext*{EndIndent}

\newcommand{\alglinenoNew}[1]{\newcounter{ALG@line@#1}}
\newcommand{\alglinenoPop}[1]{\setcounter{ALG@line}{\value{ALG@line@#1}}}
\newcommand{\alglinenoPush}[1]{\setcounter{ALG@line@#1}{\value{ALG@line}}}

\alglinenoNew{alg1}

\newcommand{\PRE}{\mathit{D}}
\newcommand{\RUALL}{\textit{RU-ALL}}
\newcommand{\UALL}{\textit{U-ALL}}
\newcommand{\PALL}{\textit{P-ALL}}
\newcommand{\key}{\mathit{key}}
\newcommand{\nextptr}{\mathit{next}}
\newcommand{\head}{\mathit{head}}

\newcommand{\pred}{\mathit{pred}}

\newcommand{\nNode}{\mathit{nNode}}
\newcommand{\uNode}{\mathit{uNode}}
\newcommand{\dNode}{\mathit{dNode}}
\newcommand{\dNodePtr}{\mathit{dNodePtr}}
\newcommand{\iNode}{\mathit{iNode}}
\newcommand{\pNode}{\mathit{pNode}}

\newcommand{\notifyList}{\mathit{notifyList}}

\newcommand{\RuallPosition}{\mathit{RuallPosition}}

\newcommand{\updateNode}{\mathit{updateNode}}
\newcommand{\updateNodeMax}{\mathit{updateNodeMax}}
\newcommand{\notifyThreshold}{\mathit{notifyThreshold}}

\newcommand{\type}{\mathit{type}}
\newcommand{\status}{\mathit{status}}
\newcommand{\delprednode}{\mathit{delPredNode}}

\newcommand{\height}{\mathit{height}}

\newcommand{\latest}{\textit{latest}}
\newcommand{\target}{\textit{target}} 
\newcommand{\stopflag}{\textit{stop}}
\newcommand{\done}{\textit{completed}}
\newcommand{\latestNext}{\mathit{latestNext}}
\newcommand{\delpredsecond}{\mathit{delPred2}}
\newcommand{\delpred}{\mathit{delPred}}
\newcommand{\threshold}{\mathit{upper0Boundary}}
\newcommand{\insthreshold}{\mathit{lower1Boundary}}

\newcommand{\Iruall}{I_{\mathit{ruall}}}
\newcommand{\Druall}{D_{\mathit{ruall}}}
\newcommand{\Iuall}{I_{\mathit{uall}}}
\newcommand{\Duall}{D_{\mathit{uall}}}
\newcommand{\Inotify}{I_{\mathit{notify}}}
\newcommand{\Dnotify}{D_{\mathit{notify}}}

\begin{document}
	\maketitle
	
	\begin{abstract}
		A  binary trie is a sequential data structure for a dynamic set on the universe $\{0,\dots,u-1\}$ supporting \textsc{Search} with $O(1)$ worst-case step complexity, and \textsc{Insert}, \textsc{Delete}, and \textsc{Predecessor} operations with $O(\log u)$ worst-case step complexity. 
		
		We give a wait-free implementation of a relaxed binary trie, 
		using read, write, CAS, and ($\log u$)-bit AND operations. 
		It supports all operations with the same worst-case step complexity as the sequential binary trie.
		However, \textsc{Predecessor} operations may not return a key when there are concurrent update operations.
		We use this as a component of a lock-free, linearizable implementation of a binary trie. It supports \textsc{Search} with $O(1)$ worst-case step complexity and \textsc{Insert}, \textsc{Delete} and \textsc{Predecessor} with $O(c^2 + \log u)$ amortized step complexity,  where $c$ is a measure of the contention.
		
		A lock-free binary trie is challenging to implement as compared to many other lock-free data structures
		because \textsc{Insert} and \textsc{Delete} operations perform a non-constant number of modifications to the binary trie in the worst-case 
		to ensure the
		correctness of \textsc{Predecessor} operations.
	\end{abstract}
	
	
	\section{Introduction}\label{section_intro}

	Finding the predecessor of a key in a dynamic set is a fundamental problem with wide-ranging applications in sorting, approximate matching and nearest neighbour algorithms. Data structures supporting \textsc{Predecessor} can be used to design efficient priority queues and mergeable heaps~\cite{BoasKZ77}, and have applications in IP routing~\cite{DegermarkBCP97} and bioinformatics~\cite{BieganskiRCR94, Martinez-Prieto16}. 
	
	A binary trie is a simple sequential data structure that maintains a dynamic set of keys $S$ from the universe $U = \{0,\dots,u-1\}$. \textsc{Predecessor}$(y)$ returns the largest key in $S$ less than key $y$, or $-1$ if there is no key smaller than $y$ in $S$.
	It
	supports \textsc{Search} with $O(1)$ worst-case step complexity and \textsc{Insert}, \textsc{Delete}, and \textsc{Predecessor} with $O(\log u)$ worst-case step complexity. It has $\Theta(u)$ space complexity. 
	
	The idea of a binary trie is to represent the prefixes of keys in $U$ in a sequence of $b+1$ arrays, $\PRE_i$, for $0 \leq i \leq b$, where  $b = \lceil \log_2 u \rceil$. Each array $\PRE_i$ has length $2^i$ and is indexed by the bit strings $\{0,1\}^i$. The array entry $\PRE_i[x]$ stores the bit 1 if $x$ is the prefix of length $i$ of some key in $S$, and 0 otherwise. 
	The sequence of arrays implicitly forms a perfect binary tree. The array entry $\PRE_i[x]$ represents the node at depth $i$ with length $i$ prefix $x$. Its left child is the node represented by $\PRE_{i+1}[x \cdot 0]$ and its right child is the node represented by $\PRE_{i+1}[x \cdot 1]$. Note that $\PRE_b$ (which represents the leaves of the binary trie) is a direct access table describing the set $S \subseteq U$. An example of a binary trie is shown in Figure~\ref{figure:seq_trie}.
	
	A \textsc{Search}$(x)$ operation reads $\PRE_b[x]$ and returns \textsc{True} if $\PRE_b[x]$ has value 1, and \textsc{False} otherwise.
	An \textsc{Insert}$(x)$ operation sets the bits of the nodes on the path from the leaf $\PRE_b[x]$ to the root to 1.  A \textsc{Delete}$(x)$ operation begins by setting $\PRE_b[x]$ to 0. It then traverses up the trie starting at $\PRE_b[x]$, setting the value of the parent of the current node to 0 if both its children have value 0.
	A \textsc{Predecessor}$(y)$ operation $pOp$ traverses up the trie starting from the leaf $\PRE_b[y]$ to the root.
	If the left child of each node on this path either has value 0 or is also on this path, then $pOp$ returns $-1$.
	Otherwise, consider the first node on this path whose left child $t$ has value 1 and is not on this path.
	Then starting from $t$, $pOp$ traverses down the right-most path of nodes with value 1 until it reaches a leaf $\PRE_b[w]$, and returns $w$.
	
	\begin{figure}[b]
		\centering
		\includegraphics[scale=0.45]{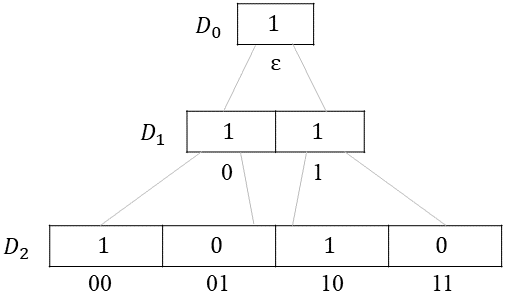}
		\caption{A sequential binary trie for the set $S = \{0,2\}$ from a universe $U = \{0,1,2,3\}$.}
		\label{figure:seq_trie}
	\end{figure}
	

	More complicated variants of sequential binary tries exist, such as van Emde Boas tries~\cite{Boas77}, x-fast tries and y-fast tries~\cite{Willard83}. Compared to binary tries, they improve the worst-case complexity of predecessor operations. Both x-fast tries and y-fast tries use hashing to improve the space complexity, and hence are not deterministic. Furthermore, none of these variants support constant time search.
	One motivation for studying lock-free binary tries is as a step towards
	efficient lock-free implementations of these data structures. 

	Universal constructions provide a framework to give (often inefficient) implementations of concurrent data structures from sequential specifications.
	A recent universal construction by Fatourou, Kallimanis, and Kanellou~\cite{FatourouKK19} can be used to implement a wait-free binary trie supporting operations with $O(N + \bar{c}(op) \cdot \log u)$ worst-case step complexity, where $N$ is the number of processes in the system, and $\bar{c}(op)$, is the \textit{interval contention} of $op$. This is the number of \textsc{Insert}, \textsc{Delete}, and \textsc{Predecessor} operations concurrent with the operation $op$. 
	Prior to this work, there have been no lock-free implementations of a binary trie or any of its variants without using universal constructions.
	
	There are many lock-free data structures that directly implement a dynamic set, including variations of linked lists, balanced binary search trees, skip lists~\cite{Pugh89}, and hash tables. There is also a randomized, lock-free implementation of a skip trie\cite{OshmanS13}. We discuss these data structures in more detail in Section~\ref{section_related}. 
	
	\textbf{Our contribution:}
	We give a lock-free implementation of a binary trie using registers, compare-and-swap (CAS) objects, and ($b+1$)-bounded min-registers. A min-write on a ($b+1$)-bit memory location can be easily implemented using a single ($b+1$)-bit AND operation, 
	so all these shared objects are supported in hardware.
	The amortized step complexity of our lock-free implementation of the binary trie is expressed using two other measures of contention.
	For an operation $op$, the \textit{point contention} of $op$,  denoted $\dot{c}(op)$, is the maximum number of concurrent \textsc{Insert}, \textsc{Delete}, and \textsc{Predecessor} operations at some point during $op$. The \textit{overlapping-interval contention}~\cite{OshmanS13} of $op$, denoted, $\tilde{c}(op)$ is the maximum interval contention of all update operations concurrent with $op$.
	Our implementation supports \textsc{Search} with  $O(1)$ worst-case step complexity, and \textsc{Insert} with $O(\dot{c}(op)^2 + \log u)$ amortized step complexity, and \textsc{Delete} and \textsc{Predecessor} operations with $O(\dot{c}(op)^2 + \tilde{c}(op) + \log u)$ amortized step complexity.
	In a configuration $C$ where there are $\dot{c}(C)$ concurrent \textsc{Insert}, \textsc{Delete}, and \textsc{Predecessor} operations, the implementation uses $O(u + \dot{c}(C)^2)$ space.
	Our data structure consists of a \textit{relaxed binary trie}, as well as auxiliary lock-free linked lists.
	Our goal was to maintain the $O(1)$ worst-case step complexity of \textsc{Search}, while avoiding $O(\bar{c}(op) \cdot \log u)$ terms in the amortized step complexity of the other operations seen in universal constructions of a binary trie.
	Our algorithms to update the bits of the relaxed binary trie and traverse the relaxed binary trie finish in $O(\log u)$ steps in the worst-case, and hence are \textit{wait-free}. 
	The other terms in the amortized step complexity are from updating and traversing the auxiliary lock-free linked lists.
	

	
	\textbf{Techniques:}
	A linearizable
	implementation of a concurrent data structure requires that all operations on the data structure 
	appear to happen atomically. 
	In our \textit{relaxed binary trie},
	predecessor operations are not linearizable. 
	We relax the properties maintained by the binary trie, so that
	the bit stored in each internal binary trie node does not always have to be accurate.
	At a high-level, we ensure that the bit at a node is accurate when there are no active update operations whose input key is a leaf of the subtrie rooted at the node. This allows us to design an efficient, wait-free algorithm for modifying the bits along a path in the relaxed binary trie. 
	
	Other lock-free data structures use the idea of relaxing properties maintained by the data structure.
	For example, in lock-free balanced binary search trees~\cite{BrownER14}, the balance conditions are often relaxed. This allows the tree to be temporarily unbalanced, provided there are active update operations.
	A node can be inserted into a tree by updating a single pointer. Following this, tree rotations may be performed, but they are only used to improve the efficiency of search operations. 
	Lock-free skip lists~\cite{FomitchevR04} relax the properties about the heights of towers of nodes.
	A new node is inserted into the bottom level of a skip list using a single pointer update. Modifications that add nodes into the
	linked lists at higher levels only improve the efficiency of search operations. 
	
	
	Relaxing the properties of the binary trie is more complicated than these two examples because it affects correctness, not just efficiency.
	For a predecessor operation to traverse a binary trie using the sequential algorithm, the bit of each node must be the logical OR of its children. This is not necessarily the case in our relaxed binary trie. 
	For example, a node in the relaxed binary trie with value 1 may have two children which each have value 0.
	
	The second way our algorithm is different from other data structures is how operations help each other complete.
	Typical lock-free data structures, including those based on universal constructions, use \textit{helping}
	when concurrent update operations require modifying the same part of a data structure: a process may help a different operation complete by making modifications to the data structure on the other operation's behalf. 
	This technique is efficient when a small, constant number of modifications to the data structure need to be done atomically. For example, many lock-free implementations of linked lists and binary search trees can insert a new node by repeatedly attempting CAS to modify a single pointer.
	For a binary trie, update operations require updating $O(\log u)$ bits on the path from a leaf to the root. 
	An operation that helps all these updates complete would have $O(\dot{c}(op) \cdot \log u)$ amortized step complexity. 
	
	Predecessor operations that cannot traverse a path through the relaxed binary trie
	do not help concurrent update operations complete. For our linearizable, lock-free binary trie, our approach is to have update operations and predecessor operations announce themselves in an update announcement linked list and a predecessor announcement linked list, respectively.
	We guarantee that a predecessor operation will either learn about concurrent update operations by traversing the update announcement linked list, or it will be notified by concurrent update operations via the predecessor announcement linked list.  A predecessor operation uses this information to determine a correct return value, especially when it cannot complete its traversal of the relaxed binary trie.
	This is in contrast to other lock-free data structures that typically announce operations so that they can be completed by concurrent operations in case the invoking processes crash.

	Section~\ref{section_model} describes the asynchronous shared memory model. In Section~\ref{section_related}, we describe related lock-free data structures and compare them to our lock-free binary trie. 
	In Section~\ref{section_relaxed_binary_trie} we give the specification of a relaxed binary trie, give a high-level description of our wait-free implementation, present the pseudocode of the algorithm, and prove it correct.
	In Section~\ref{section_implementation}, we give a high-level description of our implementation of a lock-free binary trie, present the pseudocode of the algorithm and a more detailed explanation, prove it is linearizable, and analyze its amortized step complexity.
	We conclude in Section~\ref{section_conclusion}.
	.
	

	\section{Model}\label{section_model}

	We use an asynchronous shared memory model. Shared memory consists of a collection of shared variables accessible by a system of $N$ processes. The primitives supported by these variables are performed atomically. 
	A register is an object supporting \textsc{Write}$(w)$, which stores $w$ into the object, and \textsc{Read}$()$, which returns the value stored in the object.
	CAS$(r, old, new)$ compares the value stored in object $r$ with the value $old$. If the two values are the same, the value stored in $r$ is replaced with the value $new$ and \textsc{True} is returned; otherwise \textsc{False} is returned. A min-register is an object that stores a value, and supports \textsc{Read}$()$, which returns the value of the object, and \textsc{MinWrite}$(w)$, which changes the value of the object to $w$ if $w$ is smaller than its previous value.
	
	A \textit{configuration}\index{configuration} of a system consists of the values of all shared objects and the states of all processes. A \textit{step}\index{step} by a process either accesses or modifies a shared object, and can also change the state of the process. An \textit{execution}\index{execution} is an alternating sequence of configurations and steps, starting with a configuration. 
	
	An \textit{abstract data type} is a collection of objects and types of operations that satisfies certain properties. A \textit{concurrent data structure} for the abstract data type provides representations of the objects in shared memory and algorithms for the processes to perform operations of these types. An operation on a data structure by a process becomes \textit{active}\index{active} when the process performs the first step of its algorithm. The operation becomes \textit{inactive}\index{inactive} after the last step of the algorithm is performed by the process. This last step may include a response to the operation. 
	The \textit{execution interval} of the operation consists of all steps between its first and last step (which may include steps from processes performing other operations). 
	In the initial configuration, the data structure is empty and there are no active operations. 
	
	
	We consider concurrent data structures that are \textit{linearizable}~\cite{HerlihyW90}, which means that its operations appear to occur atomically.
	One way to show that a concurrent data structure is linearizable is by defining a \textit{linearization function}, which, for all executions of the data structure, maps all completed operations and a subset of the uncompleted operations in the execution to a step or configuration, called its \textit{linearization point}. 
	The linearization points of these operations must satisfy two properties.
	First, each linearized operation is mapped to a step or configuration within its execution interval. Second, the return value of the linearized operations must be the same as in the execution in which all the linearized operations are performed atomically in the order of their linearization points,
	no matter how operations with the same linearization point are ordered.
	
	A concurrent data structure is \textit{strong linearizability} if its linearization function satisfies an additional \textit{prefix preserving} property:
	For all its executions $\alpha$ and for all prefixes $\alpha'$ of $\alpha$,
	if an operation is assigned a linearization point for $\alpha'$,
	then it is assigned the same linearization point for $\alpha$,
	and if it is assigned a linearization point for $\alpha$ that occurs during $\alpha'$,
	then it is assigned the same linearization point for $\alpha'$.
	This means that the linearization point of each operation is determined as steps are taken in the execution and cannot depend on steps taken later in the execution.
	This definition, phrased somewhat differently, was introduced by Golab, Higham and Woelfel~\cite{GolabHW11}.
	A concurrent data structure is \textit{strongly linearizable with respect to a set of operation types $\cal{O}$} if it has a linearization function that is only defined on operations whose types belong to $\cal{O}$. 
	
	
	A \textit{lock-free}\index{lock-free} data structure guarantees that,  in every infinite execution, whenever there are active operations, at least one will complete within a finite number of steps, provided at least one of the processes performing these operations continues to take steps in the execution.
	Note that any particular operation may not complete, provided other operations do complete.
	A \textit{wait-free}\index{wait-free}  data structure guarantees that, in every infinite execution, every operation completes within a finite number of steps by the process that invoked the operation, provided that 
	process continues to take steps in the execution.
	
	The \textit{worst-case step complexity}\index{worst-case step complexity} of an operation is the maximum number of steps taken by a process to perform any instance of this operation in any execution. The \textit{amortized step complexity}\index{amortized step complexity} of a data structure is the maximum number of steps in any execution consisting of operations on the data structure, divided by the number operations invoked in the execution. One can determine an upper bound on the amortized step complexity by assigning an \textit{amortized cost to each operation}, such that for all possible executions $\alpha$ on the data structure, the total number of steps taken in $\alpha$ is at most the sum of the amortized costs of the operations in $\alpha$. 
	

	\section{Related Work}\label{section_related}
	
	
	
	
	
	
	
	
	In this section, we discuss related lock-free data structures and the techniques used to implement them. We first describe simple lock-free linked lists, which are a component of our binary trie. We next describe search tree implementations supporting \textsc{Predecessor}, and discuss the general techniques used. Next we describe implementations of a Patricia trie and a skip trie. Finally, we discuss some universal constructions and general techniques for augmenting existing data structures.

	There are many existing implementations of lock-free linked lists~\cite{Harris01, Shafiei15, Valois95}. The implementation with best amortized step complexity is by Fomitchev and Ruppert~\cite{FomitchevR04}. It supports \textsc{Insert}, \textsc{Delete}, \textsc{Predecessor}, and \textsc{Search} operations $op$ with  $O(L(op) + \dot{c}(op))$ amortized step complexity, where $L(op)$ is the number of nodes in the data structure at the start of $op$. When $L(op)$ is in $O(\dot{c}(op))$, such as in the case of our binary trie implementation, operations have $O(\dot{c}(op))$ amortized step complexity.
	This implementation uses \textit{flagging}: a pointer that is flagged indicates an update operation wishes to modify it. Concurrent update operations that also need to modify this pointer must help the  operation that flagged the pointer to complete before attempting to perform their own operation.
	After an update operation completes helping, it \textit{backtracks} (using back pointers set in deleted nodes) to a suitable node in the linked list before restarting its own operation.
	
	Ellen, Fatourou, Ruppert, and van Breugel \cite{EllenFRB10} give the first provably correct lock-free implementation of an unbalanced binary search tree using CAS.  Update operations are done by flagging a constant number of nodes, followed by a single pointer update. A flagged node contains a pointer to an operation record, which includes enough information about the update operation so that other processes can \textit{help} it complete. Ellen, Fatourou, Helga, and Ruppert \cite{EllenFHR13} improve the efficiency of this implementation so each  operation $op$ has $O(h(op) + \dot{c}(op))$ amortized step complexity, where $h(op)$ is the height of the binary search tree when $op$ is invoked. This is done by allowing an update operation to backtrack along its search paths by keeping track of nodes it visited in a stack. There are many other implementations of lock-free unbalanced binary search trees~\cite{BrownA12, ChatterjeeDT13, HowleyJ12, NatarajanM14}.
	
	There are also many implementations of 
	lock-free balanced binary search trees~\cite{BraginskyP12, Brown14, BrownThesis17, DrachslerVY14}. 
	Brown, Ellen, and Ruppert designed a lock-free balanced binary search tree~\cite{BrownER14}
	by implementing more powerful primitives, called LLX and SCX, from CAS~\cite{BrownER13}. 
	These primitives are generalizations of LL and SC. SCX allows a single field to be updated by a process provided a specified set of nodes have not been modified since that process last performed LLX on them.
	Ko~\cite{Ko18} shows that a version of their lock-free balanced binary search tree has good amortized step complexity. 
	
	Although a binary trie represents a perfect binary tree, the techniques we use to implement a binary trie are quite different than those that have been used to implement lock-free binary search trees.
	The update operations of a binary trie may require updating the bits of all nodes on the path from a leaf to the root. LLX and SCX do not facilitate this because SCX only updates a single field. 
	
	Brown, Prokopec, and Alistarh~\cite{BrownPA20} give an implementation of an interpolation search tree supporting \textsc{Search} with $O(N + \log L(op))$ amortized step complexity, and \textsc{Insert} and \textsc{Delete} with $O(\bar{c}_{avg}(N + \log L(op))$ amortized step complexity, where $\bar{c}_{avg}$ is the average interval contention of the execution. Their data structure is a balanced search tree where nodes can have large degree.
	A node containing $n$ keys in its subtree is \textit{ideally balanced} when it has $\sqrt{n}$ children, each of which contains $\sqrt{n}$ nodes in its subtree. Update operations help replace subtrees that become too unbalanced with ideally balanced subtrees. When the input distribution of keys is well-behaved, \textsc{Search} can be performed with $O(P + \log\log L(op))$ expected amortized step complexity and update operations can be performed with $O(\bar{c}_{avg}(P + \log\log L(op))$ expected amortized step complexity. Their implementation relies on the use of double-compare single-swap (DCSS).
	DCSS is not a primitive typically supported in hardware, although there exist implementations of DCSS from CAS~\cite{Arbel-Raviv017, GiakkoupisGW21, HarrisFP02}.
	
	Shafiei~\cite{Shafiei13} gives an implementation of a Patricia trie. The data structure is similar to a binary trie, except that only internal nodes whose children both have value 1 are stored. In addition to \textsc{Search}, \textsc{Insert} and \textsc{Delete}, it supports \textsc{Replace}, which removes a key and adds another key in a possibly different part of the trie. Her implementation uses a variant of the flagging technique described in \cite{EllenFRB10}, except that it can flag two different nodes with the same operation record.
	
	Oshman and Shavit~\cite{OshmanS13} introduce a randomized data structure called a skip trie. It combines an x-fast trie with a truncated skip list whose max height is $\log_2 \log_2 u$
	(i.e. it is a y-fast trie whose balanced binary search trees are replaced with a truncated skip list).
	Only keys that are in the top level of the skip list are in the x-fast trie. 
	They give a lock-free implementation of a skip trie supporting \textsc{Search}, \textsc{Insert}, \textsc{Delete}, and \textsc{Predecessor} operations with $O(\tilde{c}(op) + \log\log u)$ expected amortized step complexity from registers, CAS and DCSS. 
	Their x-fast trie uses lock-free hash tables~\cite{ShalevS03}. 
	Their x-fast trie implementation supports update operations with $O(\dot{c}(op) \cdot \log u)$ expected amortized step complexity. 
	Their skip list implementation is similar to Fomitchev and Ruppert's skip list implementation~\cite{FomitchevR04}.
	In the worst-case (for example, when the height of the skip list is 0), a skip trie performs the same as a linked list, so \textsc{Search} and \textsc{Predecessor} take $\Theta(n)$ steps, even when there are no concurrent updates.
	Our lock-free binary trie implementation is deterministic, does not rely on hashing, uses primitives supported in hardware, and always performs well when there are no concurrent update operations. Furthermore, \textsc{Search} operations in our binary trie complete in a constant number of reads in the worst-case.
	
	Fatourou and Ruppert~\cite{FatourouR24} recently implemented an augmented wait-free binary trie.
	They introduce a technique that allows augmenting data structures with additional information to support new operations. 
	For example, suppose each node is associated with a \textit{sum} field, which stores the number of keys in the subtrie rooted at that node.
	Then they can implement a \textsc{Size} operation, which returns the number of keys in the trie.
	In their implementation, each binary trie node points to an immutable version node, which represents 
	the latest version of the fields stored at that binary trie node.
	Each version node stores a sum field and pointers to version nodes of its left and right children. 
	Because version nodes are immutable,
	a snapshot of the binary trie can be taken by reading the version node pointed to by the root of the binary trie. 
	Binary trie nodes can only be updated by making them point to new version nodes. 
	\textsc{Insert}$(k)$ or \textsc{Delete}$(k)$ operations start by updating the version node pointed to by the leaf with key $k$. Each node on the path from this leaf to the root is updated
	by pointing to a new version containing pointers to versions of the node's children. 
	The update operation is only linearized once the version of root is updated.
	A \textsc{Search}$(k)$ operation begins by taking a snapshot of the binary trie by reading the version node pointed to by the root, and then traversing down to the leaf version node with key $k$ of this snapshot.
	It cannot simply read the version pointed to by the leaf with key $k$, because the update operation that last updated this leaf node may not yet be linearized.
	Their \textsc{Search}, \textsc{Insert}, and \textsc{Delete} operations have $O(\log u)$ worst-case step complexity.
	
	The first universal constructions were by Herhily~\cite{Herlihy91, Herlihy93}. To achieve wait-freedom, he introduced an announcement array where operations announce themselves. Processes help perform these announced operations in a round-robin order.
	Barnes~\cite{Barnes93} gives a universal construction for obtaining lock-free data structures. He introduces the idea of using operation records to facilitate helping.
	A lot of work has been done to give more efficient universal constructions~\cite{AfekDT95, AndersonM95, ChuongER10, FatourouK11}, as well as prove limitations of universal constructions~\cite{BedinLMPP21, EllenFKMT16}.
	Recently, Fatourou, Kallimanis, and Kanellou~\cite{FatourouKK19} designed a universal construction that can give wait-free data structures supporting operations with  $O(N + \bar{c}(op) \cdot W)$ worst-case step complexity, where $W$ is the  worst-case step complexity to perform the operation on the sequential data structure it is based on.
	Operations announcement themselves in an announcement array and are executed in ordered batches.
	
	
	Lock-free data structures can be augmented to support iterators, snapshots, and range queries~\cite{Chatterjee17, FatourouPR19, PetrankT13, ProkopecBBO12}.
	Wei et al. \cite{WeiBBFR021} give a simple technique to take snapshots of concurrent data structures in constant time. 
	This is done by implementing a \textit{versioned CAS object} that allows old values of the object to be read.
	The number of steps needed to read the value of a versioned CAS object at the time of a snapshot is equal to the number of times its value changed since the snapshot was taken. 
	Provided update operations only perform a constant amortized number of successful versioned CAS operations, balanced binary search trees can be augmented to support \textsc{Predecessor} with  $O(\bar{c}(op) + \log L(op)) = O(\dot{c}(op) + \log L(op))$ amortized step complexity.

	\section{Relaxed Binary Trie}\label{section_relaxed_binary_trie}
	
	In this section, we describe our relaxed binary trie, which is used as a component of our lock-free binary trie. 
	We begin by giving the formal specification of the relaxed binary trie in Section~\ref{section_relaxed_binary_trie_specification}. 
	In Section~\ref{section_relaxed_binary_trie_representation}, we describe how our implementation is represented in memory. In Section~\ref{section_relaxed_binary_trie_high_level}, we give a high-level description of our algorithms for each operation. In Section~\ref{section_relaxed_binary_trie_pseudocode}, we give a detailed description of our algorithms for each operation and its pseudocode. Finally, in Section~\ref{section_relaxed_binary_trie_correctness}, we show that our implementation satisfies the specification.
	
	\subsection{Specification}\label{section_relaxed_binary_trie_specification}
	
	A relaxed binary trie is a concurrent data structure maintaining a dynamic set $S$ from the universe $U = \{0,\dots,u-1\}$ that supports the following strongly linearizable operations:
	\begin{itemize}
		\item \textsc{TrieInsert}$(x)$, which adds key $x$ into $S$ if it is not already in $S$,
		\item \textsc{TrieDelete}$(x)$, which removes key $x$ from $S$ if it is in $S$, and
		\item \textsc{TrieSearch}$(x)$, which returns \textsc{True} if  key $x \in S$, and \textsc{False} otherwise.
	\end{itemize}
	It additionally supports the (non-linearizable) \textsc{RelaxedPredecessor}$(y)$ operation. Informally, we allow a \textsc{RelaxedPredecessor}$(y)$ operation to sometimes return $\bot$, indicating concurrent update operations are interfering with it, preventing it from returning an answer. However, when there are no concurrent update operations, \textsc{RelaxedPredecessor}$(y)$ must return the correct predecessor of $y$.
	Its formal specification relies on a few definitions.
	
	Because the relaxed binary trie is strongly linearizable with respect to all of its update operations, 
	it is possible to determine the value of the set $S$ represented by the data structure in every configuration of every execution
	from the sequence of linearization points of the update operations prior to this configuration.
	For any execution of the relaxed binary trie and any key $x \in U$, consider the sequence $\sigma$ of \textsc{TrieInsert}$(x)$ and \textsc{TrieDelete}$(x)$ operations in the order of their linearization points. A \textsc{TrieInsert}$(x)$ operation is \textit{$S$-modifying} if it is the first \textsc{TrieInsert}$(x)$ operation in $\sigma$, or if it is the first \textsc{TrieInsert}$(x)$ operation that follows a \textsc{TrieDelete}$(x)$ operation. In other words, a \textsc{TrieInsert}$(x)$ operation is $S$-modifying if it successfully adds the key $x$ to $S$.
	Likewise, a \textsc{TrieDelete}$(x)$ operation is \textit{$S$-modifying} if it is the first \textsc{TrieDelete}$(x)$ operation that follows a \textsc{TrieInsert}$(x)$ operation.
	A key $x$ is \textit{completely present} throughout a \textsc{RelaxedPredecessor} operation, $pOp$, if there is an $S$-modifying \textsc{TrieInsert}$(x)$ operation, $iOp$, that completes before the invocation of $pOp$ and there is no $S$-modifying \textsc{TrieDelete}$(x)$ operation that is linearized after $iOp$ but before the end of $pOp$.
	
	
	\textbf{Specification of RelaxedPredecessor:}
	Let $pOp$ be a completed \textsc{RelaxedPredecessor}$(y)$ operation in some execution.
	Let $k$ be the largest key less than $y$ that is completely present throughout $pOp$, or $-1$ if no such key exists. Then $pOp$ returns a value in  $\{\bot\} \cup \{k,\dots,y-1\}$ such that:
	\begin{itemize}
		
		\item If $pOp$ returns $\bot$, then there exists a key $x$, where $k < x < y$, such that the last $S$-modifying update operation with key $x$ linearized prior to the end of $pOp$ is concurrent with $pOp$.
		
		
		\item If $pOp$ returns a key $x > k$, then $x \in S$ sometime during $pOp$.
	\end{itemize}
	
	\noindent These properties imply that if, for all $k < x < y$, the $S$-modifying update operation with key $x$ that was last linearized prior to the end of $pOp$ is not concurrent with $pOp$, then $pOp$ returns $k$. Furthermore, $k$ is the predecessor of $y$ throughout $pOp$.
	To see why, first notice that in this scenario, for any key $x$, where $k < x < y$, $x$ is not in $S$ at any point during $pOp$. Otherwise the last $S$-modifying update operation with key $x$ that was linearized prior to the end of $pOp$ must be a \textsc{TrieInsert}$(x)$ operation that completed before the start of $pOp$. Then $x$ is completely present throughout $pOp$, contradicting the definition of $k$.
	Since $x$ is not in $S$ at any point during $pOp$,
	the second property of the specification implies that $pOp$ does not return a key greater than $k$. Furthermore, the first property of the specification states that $pOp$ does not return $\bot$.
	Then $pOp$ must return $k$, which is the largest key in $S$ less than $y$ at any point during $pOp$.
	
	\subsection{Our Relaxed Binary Trie Implementation}\label{section_relaxed_binary_trie_representation}
	Our wait-free implementation of a relaxed binary trie supports \textsc{TrieSearch} with $O(1)$ worst-case step complexity, and \textsc{TrieInsert}, \textsc{TrieDelete}, and \textsc{RelaxedPredecessor} with $O(\log u)$ worst-case step complexity.
	We first describe the major components of the relaxed binary trie and how it is stored in memory. 
	
	Like the sequential binary trie, the relaxed binary trie consists of a collection of arrays, $\PRE_i$ for $0 \leq i \leq b = \lceil \log_2 u \rceil$. Each array $\PRE_i$, for $0 \leq i \leq b$, has length $2^i$ and represents the nodes at depth $i$ of the relaxed binary trie. 
	
	An \textit{update node} is created by a \textsc{TrieInsert} or \textsc{TrieDelete} operation. It is an INS node if it is created by a \textsc{TrieInsert} operation, or a DEL node if it is created by a \textsc{TrieDelete} operation.
	It includes the input key of the operation that created it.
	
	In a configuration $C$, the \textit{latest update operation with key $x$} is the last $S$-modifying update operation with key $x$ linearizezd prior to $C$. The \textit{latest update  node with key $x$} is the update node created by this update operation.
	There is an array $\latest$ indexed by each key in $U$, where $\latest[x]$ contains a pointer to the latest update node with key $x$.
	So the update node pointed to by $\latest[x]$ is an INS node if and only if $x \in S$.
	In the initial configuration, when $S = \emptyset$, $\latest[x]$ points to a dummy DEL node.
	

	
	
	Recall that, in the sequential binary trie, each binary trie node $t$ contains the bit 1 if there is leaf in its subtrie whose key is in $S$, and 0 otherwise. These bits need to be accurately set for the correctness of predecessor operations. 
	In our relaxed binary trie, each binary trie node has an associated value, called its \textit{interpreted bit}.
	The interpreted bit of a leaf with key $x$ is 1 if and only if $x \in S$. 
	
	For each internal binary trie node $t$,
	let $U_t$ be the set of keys of the leaves contained in the subtrie rooted at $t$.
	When there are no active update operations with keys in $U_t$, the interpreted bit of $t$ is the logical OR of the interpreted bits of the leaves of the subtrie rooted at $t$. More generally, our relaxed binary trie maintains the following two properties concerning its interpreted bits. For all binary trie nodes $t$ and configurations $C$:
	\begin{enumerate}
		
		\item[IB0] If $U_t \cap S = \emptyset$ and for all $x \in U_t$, 
		either there has been no $S$-modifying \textsc{TrieDelete}$(x)$ operation or
		the last $S$-modifying \textsc{TrieDelete}$(x)$ operation linearized prior to $C$ is no longer active,
		then the interpreted bit of $t$ is 0 in $C$.
		
		\item[IB1] If there exists $x \in U_t \cap S$ such that the last $S$-modifying \textsc{TrieInsert}$(x)$ operation linearized prior to $C$ is no longer active,
		then the interpreted bit of $t$ is 1 in $C$.
	\end{enumerate}
	\noindent When there are active update operations with a key in $U_t$, the interpreted bit of the binary trie node $t$ may be different from the bit stored in $t$ of the sequential binary trie representing the same set.
	
	The interpreted bit of $t$ is not physically stored in $t$, but is, instead, computed from the update node pointed to by $\latest[x]$, for some key $x \in U_t$.
	Each internal binary trie node $t$ stores this key $x$.
	The interpreted bit of $t$ \textit{depends on} the update node, $\uNode$, pointed to by $\latest[x]$. If $\uNode$ is an INS node, the interpreted bit of $t$ is 1.
	When $\uNode$ is a DEL node, the interpreted bit of $t$ is determined by two thresholds, $\uNode.\threshold$ and $\uNode.\insthreshold$.
	In this case, the interpreted bit of $t$ is
	\begin{itemize}
		\item 1 if $t.\height \geq \uNode.\insthreshold$,
		\item 0 if $t.\height < \uNode.\insthreshold$ and $t.\height \leq \uNode.\threshold$, and
		\item 1 otherwise.
	\end{itemize}
	
	\noindent Only \textsc{TrieInsert} operations decrease $\insthreshold$ (which can change the interpreted bit of a binary trie node to 1) and only \textsc{TrieDelete} operations increase $\threshold$ (which can change the interpreted bit of a binary trie node to 0). We discuss these thresholds in more detail when describing \textsc{TrieInsert} and \textsc{TrieDelete} in the following section.

	
	\subsection{High-Level Algorithm Description}\label{section_relaxed_binary_trie_high_level}
	
	A \textsc{TrieSearch}$(x)$ operation reads the update node pointed to by $\latest[x]$,
	returns \textsc{True} if it is an INS node, and returns \textsc{False} if it is a DEL node.
	
	An \textsc{TrieInsert}$(x)$ or \textsc{TrieDelete}$(x)$ operation, $uOp$, begins by reading the update node pointed to by $\latest[x]$. If it has the same type as $uOp$, then $uOp$ can return because $S$ does not need to be modified. Otherwise $op$ creates a new update node $\uNode$ with key $x$. 
	It then 
	attempts to change $\latest[x]$ to point to $\uNode$ using CAS. If successful, the operation is linearized at this step. 
	If multiple update operations with key $x$ concurrently attempt change $\latest[x]$ to point to their own update node, exactly one will succeed. 
	
	In the case that $uOp$ successfully changes $\latest[x]$ to point to $\uNode$, $uOp$ must then update the interpreted bits of the relaxed binary trie.
	Both \textsc{TrieInsert} and \textsc{TrieDelete} operations update the interpreted bits of the relaxed binary trie in manners similar to the sequential data structure. 
	A \textsc{TrieInsert}$(x)$ operation traverses from the leaf with key $x$ to the root and sets the interpreted bits along this path to 1 if they are not already 1. This is described in more detail in Section~\ref{section_relax_binary_trie_ins}. A \textsc{TrieDelete}$(x)$ operation traverses the binary trie starting from the leaf with key $x$ and proceeds to the root. It changes the interpreted bit of a binary trie node on this path to 0 if both its children have interpreted bit 0, and returns otherwise. This is described in more detail in Section~\ref{section_relax_binary_trie_del}.

	\subsubsection{TrieInsert}\label{section_relax_binary_trie_ins}
	
	Consider a \textsc{TrieDelete}$(x)$ operation, $iOp$, and let $\iNode$ be the INS node it created. So $\iNode$ is the update node pointed to by $\latest[x]$. 
	Let $t$ be a binary trie node $iOp$ encounters as it is updating the interpreted bits of the binary trie.
	If $t$ already has interpreted bit 1, then it does not need to be updated. This can happen when the interpreted bit of $t$ depends on an INS node (for example, when $t$ stores key $x$ and hence depends on $\iNode$).
	So suppose the interpreted bit of $t$ is 0. In this case, $t$ stores key $x' \neq x$ and its interpreted bit depends on a DEL node $\dNode$.
	Only a \textsc{Delete} operation can change the key stored in $t$ and it can only change the key to its own key. 
	We do not allow \textsc{Insert} operations to change this key to avoid concurrent \textsc{Delete} operations from repeatedly interfering with an \textsc{Insert} operation.
	Instead, \textsc{Insert} operations can modify $\dNode.\insthreshold$ to change the interpreted bit of $t$ from 0 to 1.
	This is a min-register whose value is initially $b+1$, which is greater than the height of any binary trie node. 
	All binary trie nodes that depend on $\dNode$
	and whose  height is at least the value of $\dNode.\insthreshold$  have interpreted bit 1.
	Therefore, to change the interpreted bit of $t$ from 0 to 1, 
	$iOp$ can perform \textsc{MinWrite}$(t.\height)$ to $\dNode.\insthreshold$. This also changes the interpreted bit of all ancestors of $t$ that depend on $\dNode$ to 1. A min-register is used so that modifying $\dNode.\insthreshold$ never changes the interpreted bit of any binary trie node from 1 to 0. 
	
	An example execution of an \textsc{TrieInsert} operation updating the interpreted bits of the relaxed binary trie is shown in Figure~\ref{figure:ins_execution}. 
	Blue rectangles represent INS nodes, while red rectangles represent DEL nodes.
	The number in each binary trie node is its interpreted bit.
	The dashed arrow from an internal binary trie node points to the update node it depends on. Note that the dashed arrow is not a physical pointer stored in the binary trie node.
	Under each update node are the values of its $\insthreshold$, abbreviated $l1b$, and $\threshold$, abbreviated $u0b$. 	
	Figure~\ref{figure:ins_execution}(a) shows a possible state of the data structure where $S = \emptyset$. In Figure~\ref{figure:ins_execution}(b), an \textsc{Insert}$(0)$ operation, $iOp$, activates its newly added INS node in $\latest[0]$. This simultaneously changes the interpreted bit of the leaf with key 0 and its parent from 0 to 1 in a single step.
	In Figure~\ref{figure:ins_execution}(c), $iOp$ changes the interpreted bit of the root from 0 to 1. This is done using a \textsc{MinWrite}, which changes the $\insthreshold$ of the DEL node in $\latest[3]$ (i.e. the update node that the root depends on) from 3 to the height of the root. 
	
	\begin{figure}[tb]
		\centering
		\begin{subfigure}[b]{0.3\textwidth}
			\centering
			\includegraphics[width=\textwidth]{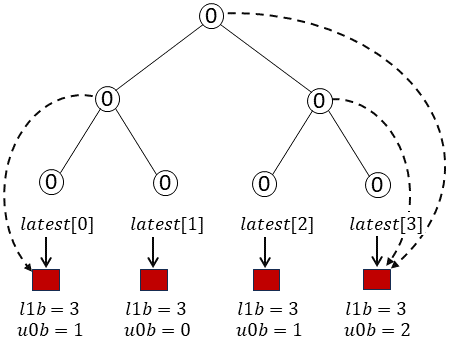}
			\caption{}
			\label{figure:ins_execution_a}
		\end{subfigure}
		\hfill
		\begin{subfigure}[b]{0.3\textwidth}
			\centering
			\includegraphics[width=\textwidth]{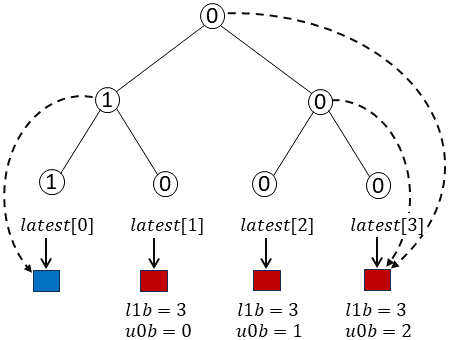}
			\caption{}
			\label{figure:ins_execution_b}
		\end{subfigure}
		\hfill
		\begin{subfigure}[b]{0.3\textwidth}
			\centering
			\includegraphics[width=\textwidth]{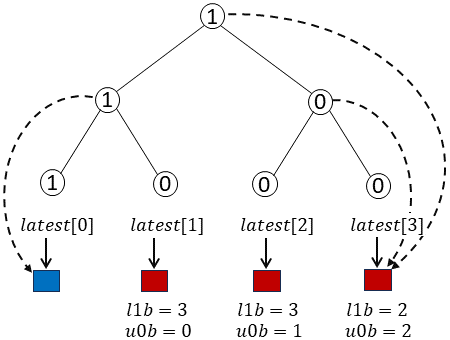}
			\caption{}
			\label{figure:ins_execution_c}
		\end{subfigure}
		\caption{An example of a \textsc{TrieInsert}$(0)$ operation setting the interpreted bits of the binary trie nodes from its leaf with key 0 to the root to 1.
		}
		\label{figure:ins_execution}
	\end{figure}
	
	\subsubsection{TrieDelete}\label{section_relax_binary_trie_del}
	
	Consider a \textsc{TrieDelete}$(x)$ operation, $dOp$, and let $\dNode$ be the DEL node it created, where $\dNode$ is the update node pointed to by $\latest[x]$. 
	This means that the leaf with key $x$ has interpreted bit 0. 
	
	Let $t$ be an internal binary trie node on the path from the leaf with key $x$ to the root. Suppose $dOp$ successfully changed the interpreted bit of one of $t$'s children to 0. 
	If the interpreted bit of the other child of $t$ is 0,
	then $dOp$ attempts to change the interpreted bit of $t$ to 0. First, $dOp$ tries to change the update node that $t$ depends on to $\dNode$ by changing the key stored in $t$ to $x$.
	After a constant number of reads and at most 2 CAS operations, our algorithm guarantees that if $dOp$ does not successfully change $t$ to depend on $\dNode$, then for some $y \in U_t$, a latest \textsc{Delete}$(y)$ operation, $dOp'$, changed $t$ to depend on the DEL node $\dNode'$ created by $dOp'$. 
	In this case, $dOp'$ will change the interpreted bit of $t$ to 0 on $dOp$'s behalf, so $dOp$ can stop updating the interpreted bits of the binary trie.
	Suppose $dOp$ does successfully change $t$ to depend on $\dNode$.
	To change the interpreted bit of $\dNode$ to 0, $dOp$ writes $t.\height$ into $\dNode.\threshold$, which is a register with initial value 0.
	This indicates that all binary trie nodes at height $t$ and below that depend on $\dNode$ have interpreted bit 0. Only $dOp$, the creator of $\dNode$, writes to $\dNode.\threshold$. 
	Since $dOp$ changes the interpreted bits of binary trie nodes in order from the leaf with key $x$ up to the root, $\dNode.\threshold$ is only ever incremented by 1 starting from 0. 
	
	An example execution of \textsc{TrieDelete} operations updating the interpreted bits of the relaxed binary trie is shown in Figure~\ref{figure:del_execution}. 
	Figure~\ref{figure:del_execution}(a) shows a possible state of the data structure where $S = \{0,1\}$.
	In Figure~\ref{figure:del_execution}(b), a \textsc{TrieDelete}$(0)$, $dOp$, and a \textsc{TrieDelete}$(1)$, $dOp'$, activate their newly added DEL nodes. This removes the keys 0 and 1 from $S$. The leaves with keys 0 and 1 have interpreted bit 0. 
	In Figure~\ref{figure:del_execution}(c), $dOp'$ sees that its sibling leaf has interpreted bit 0. Then $dOp'$ successfully changes the left child of the root to depend on its DEL node, while $dOp$ is unsuccessful and returns.
	In Figure~\ref{figure:del_execution}(d), $dOp'$ increments the $\threshold$ of its DEL node, so it is now equal to the height of the left child of the root. This changes the interpreted bit of the left child of the root to 0. 
	In Figure~\ref{figure:del_execution}(e), $dOp'$ sees that the right child of the root has interpreted bit 0, so the interpreted bit of the root needs to be updated. So $dOp'$ changes the root to depend on its DEL node. 
	In Figure~\ref{figure:del_execution}(f),  $dOp'$ increments the $\threshold$ of its DEL node, so it is now equal to the height of the root. This changes the interpreted bit of the left child of the root to 0. 
	
	\begin{figure}[tb]
		\centering
		\begin{subfigure}[b]{0.3\textwidth}
			\centering
			\includegraphics[width=\textwidth]{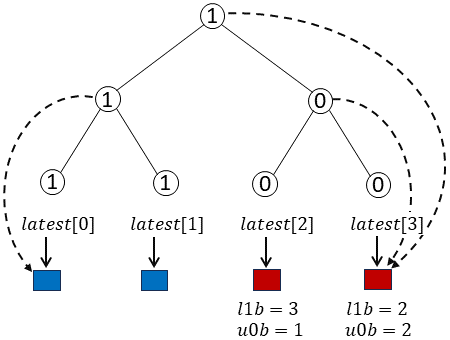}
			\caption{}
			\label{figure:del_execution_a}
		\end{subfigure}
		\hfill
		\begin{subfigure}[b]{0.3\textwidth}
			\centering
			\includegraphics[width=\textwidth]{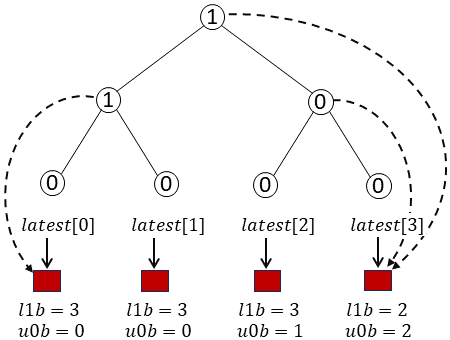}
			\caption{}
			\label{figure:del_execution_b}
		\end{subfigure}
		\hfill
		\begin{subfigure}[b]{0.3\textwidth}
			\centering
			\includegraphics[width=\textwidth]{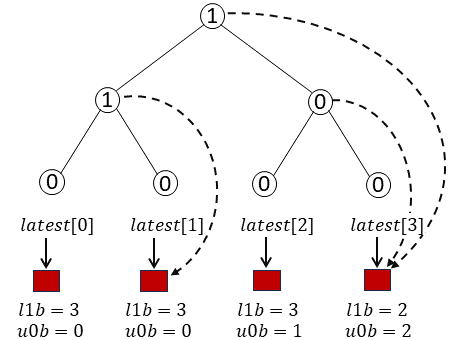}
			\caption{}
			\label{figure:del_execution_c}
		\end{subfigure}
		\hfill
		\begin{subfigure}[b]{0.3\textwidth}
			\centering
			\includegraphics[width=\textwidth]{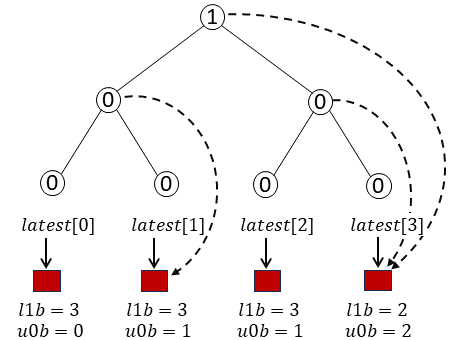}
			\caption{}
			\label{figure:del_execution_d}
		\end{subfigure}
		\hfill
		\begin{subfigure}[b]{0.3\textwidth}
			\centering
			\includegraphics[width=\textwidth]{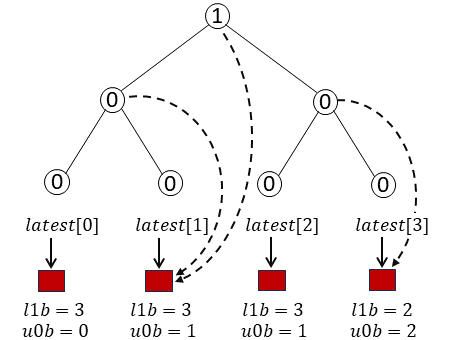}
			\caption{}
			\label{figure:del_execution_e}
		\end{subfigure}
		\hfill
		\begin{subfigure}[b]{0.3\textwidth}
			\centering
			\includegraphics[width=\textwidth]{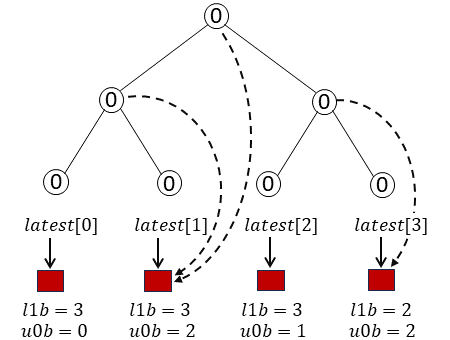}
			\caption{}
			\label{figure:del_execution_f}
		\end{subfigure}	
		\caption{An example of a \textsc{TrieDelete}$(0)$ and \textsc{TrieDelete}$(1)$ updating the interpreted bits of the binary trie. 
		}
		\label{figure:del_execution}
	\end{figure}
	
	\subsubsection{RelaxedPredecessor}
	A \textsc{RelaxedPredecessor}$(y)$ operation, $pOp$, traverses the relaxed binary trie in a manner similar to the sequential algorithm, except that it uses the interpreted bits of binary trie nodes to direct the traversal.
	If $pOp$ completes its traversal of the relaxed binary trie following the sequential algorithm, then it either returns the key of the leaf it reaches, or $-1$ if the traversal ended at the root.
	
	It is possible that $pOp$ is unable to complete a traversal of the relaxed binary trie due to inaccurate interpreted bits. During the downward part of its traversal, it may encounter an internal binary trie node with interpreted bit 1, but both of its children have interpreted bit 0. When this occurs, $pOp$ terminates and returns $\bot$.
	There is a concurrent update operation that needs to update this part of the relaxed binary trie. 
	
	\subsection{Detailed Algorithm Description and Pseudocode}\label{section_relaxed_binary_trie_pseudocode}
	
	In this section, we give a detailed description of the algorithm for each relaxed binary trie operation, and present its pseudocode. 
	
	Figure~\ref{data_records_relaxed_trie} summaries the fields of each update node and binary trie node. 
	Recall that in our high-level description of the binary trie nodes, we noted that each binary trie node $t$ stores a key $x$, which is used to determine the update node that its interpreted bit depends on. 
	In our implementation, $t$ does not store this key directly, but instead stores a pointer, called $\dNodePtr$, to a DEL node with key $x \in U_t$.  
	So the interpreted bit of $t$ depends on the update node pointed to by $\latest[t.\dNodePtr.\key]$.
	The reason we store a pointer to a DEL node with key $x$, instead of the key $x$ itself, is to prevent ABA problems involving out-dated \textsc{TrieDelete}$(x)$ operations incorrectly changing this key.
	
	\begin{figure}[htbp!]
		\begin{algorithmic}[1]
			\alglinenoPop{alg1}
			\State \textbf{Update Node}
			\Indent
			\State $\key$ (Immutable) \Comment{A key in $U$}
			\State $\type$ (Immutable) \Comment{Either INS or DEL}
			\State $\triangleright$ Additional fields when $\type = \text{INS}$
			\State $\target$ (Mutable, initially $\bot$) \Comment{pointer to DEL node}
			
			\State $\triangleright$ Additional fields when $\type = \text{DEL}$
			\State $\stopflag$ (From 0 to 1) \Comment{Boolean value}
			\State $\latestNext$ (Immutable) \Comment{Pointer to INS node}
			\State $\threshold$ (Mutable, initially $0$)\label{ln:init:threshold_relaxed} \Comment{An integer in $\{0,\dots,b\}$}
			\State $\insthreshold$ (Mutable min-register, initially $b+1$) \Comment{An integer in $\{0,\dots,b+1\}$}
			\EndIndent
			\State \textbf{Binary Trie Node}
			\Indent
			\State $\dNodePtr$ (Mutuable, initially points to a dummy DEL node) \Comment{A pointer to a DEL node}
			\EndIndent
			
			
			\alglinenoPush{alg1}
		\end{algorithmic}
		\caption{Summary of the fields and initial values of each node used by the data structure.}
		\label{data_records_relaxed_trie}
	\end{figure}
	
	\subsubsection{TrieSearch and Basic Helper Functions}\label{section_alg_search}
	
	The \textsc{TrieSearch}$(x)$ algorithm finds the update node pointed to by $\latest[x]$. 
	It returns \textsc{True} if this update node has type INS, and \textsc{False} if this update node has type DEL.
	We use the helper function \textsc{FindLatest}$(x)$ to return the update node pointed to by $\latest[x]$. 
	The helper function \textsc{FirstActivated}$(\uNode)$ takes a pointer to an update node, $\uNode$, and checks if $\uNode$ is pointed to by $\latest[\uNode.\key]$.
	For the relaxed binary trie, we consider all update nodes to be \textit{active}.
	The implementation of these helper functions are simple
	the case of the relaxed binary trie, but will be replaced with a different implementation when we consider the lock-free binary trie.
	The helper function \textsc{InterpretedBit}$(t)$ receives a binary trie node $t$, and returns its interpreted bit. 
	We will show it has the property that if the interpreted bit of $t$ is $i \in \{0,1\}$ throughout the instance of \textsc{InterpretedBit}, then it will return $i$.

	

	
	\begin{figure}[h!]
		\begin{algorithmic}[1]
			\alglinenoPop{alg1}
			\State \textbf{Algorithm} \textsc{FindLatest}$(x)$
			\Indent
			\State \Return $\ell \leftarrow \latest[x]$
			\EndIndent
			\alglinenoPush{alg1}
		\end{algorithmic}
	\end{figure}
	
	\begin{figure}[h!]
		\begin{algorithmic}[1]
			\alglinenoPop{alg1}
			\State \textbf{Algorithm} \textsc{TrieSearch}$(x)$
			\Indent 
			\State $\uNode \leftarrow \textsc{FindLatest}(x)$\label{ln:search:findLatest_relaxed}
			\If {$\uNode.\type = \textsc{INS}$} \Return \textsc{True}
			\Else ~\Return \textsc{False}
			\EndIf 
			\EndIndent
			\alglinenoPush{alg1}
		\end{algorithmic}
	\end{figure}
	
	\begin{figure}[h!]
		\begin{algorithmic}[1]
			\alglinenoPop{alg1}
			\State \textbf{Algorithm} \textsc{FirstActivated}$(\uNode)$
			\Indent 
			\State $\uNode' \leftarrow \latest[\uNode.\key]$ \label{ln:firstactivated:head_relaxed}
			\State \Return $\uNode = \uNode'$
			\EndIndent
			\alglinenoPush{alg1}
		\end{algorithmic}
	\end{figure}
	
	\begin{figure}[h!]
		\begin{algorithmic}[1]
			\alglinenoPop{alg1}
			\State \textbf{Algorithm} \textsc{InterpretedBit}$(t)$
			\Indent
			\State $\uNode \leftarrow \textsc{FindLatest}(t.\dNodePtr.\key)$\label{ln:ib:findlatest}
			
			\If {$\uNode.\type = \text{INS}$ }\label{ln:ib:ins}
			\Return $1$
			\EndIf
			
			
			
			
			\If {$t.\height \leq \uNode.\threshold$} \label{ln:ib:u0b}
			\If {$t.\height < \uNode.\insthreshold$ and \textsc{FirstActivated}$(\uNode)$} \Return $0$\label{ln:ib:l1b}
			\EndIf
			\EndIf
			\State \Return $1$
			
			\EndIndent
			
			\alglinenoPush{alg1}
		\end{algorithmic}
	\end{figure}
	
	\subsubsection{TrieInsert}
	
	A \textsc{TrieInsert}$(x)$ operation $iOp$ begins by reading the update node, $\dNode$, pointed to by $\latest[x]$. If $\dNode$ is not a DEL node, then $x$ is already in $S$ so  \textsc{TrieInsert}$(x)$ returns (on line~\ref{ln2:insert:return_early}). Otherwise $iOp$ creates a new INS node, denoted $\iNode$ with key $x$.
	The purpose of line~\ref{ln2:insert:helpStop} will be discussed later when describing \textsc{TrieDelete}.
	It then 
	attempts to change  $\latest[x]$ to point to $\iNode$ using CAS (on line~\ref{ln2:insert:cas_latest_head}). A \textsc{TrieInsert}$(x)$ operation that successfully performs this CAS adds $x$ to $S$, and is linearized at this successful CAS.
	If multiple \textsc{TrieInsert}$(x)$ operations concurrently attempt to change $\latest[x]$, exactly one will succeed. Any \textsc{TrieInsert}$(x)$ operations that perform an unsuccessful CAS can return, because some other \textsc{TrieInsert}$(x)$ operation successfully added $x$ to $S$.
	Note that by updating $\latest[x]$ to point to $\iNode$, the interpreted bit of the leaf with key $x$ is 1.
	
	\begin{figure}[h!]
		\begin{algorithmic}[1]
			\alglinenoPop{alg1}
			\State \textbf{Algorithm} \textsc{TrieInsert}$(x)$
			\Indent
			\State $\dNode \leftarrow \textsc{FindLatest}(x)$\label{ln2:insert:find_latest}
			\If {$\dNode.\type \neq \text{DEL}$} \Return \Comment{$x$ is already in $S$}\label{ln2:insert:return_early}
			\EndIf
			
			\State Let $\iNode$ be a pointer to a new update node: 
			\State \quad $\iNode.\key \leftarrow x$
			\State \quad $\iNode.\type \leftarrow \text{INS}$
			\State $\dNode.\latestNext.\target.\stopflag \leftarrow \textsc{True}$\label{ln2:insert:helpStop}\Comment{Ignore if any field reads $\bot$}
			\If {CAS$(\latest[x], \dNode, \iNode) = \textsc{False}$} \Return \label{ln2:insert:cas_latest_head}\Comment{$x$ added to $S$ if successful}
			\EndIf
			\State \textsc{InsertBinaryTrie}$(\iNode)$\label{ln2:insert:binaryTrie}
			\State \Return
			\EndIndent
			\alglinenoPush{alg1}
		\end{algorithmic}
	\end{figure}
	
	The purpose of \textsc{InsertBinaryTrie}$(\iNode)$ is to set the interpreted bit of each binary trie node $t$ on the path from the parent of the leaf with key $x$ to the root to 1. On line~\ref{ln:insert:findLatest},  $iOp$ determines the update node, $\uNode$, that $t$ depends on.
	If $\uNode$ is an INS node, then the interpreted bit of $t$ is 1 and $iOp$ proceeds to the parent of $t$.
	So $\uNode$ is a DEL node. The check on line~\ref{ln:insert:IfThreshold} indicates that the creator of $\uNode$ is currently trying to update the interpreted bit of $t$ to 0 (when $t.\dNodePtr = \uNode$), or the creator of $\uNode$ has previously set the interpreted bit of $t$ to 0 (when $t.\height \leq \uNode.\threshold$). 
	Then $iOp$ sets  $\iNode.\target$ to point to $\uNode$ (on line~\ref{ln:insert:setTarget}), indicating it is attempting to perform a minwrite to $\uNode.\insthreshold$.
	Note that $\iNode.\target$ always points to a DEL node, but not necessarily to one with the same key as $\iNode$.
	Provided $iOp$ is still the latest update operation with key $x$, $iOp$ updates
	$\uNode.\insthreshold$ to the value $t.\height$ using a \textsc{MinWrite} on line~\ref{ln:insert:minwrite}. 
	Updating $\uNode.\insthreshold$  serves the purpose of changing the interpreted bit of $t$ to 1, as well as informing the \textsc{TrieDelete} operation that created $\uNode$ to stop updating the interpreted bits of the binary trie. This \textsc{TrieDelete} operation no longer needs to continue checking if the interpreted bits of ancestors of $t$ need to be changed to 0, because $iOp$ is going to change them to 1.
	

	
	\begin{figure}[h!]
		\begin{algorithmic}[1]
			\alglinenoPop{alg1}
			\State \textbf{Algorithm} \textsc{InsertBinaryTrie}$(\iNode)$
			\Indent
			\For { each binary trie node $t$ on path from parent of the leaf $\iNode.\key$ to root } \label{ln:insert:for_loop}
			\State $\uNode \leftarrow \textsc{FindLatest}(t.\dNodePtr.\key)$ \label{ln:insert:findLatest}
			\If {$\uNode.type = \text{DEL}$} \label{ln:insert:delCheck}
			\If {$t.\dNodePtr  = \uNode$ or $t.\height \leq \uNode.\threshold$ }\label{ln:insert:IfThreshold}
			\State $\iNode.\target \leftarrow \uNode$ \label{ln:insert:setTarget}
			\If {\textsc{FirstActivated}$(\iNode) = \textsc{False}$} \Return \label{ln:insert:firstactivated2}
			\EndIf 
			\If {$t.\height < \uNode.\insthreshold$}\label{ln:insert:IfMinwrite}
			\State $\textsc{MinWrite}(\uNode.\insthreshold, t.\height)$ \label{ln:insert:minwrite}
			\EndIf
			\EndIf
			\EndIf			
			\EndFor
			\EndIndent
			\alglinenoPush{alg1}
		\end{algorithmic}
	\end{figure}
	
	\subsubsection{TrieDelete}
	
	A \textsc{TrieDelete}$(x)$ operation, $dOp$, checks if the update node, $\iNode$, pointed to by $\latest[x]$ is an INS node. If not, it returns because $x$ is not in $S$ (on line~\ref{ln2:delete:return_early}).
	Otherwise, it creates a new DEL node, $\dNode$. 
	It then updates $\latest[x]$ to point to $\dNode$ using CAS in the same way as \textsc{TrieInsert}$(x)$. A \textsc{TrieDelete}$(x)$ that successfully performs this CAS removes $x$ from $S$, and is linearized at this successful CAS.
	On line~\ref{ln2:delete:setStop}, $dOp$ sets $\iNode.\target.\stopflag$ to $\textsc{True}$. This is important when the creator of $\iNode$ has set $\iNode.\target$ on line~\ref{ln:insert:setTarget}, but not yet performed a minwrite to  $\iNode.\target.\insthreshold$ on line~\ref{ln:insert:minwrite}. It serves to inform the \textsc{TrieDelete} operation that created the DEL node pointed to by $\iNode.\target$ to stop trying to update the interpreted bits of binary trie nodes to 0. 
	
	\begin{figure}[h!]
		\begin{algorithmic}[1]
			\alglinenoPop{alg1}
			\State \textbf{Algorithm} \textsc{TrieDelete}$(x)$
			\Indent
			\State $\iNode \leftarrow \textsc{FindLatest}(x)$\label{ln2:delete:findLatest}
			\If {$\iNode.\type \neq \text{INS}$} \Return \Comment{$x$ is not in $S$}\label{ln2:delete:return_early}
			\EndIf
			
			\State Let $\dNode$ be a pointer to a new update node:
			\State \quad $\dNode.\key \leftarrow x$
			\State \quad $\dNode.\type \leftarrow \text{DEL}$
			\State \quad $\dNode.\latestNext \leftarrow \iNode$
			
			\If {CAS$(\latest[x], \iNode, \dNode) = \textsc{False}$} \Return \label{ln2:delete:cas_latest_head}\Comment{$x$ removed from $S$ if successful}
			\EndIf
			
			\State $\iNode.\target.\stopflag \leftarrow \textsc{True}$\label{ln2:delete:setStop}
			
			\State \textsc{DeleteBinaryTrie}$(\dNode)$\label{ln2:delete:binaryTrie}
			\State \Return
			\EndIndent
			\alglinenoPush{alg1}
		\end{algorithmic}
	\end{figure}
	
	Next, $dOp$ calls \textsc{DeleteBinaryTrie}$(\dNode)$ to update the interpreted bits of the relaxed binary trie nodes from the parent of the leaf with key $x$ to the root. 
	Let $t$ be an internal binary trie node on the path from the leaf with key $x$ to the root. Suppose $dOp$ successfully changed the interpreted bit of one of $t$'s children to 0. 
	On line~\ref{ln:delete:checkIB1}, $dOp$ checks that both of $t$'s children have interpreted bit 0.
	If so, $dOp$ attempts to change the interpreted bit of $t$ to 0.
	Recall that $t$ depends on the update node pointed to by $\latest[t.\dNodePtr.\key]$. To change $t$ to depend on $\dNode$, $dOp$ must perform a successful CAS to change $t.\dNodePtr$ to point to $\dNode$. 
	Two attempts are made by $dOp$ to perform this CAS (on line~\ref{ln:delete:trieCAS} and \ref{ln:delete:trieCAS2}).
	Before each attempt, $dOp$ checks that $\dNode$ is pointed to by $\latest[x]$ (on line~\ref{ln:delete:firstactivated1} or \ref{ln:delete:firstactivated2}) and that $\dNode.\stopflag = \textsc{True}$ or $\dNode.\insthreshold \neq b+1$ (on line~\ref{ln:delete:readstop1} or \ref{ln:delete:readstop2}).
	Recall that when $\dNode.\stopflag = \textsc{True}$ or $\dNode.\insthreshold \neq b+1$,
	$dOp$ should stop updating the interpreted bits of the binary trie to 0 because there is a \textsc{TrieInsert} operation trying to set the same interpreted bits to 1.
	Two CAS attempts are performed to prevent out-dated \textsc{TrieDelete} operations that were poised to perform CAS from conflicting with latest \textsc{TrieDelete} operations.
	If $dOp$ is unsuccessful in both its CAS attempts, it can stop updating the binary trie because some concurrent \textsc{TrieDelete}$(x')$ operation, with key $x' \in U_t$, successfully changed $t.\dNodePtr$ to point to its own DEL node.
	Otherwise $dOp$ is successful in changing the interpreted bit of $t$ to depend on $\dNode$. 
	Immediately after $dOp$'s successful CAS, the interpreted bit of $t$ is still 1 (because $\dNode.\threshold$ has not yet been incremented to $t.\height$). Once again, $dOp$ verifies both children of $t$ have interpreted bit 0 (on line~\ref{ln:delete:checkIB2}), otherwise it returns.
	To change the interpreted bit of $\dNode$ to 0, $dOp$ writes $t.\height$ into $\dNode.\threshold$ (on line~\ref{ln:delete:threshold}), which increments its value.
	This indicates that all binary trie nodes at height $t$ and below that depend on $\dNode$ have interpreted bit 0. 
	
	
	\begin{figure}[h!]
		\begin{algorithmic}[1]
			\alglinenoPop{alg1}
			\State \textbf{Algorithm} \textsc{DeleteBinaryTrie}$(\dNode)$
			\Indent
			
			\State $t \leftarrow $ leaf of binary trie with key $\dNode.key$
			\While {$t$ is not the root of the binary trie}
			
			\If { $\textsc{InterpretedBit}(t.\mathit{sibling}) = 1$ or $\textsc{InterpretedBit}(t) = 1$} \Return \label{ln:delete:checkIB1}
			\EndIf
			
			\State $t \leftarrow t.\mathit{parent}$
			\State $\mathit{d} \leftarrow t.\dNodePtr$\label{ln:delete:readPtr}
			
			\If  {\textsc{FirstActivated}$(\dNode)$ = \textsc{False}} \Return \label{ln:delete:firstactivated1}
			\EndIf
			
			\If  {$\dNode.\stopflag = \textsc{True}$ or $\dNode.\insthreshold \neq b+1$} \Return \label{ln:delete:readstop1}
			\EndIf
			
			\If  { CAS$(t.\dNodePtr, d, \dNode) = \textsc{False}$} \label{ln:delete:trieCAS}  \Comment{Try second attempt if unsuccessful}
			
			\State  $\mathit{d} \leftarrow t.\dNodePtr$ \label{ln:delete:readPtr2}
			
			\If  {\textsc{FirstActivated}$(\dNode)$ = \textsc{False}} \Return \label{ln:delete:firstactivated2}
			\EndIf
			\If  {$\dNode.\stopflag = \textsc{True}$ or $\dNode.\insthreshold \neq b+1$} \Return \label{ln:delete:readstop2}
			\EndIf
			
			\If {CAS$(t.\dNodePtr, d, \dNode) = \textsc{False}$} \Return \label{ln:delete:trieCAS2}
			\EndIf
			\EndIf
			\If {$\textsc{InterpretedBit}(t.\mathit{left}) = 1$ or  $\textsc{InterpretedBit}(t.\mathit{right}) = 1$} \Return \label{ln:delete:checkIB2}
			\EndIf
			\State $\dNode.\threshold \leftarrow t.\height$ \label{ln:delete:threshold}
			\EndWhile
			\EndIndent
			
			\alglinenoPush{alg1}
		\end{algorithmic}
	\end{figure}
	
	\subsubsection{RelaxedPredecessor}
	
	A \textsc{RelaxedPredecessor}$(y)$ operation, $pOp$, begins by traversing up the relaxed binary trie starting from the leaf with key $y$ towards the root (during the while-loop on line\ref{ln:traverseTrie:whileUp}). If the left child of each node on this path either has interpreted bit 0 or is also on this path, then $pOp$ returns $-1$ (on line~\ref{ln:traverseTrie:return_minus}).
	Otherwise, consider the first node on this path whose left child $t$ (set on line~\ref{ln:traverseTrie:downstart}) has interpreted bit 1 and is not on this path. Starting from $t$, $pOp$ traverses the right-most path of binary trie nodes with interpreted bit 1 (during the while-loop on line~\ref{ln:traverseTrie:whiledown}). If a binary trie node $t$ is encountered where both its children have interpreted bit 0, then $\bot$ is returned (on line~\ref{ln:traverseTrie:return_minus}). Otherwise, the $pOp$ reaches a leaf, and returns its key (on line~\ref{ln:traverseTrie:return_leaf}).
	
	\begin{figure}[h!]
		\begin{algorithmic}[1]
			\alglinenoPop{alg1}		
			\State \textbf{Algorithm} \textsc{RelaxedPredecessor}$(y)$
			\Indent
			\State $t \leftarrow$ the binary trie node represented by $\PRE_b[y]$
			\While {$t$ is the left child of $t.\mathit{parent}$ or $\textsc{InterpretedBit}(t.\mathit{sibling}) = 0$} \label{ln:traverseTrie:whileUp}
			\State $t \leftarrow t.parent$
			\If {$t$ is the root}
			\State \Return $-1$\label{ln:traverseTrie:return_minus}
			\EndIf 
			\EndWhile
			
			\State $\triangleright$ Traverse right-most path of nodes with interpreted bit 1 from $t.\mathit{parent}.\mathit{left}$
			\State $t \leftarrow t.\mathit{parent}.\mathit{left}$\label{ln:traverseTrie:downstart}
			\While {$t.\mathit{height} > 0$ }\label{ln:traverseTrie:whiledown}
			\If {$\textsc{InterpretedBit}(t.\mathit{right}) = 1$}\label{ln:traverseTrie:go_right}
			\State $t \leftarrow t.\mathit{right}$
			\ElsIf {$\textsc{InterpretedBit}(t.\mathit{left}) = 1$}\label{ln:traverseTrie:go_left}
			\State $t \leftarrow t.\mathit{left}$
			\Else 
			\State $\triangleright$ both children of $t$ have interpreted bit 0
			\State \Return $\bot$\label{ln:traverseTrie:return_t}
			\EndIf
			
			\EndWhile
			\State $\triangleright$ $t$ is a leaf node with key $t.key$
			\State \Return $t.\mathit{key}$\label{ln:traverseTrie:return_leaf}
			\EndIndent
			
			\alglinenoPush{alg1}
		\end{algorithmic} 
	\end{figure}

	\subsection{Proof of Correctness}\label{section_relaxed_binary_trie_correctness}
	
	In this section, we prove that our implementation satisfies the specification of the relaxed binary trie. 
	In Section~\ref{section_ins_del_search_lin_relaxed}, we formally define the linearization points for \textsc{TrieSearch}, \textsc{TrieInsert}, and \textsc{TrieDelete}, and prove that the implementation is strongly linearizable with respect to these operation types.
	In Section~\ref{section_interpreted_bits} we prove properties satisfied by the interpreted bits of the binary trie. 
	In Section~\ref{section_predecessor_linearization}, we prove that \textsc{RelaxedPredecessor} operations satisfies its specification.

	\subsubsection{Strong Linearizability of TrieInsert, TrieDelete, and TrieSearch}\label{section_ins_del_search_lin_relaxed}
	
	We let the $\textit{latest}$ update operation with key $x$ in a configuration $C$ be the last $S$-modifying update operation with key $x$ that was linearized prior to $C$. Let the $\textit{latest}$ update node with key $x$ be the update node created by this operation. For the relaxed binary trie, it is the update node pointed to by $\latest[x]$.

	A \textsc{TrieSearch}$(x)$ operation is linearized immediately after it reads $\latest[x]$. We argue that this operation returns \textsc{True} if and only if $x \in S$ in this configuration. 
	
	\begin{lemma}\label{lemma:search:true}
		Let $op$ be a \textsc{TrieSearch}$(x)$ operation. Then $op$ returns \textsc{True} if and only if $x \in S$ in the configuration $C$ immediately after $op$ reads $\latest[x]$.
	\end{lemma}
	
	\begin{proof}
		Let $\uNode$ be the update node pointed to by $\latest[x]$ read by $op$, and let $C$ be the configuration immediately after this read by $op$.
		If $op$ returned \textsc{True}, it read that $\uNode.\type = \textsc{INS}$. The type of an update node is immutable, so  $\latest[x]$ points to an \textsc{INS} node in $C$. So it follows by definition that $x \in S$ in $C$.
		If $op$ returned \textsc{False}, it read that $\uNode.\type = \textsc{DEL}$. It follows by definition that $x \notin S$ in $C$.
	\end{proof}
	
	
	

	The  \textsc{TrieInsert}$(x)$ and \textsc{TrieDelete}$(x)$ operations that are $S$-modifying successfully change $\latest[x]$ to point to their own update node using CAS, and are linearized at this CAS step.
	A \textsc{TrieInsert}$(x)$ operation that is not $S$-modifying does not update $\latest[x]$ to point to its own update node. This happens when it reads that $\latest[x]$ points to a INS node, or when it performs an unsuccessful CAS.
	In the following two lemmas, we prove that for each of these two cases, there is a configuration during the \textsc{TrieInsert}$(x)$ in which $x \in S$, and hence does not need to add $x$ to $S$. 
	
	\begin{lemma}\label{lemma:insert:return_early}
	If $uOp$ is a \textsc{TrieInsert}$(x)$ operation that returns on line~\ref{ln2:insert:return_early},  
	then in the configuration $C$ immediately after $uOp$ reads $\latest[x]$,  $x \in S$.
	\end{lemma}
	
	\begin{proof}
	Suppose $uOp$ is a \textsc{TrieInsert}$(x)$ operation.
	Let $\uNode$ be the pointed to by $\latest[x]$ that is read by $uOp$.
	Since $uOp$ returned on line~\ref{ln2:insert:return_early}, it saw that $\uNode.\type = \text{INS}$. By definition, $x \in S$ in the configuration $C$ immediately after this read.
	\end{proof}
	
	\begin{lemma}\label{lemma:insert:return_early2}
	If $uOp$ is a \textsc{TrieInsert}$(x)$ operation that returns on line~\ref{ln2:insert:cas_latest_head}, then there is a configuration during $uOp$ in which $x \in S$.
	\end{lemma}
	
	\begin{proof}
	
	Since $uOp$ does not return on line~\ref{ln2:insert:return_early}, it read that $\latest[x]$ points to a DEL node.
	From the code, $\latest[x]$ can only change from pointing to this DEL node to an INS node by a successful CAS of some \textsc{TrieInsert}$(x)$ operation.
	Since $uOp$ performs an unsuccessful CAS on line~\ref{ln2:insert:cas_latest_head}, some other \textsc{TrieInsert}$(x)$ changed $\latest[x]$ to point to an INS node using a successful CAS sometime between $uOp$'s read of $\latest[x]$ and $uOp$'s unsuccessful CAS on $\latest[x]$. In the configuration immediately after this successful CAS,  $x \in S$.
	\end{proof}
	
	The next two lemmas are for \textsc{TrieDelete} operations. Their proofs are symmetric to \textsc{TrieInsert}.
	
	\begin{lemma}
	If $uOp$ is a \textsc{TrieDelete}$(x)$ operation that returns on line~\ref{ln2:delete:cas_latest_head}, then there is a configuration during $uOp$ in which $x \notin S$.
	\end{lemma}
	
	\begin{lemma}
	If $uOp$ is a \textsc{TrieDelete}$(x)$ returns on line~\ref{ln2:delete:return_early}, then in the configuration $C$ immediately after $uOp$ reads $\latest[x]$, $x \notin S$
	\end{lemma}
	
	Finally, we prove that the relaxed binary trie is strongly linearizable with respect to \textsc{TrieInsert}, \textsc{TrieDelete}, and \textsc{TrieSearch} operations.
	
	\begin{lemma}
	The implementation of the relaxed binary trie is strongly linearizable with respect to \textsc{TrieInsert}, \textsc{TrieDelete}, and \textsc{TrieSearch} operations.
	\end{lemma}
	
	\begin{proof}
	Consider an execution $\alpha$ of the relaxed binary trie. Let $\alpha'$ be any prefix of $\alpha$. 
	Suppose that all \textsc{TrieInsert} and \textsc{TrieDelete} operations have been linearized in $\alpha$.
	Suppose $op$ is a \textsc{TrieInsert}$(x)$ operation. Suppose $op$ reads that $\latest[x]$ points to a INS node (on line~\ref{ln2:insert:find_latest}) during $\alpha'$ and, hence, returns without changing $\latest[x]$. 
	By Lemma~\ref{lemma:insert:return_early}, $x \in S$ in the configuration $C$ immediately after this read.
	If $\alpha'$ contains $C$, then $op$ is linearized at $C$ for both $\alpha$ and $\alpha'$. 
	
	So $op$ reads that $\latest[x]$ points to a DEL node, so $x \notin S$ immediately after this step.
	If $op$ successfully changes $\latest[x]$ to point to its own update node using CAS on line~\ref{ln2:insert:cas_latest_head}, then $op$ is linearized at this CAS step. 
	If $\alpha'$ contains this CAS step, then 
	$op$ is linearized at this CAS step for both $\alpha$ and $\alpha'$. 
	
	If $op$ does not successfully change $\latest[x]$ to point to its own update node using CAS on line~\ref{ln2:insert:cas_latest_head}, then by Lemma~\ref{lemma:insert:return_early2}, 
	there is a configuration sometime during $op$ in which $x \in S$.
	Consider the earliest such configuration $C$. Note that $C$ immediately follows the successful CAS of some \textsc{TrieInsert}$(x)$ operation by another processes. 
	By assumption, the update nodes allocated by different processes are distinct.
	Hence, in any continuation of the execution that contains this successful CAS, $op$ performs an unsuccessful CAS on line~\ref{ln2:insert:cas_latest_head}.
	If $\alpha'$ contains this successful CAS (and hence $C$), then  
	$op$ is linearized at $C$ for both $\alpha$ and $\alpha'$. 
	
	The case when $op$ is a \textsc{TrieDelete}$(x)$ operation is similar.

	Now consider any arbitrary \textsc{TrieSearch} operation, $op$.
	Let $C$ be the configuration immediately after $op$ reads $\latest[x]$. 
	By Lemma~\ref{lemma:search:true}, $op$ returns \textsc{True} if and only if $x \in S$ in $C$.
	The last $S$-modifying update operation linearized before $C$ is a \textsc{TrieInsert}$(x)$ operation if and only if $x \in S$ in $C$, so the return value of $op$ is consistent with the sequential execution in which operations are performed sequentially in order of their linearization points.
	If $\alpha'$ contains $C$, then $op$ is linearized at $C$ for both $\alpha$ and $\alpha'$. 
	\end{proof}

	\subsubsection{Properties of the Interpreted Bits}\label{section_interpreted_bits}
	
	In this section, we prove that properties IB0 and IB1 of the interpreted bits are satisfied by our implementation.
	

	We first observe that \textsc{FindLatest}$(x)$ correctly returns the update node of the latest update operation with key $x$. In the implementation of the relaxed binary trie, this is done by a single read to $\latest[x]$. 
	
	\begin{observation}\label{lemma:findLatestRelaxed}
	Let $\tau$ be a completed instance of \textsc{FindLatest}$(x)$ that returns $\uNode$.
	Then there is a configuration during $\tau$ in which the update operation that created $\uNode$ is the latest update operation with key $x$.
	\end{observation}
	
	The next two observations state that \textsc{FirstActivated}$(\uNode)$ correctly determines if $\uNode$ is pointed to by $\latest[x]$. In particular, it holds in the configuration immediately after $\latest[x]$ is read and compared with $\uNode$. 
	
	\begin{observation}\label{lemma:firstActivatedRelaxed_true}
	Let $\tau$ be a completed instance of \textsc{FirstActivated}$(\uNode)$, for some update node, $\uNode$. If $\tau$ returns \textsc{True}, there is a configuration during $\tau$ in which
	the update operation that created $\uNode$ is the latest update operation with key $x$.
	\end{observation}
	
	\begin{observation}\label{lemma:firstActivatedRelaxed_false}
	Let $\tau$ be a completed instance of \textsc{FirstActivated}$(\uNode)$, for some update node, $\uNode$. If $\tau$ returns \textsc{False}, there is a configuration during $\tau$ in which
	the update operation that created $\uNode$ is not the latest update operation with key $x$.
	\end{observation}

	
	
	We next state basic observations about how $t.\dNodePtr$ changes, for any binary trie node $t$. They follow from the code of \textsc{DeleteBinaryTrie}.
	
	\begin{observation}\label{obs:dNodeOwner}
	Only \textsc{TrieDelete} operations change the $\dNodePtr$ field of a binary trie node, 
	and they only change it
	to point to their own DEL node (on line~\ref{ln:delete:trieCAS} or \ref{ln:delete:trieCAS2}).
	\end{observation}
	
	\begin{observation}\label{obs:aba}
	For any binary trie node $t$, suppose $t.\dNodePtr$ is changed to point from a DEL node, $\dNode$, to a different DEL node. Then in any future configuration, $t.\dNodePtr$ does not point to $\dNode$.
	\end{observation}
	
	\begin{observation}\label{obs:dNodeThreshold}
	Let $\dNode$ be the DEL node created by a  \textsc{TrieDelete} operation, $dOp$. Only $dOp$ writes to $\dNode.\threshold$ and it only does so on line~\ref{ln:delete:threshold}. 
	\end{observation}
	
	
	The next lemma follows from the fact that \textsc{TrieDelete} operations traverse up the binary trie, updating the $\dNodePtr$ field of each node on its path to point to the DEL node it created.  As \textsc{TrieDelete} operations traverse up the trie during \textsc{DeleteBinaryTrie}, they increment the $\threshold$ field of the DEL node they created.
	
	\begin{lemma}\label{lemma:sequenceCAS}
	Let $\dNode$ be a DEL node created by a  \textsc{TrieDelete}$(x)$ operation, $dOp$.
	Suppose $dOp$ changes $t.\dNodePtr$ of some binary trie node $t$ to point to $\dNode$ due to a successful CAS, $s$ (performed on line~\ref{ln:delete:trieCAS} or \ref{ln:delete:trieCAS2}).	
	Before $s$, $dOp$ has performed a sequence of successful CASs updating each binary trie node from the parent of the leaf with key $x$ to a child of $t$ to point to $\dNode$. Immediately after $s$, $\dNode.\threshold = t.\height-1$.
	\end{lemma}
	
	\begin{proof}
	Let $C$ be the configuration immediately after the CAS, $s$, that changes $t.\dNodePtr$ to point to $\dNode$. Let $h = t.\height$. Let $t_0,\dots,t_h$ be the sequence of nodes on the path from the leaf $t_0$ with key $x$ to $t_h = t$. 
	
	Suppose that $t_h$ is the parent of the leaf with key $x$, so $t.\height = 1$. Recall that the initial value of $\dNode.\threshold$ is 0 and only $dOp$ increments it. In the first iteration of \textsc{DeleteBinaryTrie}, $dOp$ does not change  $\dNode.\threshold$ prior to performing $s$. Therefore, in $C$, $\dNode.\threshold = 0$.
	
	So suppose that $t_h$ is not the parent of the leaf with key $x$.
	From the code of \textsc{DeleteBinaryTrie}, $dOp$ does not proceed to its next iteration unless it performs at least one successful CAS on either  line~\ref{ln:delete:trieCAS} or \ref{ln:delete:trieCAS2}.
	Hence, prior to $C$, $dOp$ has performed a sequence of successful CASs, updating the $t_i.\dNodePtr$ to point to $\dNode$, for $1 \leq i \leq h-1$. 
	
	During the iteration \textsc{DeleteBinaryTrie} prior to the iteration in which $dOp$ performs a successful CAS on $t_h$, 
	$dOp$ increased $\dNode.\threshold$ to $t_{h-1}.\height = t.\height-1$ on line~\ref{ln:delete:threshold}. By Observation~\ref{obs:dNodeThreshold}, no operation besides $dOp$ changes $\dNode.\threshold$ between this write and $C$.
	Hence, in $C$, $\dNode.\threshold = t.\height-1$.
	
	\end{proof}
	
	
	
	
	The next lemma states that once  a \textsc{TrieDelete} operation successfully changes the $\dNodePtr$ field of some binary trie node, \textsc{TrieDelete} operations linearized earlier can no longer change this $\dNodePtr$ field.
	
	\begin{lemma}\label{lemma:staleCASFail}
	
	Suppose $\dNode$ and $\dNode'$ are the DEL nodes created by two different, $S$-modifying \textsc{TrieDelete}$(x)$ operations, $dOp$ and $dOp'$, respectively. 
	Suppose that $dOp$ is linearized before $dOp'$. 
	Once $dOp'$ performs a successful CAS on the $\dNodePtr$ of some binary trie node, $dOp$ does not later perform a successful CAS on the $\dNodePtr$ of that same binary trie node.
	
	\end{lemma}
	
	\begin{proof}
	Let $C'$ be the configuration immediately after $dOp'$'s first successful CAS on $t.\dNodePtr$ (on line~\ref{ln:delete:trieCAS} or \ref{ln:delete:trieCAS2}) of some binary trie node, $t$. From the code, this only occurs after  $\latest[x]$ is updated to point to $\dNode'$. Then in $C'$,  $\latest[x]$ no longer points to $\dNode$.
	
	
	Suppose, for contradiction, that $dOp$ performs a CAS on $t.\dNodePtr$ in some configuration after $C'$. Before this CAS, \textsc{FirstActivated}$(\dNode)$ returned \textsc{True} on line~\ref{ln:delete:firstactivated1} or \ref{ln:delete:firstactivated2}. By Observation~\ref{lemma:firstActivatedRelaxed_true}, $dOp$ read that $\latest[x]$ points to $\dNode$.
	
	
	This read must occur before $C'$, since $\latest[x]$ no longer points to $\dNode$ in $C'$.
	Prior to this read, $dOp'$ reads $t.\dNodePtr$ (on line~\ref{ln:delete:readPtr} or \ref{ln:delete:readPtr2}).
	Hence, $dOp$'s last read of $t.\dNodePtr$ also occurred before $C'$.
	Then $dOp'$ changes the value of $t.\dNodePtr$ to a value different than what was last read by $dOp$. So the CAS that $dOp$ performs on $t.\dNodePtr$ will be unsuccessful. 
	\end{proof}
	
	
	\begin{lemma}\label{lemma:deleteUpdateTrieVal}
	Let $\dNode$ be the DEL node created by a \textsc{TrieDelete}$(x)$ operation, $dOp$.
	Suppose $dOp$ changes $t.\dNodePtr$ of some binary trie node $t$ to point to $\dNode$ (due to a successful CAS on line~\ref{ln:delete:trieCAS} or \ref{ln:delete:trieCAS2}). In the configuration immediately after this CAS,
	the interpreted bit of $t$ is 1.
	\end{lemma}
	
	\begin{proof}
	Let $cas$ be the successful CAS operation where $dOp$ updates $t.\dNodePtr$ to point to $\dNode$. Let $C$ be the configuration immediately after this CAS.
	
	Suppose, for contradiction, that the interpreted bit of $t$ is 0 in $C$.
	The interpreted bit of $t$ depends on the update node pointed to by $\latest[x]$. This must be a DEL node, otherwise the interpreted bit of $t$ is 1.
	
	Let $\dNode'$ be the update node pointed to by $\latest[x]$ in $C$.
	The interpreted bit of $t$ is 0 when $t.\height < \dNode'.\insthreshold$ and $t.\height \leq \dNode'.\threshold$. By Lemma~\ref{lemma:sequenceCAS}, $\dNode.\threshold = t.\height-1$, so $\dNode' \neq \dNode$. So $\dNode'$ is a DEL node created by a \textsc{TrieDelete}$(x)$ operation, $dOp'$, where $dOp' \neq dOp$. 
	Furthermore, $dOp'$ is the latest update operation with key $x$ in $C$, so it was linearized after $dOp$.
	From the code $dOp'$ does not increment $\dNode'.\threshold$ to equal $t.\height$ until after it performs a successful CAS on $t.\dNodePtr$ (on line~\ref{ln:delete:trieCAS} or \ref{ln:delete:trieCAS2}).
	By Lemma~\ref{lemma:staleCASFail}, this CAS must be unsuccessful, a contradiction. 
	
	
	
	\end{proof}

	The next two lemmas state that the helper function \textsc{InterpretedBit}$(t)$ returns a correct value for the interpreted bit of a binary trie node $t$.
	
	\begin{lemma}\label{lemma:IB_return0}
	Consider an instance $\tau$ of \textsc{InterpretedBit}$(t)$ for some binary trie node $t$. If the interpreted bit of $t$ is 0 throughout $\tau$, then $\tau$ returns 0.
	\end{lemma}
	
	\begin{proof}
	Let $\uNode$ be the update node returned by \textsc{FindLatest}$(t.\dNodePtr.\key)$ on line~\ref{ln:ib:findlatest}. Since the interpreted bit of $t$ is 0 throughout $\tau$, $\uNode$ must be a DEL node. So $\tau$ does not return 1 on line~\ref{ln:ib:ins}.
	
	Since the interpreted bit of $t$ is 0 throughout $\tau$, it follows from Lemma~\ref{lemma:deleteUpdateTrieVal} that $t.\dNodePtr$ does not change to point to any other DEL node. Furthermore, 
	if $\latest[t.\dNodePtr.\key]$ changes to point to an update node other than $\uNode$, it will be changed to point to an INS node and cause  $t$ to have interpreted bit 1 by the definition of interpreted bits. So $\latest[t.\dNodePtr.\key]$  points to $\uNode$ throughout $\tau$.
	So the interpreted bit of $t$ is determined by $\uNode.\insthreshold$ and $\uNode.\threshold$ throughout $\tau$. Additionally, by Observation~\ref{lemma:firstActivatedRelaxed_false}, any calls of \textsc{FirstActivated}$(\uNode)$ during $\tau$ must return \textsc{True}.
	
	If $\tau$ reads that $t.\height > \uNode.\threshold$ on line~\ref{ln:ib:u0b}, then the interpreted bit of $t$ is 1 immediately after the read. So $\tau$ reads that $t.\height \leq \uNode.\threshold$.
	If $\tau$ reads that $t.\height \geq \uNode.\insthreshold$ on line~\ref{ln:ib:l1b}, then the interpreted bit of $t$ is 1 immediately after the read.  So $\tau$ reads that $t.\height < \uNode.\insthreshold$. Since \textsc{FirstActivated}$(\uNode)$ on line~\ref{ln:ib:l1b} returns \textsc{True}, it follows that $\tau$ returns 0.
	\end{proof}

	\begin{lemma}\label{lemma:IB_return1}
	Consider an instance $\tau$ of \textsc{InterpretedBit}$(t)$ for some binary trie node $t$. If the interpreted bit of $t$ is 1 throughout $\tau$, then $\tau$ returns 1.
	\end{lemma}
	
	\begin{proof}
	Let $\uNode$ be the update node returned by \textsc{FindLatest}$(t.\dNodePtr.\key)$ on line~\ref{ln:ib:findlatest}.
	If $\uNode$ is an INS node, then it follows by line~\ref{ln:ib:ins} that $\tau$ returns 1. So suppose $\uNode$ is a DEL node.
	
	Suppose, for contradiction, that $\tau$ returns 0.
	From the code, it follows that this can only happen if $\tau$ reads  $t.\height \leq \uNode.\threshold$ on line~\ref{ln:ib:u0b}, $t.\height < \uNode.\insthreshold$ on line~\ref{ln:ib:l1b}, and \textsc{FirstActivated}$(\uNode)$ returns \textsc{True} on line~\ref{ln:ib:l1b}.
	
	Since \textsc{FirstActivated}$(\uNode)$ returns \textsc{True}, 
	it follows from Observation~\ref{lemma:findLatestRelaxed} and Observation~\ref{lemma:firstActivatedRelaxed_true} that
	$\uNode$ is the latest update node from the end of $\tau$'s call to \textsc{FindLatest} and the beginning of \textsc{FirstActivated}$(\uNode)$. 
	So the interpreted bit of $t$ depends on $\uNode$ and is determined by $\uNode.\insthreshold$ and $\uNode.\threshold$ throughout this interval.
	Since $\uNode.\insthreshold$ only decreases, in the configuration immediately after $\tau$  reads $t.\height \leq \uNode.\threshold$ on line~\ref{ln:ib:u0b}, it also holds that $t.\height < \uNode.\insthreshold$. Therefore, in this configuration, the interpreted bit of $t$ is 0, a contradiction.
	
	
	\end{proof}

	Let $iOp$ be the latest \textsc{TrieInsert} operation with key $x$ in configuration $C$.
	The following definitions describe the number of iterations of \textsc{InsertBinaryTrie}  $iOp$ has completed.
	For the leaf $\ell$ with key $x$, we say that $iOp$ has \textit{completed iteration $\ell$} in all configurations after $iOp$ is linearized.
	Now consider an internal binary trie node $t$ where $x \in U_t$.
	We say that $iOp$ has \textit{completed iteration $t$ in $C$} of \textsc{InsertBinaryTrie} if, during iteration $t$ of \textsc{InsertBinaryTrie}, $iOp$
	\begin{itemize}
	\item reads that $\uNode$, the update node returned on line~\ref{ln:insert:findLatest}, has type INS on line~\ref{ln:insert:delCheck}, 
	\item reads that $t.\dNodePtr \neq \uNode$ and $t.\height > \uNode.\threshold$ on line~\ref{ln:insert:IfThreshold},
	\item reads that  $t.\height \geq \uNode.\insthreshold$ on line~\ref{ln:insert:IfMinwrite}, or
	\item performed a \textsc{MinWrite} of $t.height$ to $\uNode$ on line~\ref{ln:insert:minwrite}.
	\end{itemize}
	Intuitively, this means that $iOp$ has ensured that $t$ has interpreted bit 1.
	Note that if $iOp$ returns while performing iteration $t$ of \textsc{InsertBinaryTrie} on line~\ref{ln:insert:firstactivated2}, it does not complete iteration $t$. 
	
	We give similar definitions for \textsc{TrieDelete} operations.
	Let $dOp$ be the \textsc{TrieDelete} operation with key $x$ in configuration $C$, and
	let $\dNode$ be the DEL node created by $dOp$.
	For the leaf $\ell$ with key $x$, we say that $dOp$ has \textit{completed iteration $\ell$} in all configurations after $dOp$ is linearized.
	Now consider an internal binary trie node $t$ where $x \in U_t$.
	We say that $dOp$ has \textit{completed iteration $t$ in $C$} of \textsc{DeleteBinaryTrie} if $dOp$ performed the write with value $t.\height$ to $\dNode.\threshold$ on line~\ref{ln:delete:threshold} of \textsc{DeleteBinaryTrie}.
	

	We now prove several lemmas about update operations completing iteration $t$, for some binary trie node $t$.

	\begin{observation}\label{obs:completeIterationSeq}
	Suppose an update operation $uOp$ with key $x$ is in iteration $t$ (of \textsc{InsertBinaryTrie} for \textsc{TrieInsert} operations or \textsc{DeleteBinaryTrie} for \textsc{TrieDelete} operations). Then for each node $t'$ on the path from the parent of the leaf with key $x$ to $t$, $uOp$ has completed iteration $t'$.
	\end{observation}
	
	\begin{proof}
	From the code of \textsc{InsertBinaryTrie} and \textsc{DeleteBinaryTrie}, if $uOp$ does not complete an iteration, it performs a step that causes it to return from \textsc{InsertBinaryTrie} or \textsc{DeleteBinaryTrie} during that iteration.
	Since $uOp$ completes iterations up the binary trie, starting from the leaf with key $x$ to the root, $uOp$ has previously completed iterations for all binary trie nodes on the path to $t$.
	\end{proof}
	
	\begin{lemma}\label{lemma:delCompleteIteration}
	Let $\dNode$ be the DEL node created by an $S$-modifying \textsc{TrieDelete}$(x)$ operation, $dOp$, where  $\dNode.\threshold > 0$. For each binary trie node $t$ where $t.\height \leq \dNode.\threshold$ and  $x \in U_t$, $dOp$ has completed iteration $t$ of \textsc{DeleteBinaryTrie}. 
	\end{lemma}
	
	\begin{proof}
	By Observation~\ref{obs:dNodeThreshold}, only the \textsc{TrieDelete} operation that created $\dNode$ writes to $dOp$. By Lemma~\ref{lemma:sequenceCAS}, each completed iteration of \textsc{DeleteBinaryTrie} increases $\dNode.\threshold$ by 1. By definition, $dOp$ completes iteration $t$ of \textsc{DeleteBinaryTrie} immediately after it writes $t.\height$ to $\dNode$. By Observation~\ref{obs:completeIterationSeq}, $dOp$ has completed iteration $t$ of \textsc{DeleteBinaryTrie} for each node $t$ where $x \in U_t$ and $t.\height \leq \dNode.\threshold$.
	\end{proof}

	The following lemma implies that, in all configurations $C$ and all binary trie nodes $t$, either the interpreted bit of $t$ is correct or there exists a key $x$ in $U_t$ whose latest update operation has a potential update to $t$. It is the main lemma used to show that property IB1 of the interpreted bits is satisfied.
	
	\begin{lemma}\label{lemma:ib1_technical}
	For any binary trie node $t$, if there exists a key $x \in U_t$ such that the latest update operation with key $x$ is a \textsc{TrieInsert}$(x)$ operation that has completed iteration $t$ of \textsc{InsertBinaryTrie}, then $t$ has interpreted bit 1.
	\end{lemma}
	
	\begin{proof}
	We prove by induction on the configurations of the execution. In the initial configuration, the latest update operation of every key in $U$ is a dummy \textsc{TrieDelete} operation. So the lemma holds in the base case.
	


	Suppose the lemma is true in all configuration before a step $s$ by an update operation $uOp$ with key $x$, and we show that it is true for some binary trie node $t$ after $s$. Let $C'$ and $C$ be the configurations immediately before and after $s$, respectively. 
	This step by $uOp$ can only affect the truth of the claim for $t$ if $uOp$ is $S$-modifying and $x \in U_t$. 
	
	First suppose $t$ is the leaf with key $x$. If $s$ is a step in which an $S$-modifying \textsc{TrieDelete}$(x)$ operation  is linearized, then the latest update operation with key $x$ is a \textsc{TrieDelete}$(x)$ operation. So the lemma is true for $t$ in $C$.
	The only other step that affects the truth of the lemma is if $s$ is a step in which an $S$-modifying \textsc{TrieInsert}$(x)$ operation  is linearized.
	By definition, $iOp$ has completed iteration $t$ immediately after it is linearized. So immediately after $s$, $\latest[x]$ points to the INS node, $\uNode$, created by $uOp$. By definition, $t$ has interpreted bit 1 in $C$.

	Now suppose $t$ is an internal binary trie node, and suppose the lemma holds for all proper descendants of $t$ and all configurations up to $C$. We show that the lemma holds for $t$ in $C$.
	The update node that the interpreted bit of $t$ depends on can change if  $\latest[x]$ changes to point to a new update node, or if $t.\dNodePtr$ is changed to point to a DEL node with different key. This may cause the interpreted bit of $t$ to change. It may also change by decreasing the $\threshold$ field or increasing the $\insthreshold$ field of the update node that the interpreted bit of $t$ depends on. Finally, we consider the case when $s$ is by a latest \textsc{TrieInsert} operation that completes iteration $t$.
	
	\begin{itemize}
	\item Suppose $s$ is a successful CAS (on line~\ref{ln2:insert:cas_latest_head}) that changes $\latest[x]$ to point to an INS node. 
	The interpreted bit of $t$ does not change as a result of $s$ unless $t.\dNodePtr.\key = x$. In this case, it follows by definition that $t$ has interpreted bit 1 in $C$.
	
	\item Suppose $s$ is a successful CAS (on line~\ref{ln2:delete:cas_latest_head}) that changes $\latest[x]$ to point to a DEL node, $\dNode$.
	The interpreted bit of $t$ does not change as a result of $s$ unless $t.\dNodePtr.\key = x$.
	Since $\dNode$ is a newly created DEL node, $\dNode.\threshold = 0$. 
	Since $t.\height \geq 1$, it follows by definition that $t$ has interpreted bit 1 in $C$.
	
	\item Suppose $s$ is a successful CAS (on line~\ref{ln:delete:trieCAS} or line~\ref{ln:delete:trieCAS2}) that changes $t.\dNodePtr$ to point to a different DEL node.
	Only the interpreted bit of $t$ may change as a result of $s$. By Lemma~\ref{lemma:deleteUpdateTrieVal}, the interpreted bit of $t$ is 1 in $C$.
	
	\item Suppose $s$ is a step in which $uOp$ completes iteration $t$ of \textsc{InsertBinaryTrie} (as a result of a read on line~\ref{ln:insert:delCheck} or line~\ref{ln:insert:IfThreshold}, or a $\textsc{MinWrite}$ of $t.\height$ to $\dNode.\insthreshold$). We argue that the interpreted bit of $t$ is 1 in $C$.
	
	
	
	
	During iteration $t$ of $uOp$'s instance of \textsc{InsertBinaryTrie}, $uOp$ reads $t.\dNodePtr$ on line~\ref{ln:insert:findLatest}. Let the configuration immediately after this read be $C''$.
	Let $\uNode$ be the update node returned by \textsc{FindLatest}$(t.\dNodePtr)$ on line~\ref{ln:insert:findLatest}.
	By Lemma~\ref{lemma:findLatestRelaxed}, there is a configuration during this instance of \textsc{FindLatest} in which $\latest[\uNode.\key]$ points to $\uNode$. 
	
	Suppose the interpreted bit of $t$ no longer depends on  $\uNode$ in $C$, but instead on a latest update node, $\uNode' \neq \uNode$. 
	Let $uOp'$ be the update operation that created $\uNode'$.
	If $\uNode'$ is an INS node, then the interpreted bit of $t$ is 1 in $C$. So suppose  $\uNode'$ is a DEL node.
	Suppose, for contradiction, that the interpreted bit of $t$ is 0 in $C$.
	By the definition of interpreted bits, it is necessary that $\uNode'.\threshold \geq t.\height$ in $C$.
	By Observation~\ref{obs:dNodeThreshold}, only $uOp'$ changes  $\uNode'.\threshold$.
	Before $uOp'$ increments  $\uNode'.\threshold$ to equal $t.\height$, it performs a successful CAS, $cas$, on $t.\dNodePtr$ so it points to $\uNode'$ (on line~\ref{ln:delete:trieCAS} or \ref{ln:delete:trieCAS2}). 
	If $t.\dNodePtr$ is changed to point to some DEL node other than $\uNode'$, then 	Lemma~\ref{lemma:staleCASFail} and the fact that $\uNode'$ is a latest update node in $C$
	implies that the interpreted bit of $t$ cannot depend on $\uNode'$ in $C$.
	So $t.\dNodePtr$ points to  $\uNode'$ in all configurations immediately after $cas$ until $C$.
	Therefore, $cas$ occurs after $C''$, otherwise $uOp$'s instance of \textsc{FindLatest} would return $\uNode'$.
	By Observation~\ref{obs:completeIterationSeq}, $uOp$ has completed iteration $t'$ of \textsc{InsertBinaryTrie}, where $t'$ is a child of $t$, prior to the start of iteration $t$.
	By the induction hypothesis, the interpreted bit of $t'$ is 1 in all configurations from $C''$ to $C$. 
	Sometime after $uOp$ performs $cas$ but before it increments  $\uNode'.\threshold$ to equal $t.\height$, $uOp'$ verifies that \textsc{InterpretedBit}$(t')$ returns 0 on line~\ref{ln:delete:checkIB2}, and returns if not.
	However, by Lemma~\ref{lemma:IB_return1}, $uOp'$'s instance of \textsc{InterpretedBit}$(t')$ returns 1. So $uOp'$ returns before  it increments  $\uNode'.\threshold$ to equal $t.\height$. So the interpreted bit of $t$ is 1 in $C$, a contradiction.

	So suppose the interpreted bit of $t$ depends on $\uNode$ in $C$. If $s$ is a read that $\uNode.\type = \text{INS}$ on line~\ref{ln:insert:delCheck}, $s$ is a read that $t.\height \geq \uNode.\insthreshold$ on line~\ref{ln:insert:IfMinwrite}, or $s$ is a \textsc{MinWrite} of $t.\height$ to $\uNode.\insthreshold$ on line~\ref{ln:insert:minwrite}, then the interpreted bit of $t$ 1 in $C$.
	So suppose $s$ completes iteration $t$ because it reads $t.\height > \uNode.\threshold$ on line~\ref{ln:insert:IfThreshold}. It follow by definition that the interpreted bit of $t$ is 1 in $C$.
	
	
	\item Suppose $s$ is a write of $t.\height$ to $\dNode.\threshold$ on line~\ref{ln:delete:threshold} of \textsc{DeleteBinaryTrie}. The interpreted bit of $t$ may only change as a result of $s$ if the interpreted bit of $t$ depends on $\dNode$.
	If there is no key $x' \in U_t$ such that the latest update operation with key $x'$ is a \textsc{TrieInsert}$(x')$ operation that has completed iteration $t$ in $C'$, then the lemma is true in $C$. So suppose there is such an \textsc{TrieInsert}$(x')$ operation, $iOp$.
	
	Since $uOp$ is poised to perform $s$ (a write to $\dNode.\threshold$ on line~\ref{ln:delete:threshold}), it follows from the code that $uOp$ had previously performed a successful CAS, $cas$, (on line~\ref{ln:delete:trieCAS} or \ref{ln:delete:trieCAS2}) that successfully changed $t.\dNodePtr$ to point to $\dNode$. Since the interpreted bit of $t$ depends on $\dNode$ in $C$, it follows from Lemma~\ref{lemma:staleCASFail} that $t.\dNodePtr$ points to $\dNode$ in all configurations immediately after $cas$ to $C$.
	Since $iOp$ has completed iteration $t$, it follows from Observation~\ref{obs:completeIterationSeq} that it also previously completed iteration $t'$ for some child of $t$. 
	
	Suppose $cas$ occurs after $iOp$ completes iteration $t'$. By the induction hypothesis, the interpreted bit of $t'$ is 1 from when it competed iteration $t'$ to $C$. It follows from the code that, sometime between $cas$ and $s$,  $uOp$ performs \textsc{InterpretedBit}$(t')$ on line~\ref{ln:delete:checkIB2}. By Lemma~\ref{lemma:IB_return1}, $uOp$'s instance of \textsc{InterpretedBit}$(t')$ returns 1. So $uOp$ returns before it performs $s$, a contradiction.
	
	So suppose $cas$ occurs before $iOp$ completes iteration $t'$. Consider iteration $t$ of $iOp$'s instance of \textsc{InsertBinaryTrie}. 
	Since $\dNode$ is pointed to by $t.\dNodePtr$ and is the latest update node with key $x'$ throughout this iteration, it follows 
	from Observation~\ref{lemma:findLatestRelaxed} that $iOp$'s instance of \textsc{FindLatest}$(t.\dNodePtr)$ on line~\ref{ln:insert:findLatest} returns $\dNode$.
	So the if-statement on line~\ref{ln:insert:delCheck} is true.
	Since $iOp$'s read of $t.\dNodePtr$ always returns $\dNode$, the if-statement on line~\ref{ln:insert:IfThreshold} is true.
	Since $\iNode$ is the latest update node with key $x'$ throughout iteration $t$, \textsc{FirstActivated}$(\iNode)$ on line~\ref{ln:insert:firstactivated2} returns \textsc{True}. If follows from the code that, $iOp$ either completes iteration $t$ by reading $t.\height \geq \uNode.\insthreshold$ on line~\ref{ln:insert:IfMinwrite}
	or
	performing a min-write of $t.\height$ to $\dNode.\insthreshold$ on line~\ref{ln:insert:minwrite}. 
	Since $\dNode.\insthreshold$ is strictly decreasing, it follows that $\dNode.\insthreshold \leq t.\height$ in $C$. 
	By definition, the interpreted bit of $t$ is 1 in $C$.

	\end{itemize}
\end{proof}

Lemma~\ref{lemma:ib1_technical} can be used to prove property IB1 of the interpreted bits.

\begin{lemma}\label{lemma:ib1}
For all binary trie nodes $t$ and configurations $C$, if there exists $x \in U_t \cap S$ such that the last $S$-modifying \textsc{TrieInsert}$(x)$ operation linearized prior to $C$ is no longer active,
then the interpreted bit of $t$ is 1 in $C$.
\end{lemma}

\begin{proof}
Let $iOp$ be the last $S$-modifying \textsc{TrieInsert}$(x)$ operation linearized prior to $C$. 
Since $x \in S$, $iOp$ is the latest $S$-modifying update operation with key $x$ in $C$.
Since $iOp$ is the latest update operation with key $x$ in all configurations during $iOp$, Lemma~\ref{lemma:firstActivatedRelaxed_false} states that any instance of \textsc{FirstActivated}$(\iNode)$ invoked by $iOp$ returns \textsc{True}. From the code of \textsc{InsertBinaryTrie}, $iOp$ does not return from \textsc{InsertBinaryTrie} unless \textsc{FirstActivated}$(\iNode)$ returns \textsc{False} on line~\ref{ln:insert:firstactivated2}. Therefore, for all binary trie nodes $t$ where $x \in U_t$,
$iOp$ has completed iteration $t$ of \textsc{InsertBinaryTrie} in configuration $C$. By Lemma~\ref{lemma:ib1_technical}, 
the interpreted bit of $t$ is 1 in $C$.
\end{proof}





We now focus on showing that property IB0 of the interpreted bits is satisfied. We begin by introducing a few definitions used to indicate which \textsc{TrieDelete} operations may change the interpreted bit of a binary trie node $t$ in a future configuration.

Let $dOp$ be an $S$-modifying \textsc{TrieDelete}$(x)$ operation that created a DEL node, $\dNode$. Let $t$ be a binary trie node such that $x \in U_t$.
We say that $dOp$ is \textit{stopped} if $\dNode.\stopflag = \textsc{True}$ or $\dNode.\insthreshold \neq b+1$. Note that this means $dOp$ will return if it executes line~\ref{ln:delete:readstop1} or \ref{ln:delete:readstop2} of \textsc{DeleteBinaryTrie}.
We say that $dOp$ has a \textit{potential update to $t$} in a configuration $C$ if
\begin{itemize}
\item $dOp$ is the latest update operation with key $x$,
\item $dOp$ is not stopped,
\item $dOp$ has not yet completed iteration $t$ of \textsc{DeleteBinaryTrie},
\item $dOp$ has not yet returned from \textsc{DeleteBinaryTrie}, and
\item one of the following is true:
\begin{itemize}
\item $dOp$ has not yet read $t.\dNodePtr$ on line~\ref{ln:delete:readPtr2},
\item  $dOp$ has read $t.\dNodePtr$ on line~\ref{ln:delete:readPtr2} but not yet performed a CAS on $t.\dNodePtr$ on line~\ref{ln:delete:trieCAS2} and $t.\dNodePtr$ has not changed since $dOp$'s last read of $t.\dNodePtr$, or
\item  $dOp$ has successfully performed a CAS on $t.\dNodePtr$ (on line~\ref{ln:delete:trieCAS} or \ref{ln:delete:trieCAS2}) and $t.\dNodePtr$ has not changed since this CAS.
\end{itemize}

\end{itemize}
A dummy \textsc{TrieDelete} operation does not have a potential update to any binary trie node in any configuration.
When $dOp$ has a potential update to $t$ in $C$, this means that if the lines~\ref{ln:delete:checkIB1} and \ref{ln:delete:checkIB2} of \textsc{DeleteBinaryTrie} are ignored, $dOp$ will complete iteration $t$ of \textsc{DeleteBinaryTrie} in its solo continuation from $C$, changing the interpreted bit of $t$ to 0.

Consider a configuration $C$ and binary trie node $t$ in which all the latest $S$-modifying update operations with keys in $U_t$ are \textsc{TrieDelete} operations that have invoked \textsc{DeleteBinaryTrie}.
Let $D(t,C)$ denote the earliest configuration before $C$ in which 
all the latest $S$-modifying \textsc{TrieDelete} operations with keys in $U_t$ have invoked \textsc{DeleteBinaryTrie}
and are the same as those in $C$.
Let $OP(t,C)$ be the set of latest \textsc{TrieDelete} operations with keys in $U_t$ that have a potential update to $t$ in $D(t,C)$.
Note that in the step immediately before $D(t,C)$, a latest \textsc{TrieDelete} operation 
invokes \textsc{DeleteBinaryTrie}.
This \textsc{TrieDelete} operation has a potential update to $t$, so $OP(t,C)$ is non-empty.



The next lemma describes when writes to $\dNode.\insthreshold$ or $\dNode.\stopflag$ may occur, for some DEL node $\dNode$. These writes are only performed by \textsc{TrieInsert} operations (or \textsc{TrieDelete} operations helping \textsc{TrieInsert} operations). Intuitively, once the latest update operations for all keys in $U_t$ become \textsc{TrieDelete} operations  that have invoked \textsc{DeleteBinaryTrie}, there are no longer any latest \textsc{TrieInsert} operations with keys in $U_t$ that perform these writes.

\begin{lemma}\label{lemma:no_minwrite}
Consider a configuration $C$ and a binary trie node $t$ such that the latest update operations for all keys in $U_t$ are \textsc{TrieDelete} operations that have invoked \textsc{DeleteBinaryTrie}.
Let $\dNode$ be the DEL node created by the latest \textsc{TrieDelete}$(x)$ operation, where $x \in U_t$.
For any key $x' \in U_t$,  
no update operation with key $x'$ writes to $\dNode.\insthreshold$ (on line~\ref{ln:insert:minwrite}) or $\dNode.\stopflag$ (on line~\ref{ln2:delete:setStop} or line~\ref{ln2:insert:helpStop}) between $D(t,C)$ and $C$.
\end{lemma}

\begin{proof}
Suppose, for contradiction, 
that in a step $s$, some update operation with key $x' \in U_t$ writes to $\dNode.\insthreshold$ (on line~\ref{ln:insert:minwrite}) or $\dNode.\stopflag$ (on line~\ref{ln2:delete:setStop} or line~\ref{ln2:insert:helpStop}) between $D(t,C)$ and $C$.

First suppose $s$ is a write to $\dNode.\stopflag$ (on line~\ref{ln2:delete:setStop} or line~\ref{ln2:insert:helpStop}).
This only occurs after some $S$-modifying \textsc{TrieInsert}$(x')$ operation has set $\iNode.\target = \dNode$ (on line~\ref{ln:insert:setTarget}), where $\iNode$ is the INS node that it created.
If $s$ occurs on line~\ref{ln2:delete:setStop}, then it follows from the code of \textsc{TrieDelete} and the CAS on line~\ref{ln2:delete:cas_latest_head} that
$s$ is performed by a \textsc{TrieDelete}$(x')$ operation, $dOp'$, linearized immediately after this \textsc{TrieInsert} operation.
Otherwise $s$ occurs on line~\ref{ln2:insert:helpStop}, in which case a \textsc{TrieInsert}$(x')$ operation is helping $dOp'$ perform this write (before attempting to change $\latest[x]$ to point to its own INS node on line~\ref{ln2:insert:cas_latest_head}).
In either case, $s$ occurs when $dOp'$ is the latest update operation with key $x'$ and $dOp'$ has not yet invoked \textsc{DeleteBinaryTrie}. 
Since $x' \in U_t$, $s$ must be before $D(t,C)$. 

Now suppose $s$ is a write to  $\dNode.\insthreshold$ (on line~\ref{ln:insert:minwrite}). So $s$ is performed by a \textsc{TrieInsert}$(x')$ operation, $iOp$, and let $\iNode$ be the INS node it created.
Prior to $s$, $iOp$ set $\iNode.\target = \dNode$ on line~\ref{ln:insert:setTarget}, and $iOp$'s instance of \textsc{FirstActivated}$(\iNode)$ on line~\ref{ln:insert:firstactivated2} returns \textsc{True}. By Observation~\ref{lemma:firstActivatedRelaxed_true}, there is a configuration $C'$ during this instance of \textsc{FirstActivated}$(\iNode)$ in which $\iNode$ is the latest update node with key $x'$. 
Since $x' \in U_t$ and $iOp$ is a latest \textsc{TrieInsert}$(x')$ operation, it follows by definition that $C'$ occurs before $D(t,C)$. 

Since $s$ occurs after $D(t,C)$,
it follows from the code that $\iNode.\target$ is not changed between $C'$ and $D(t,C)$, and $\iNode.\target = \dNode$ during this interval.
Consider the $S$-modifying \textsc{TrieDelete}$(x')$ operation, $dOp'$, linearized after $iOp$. So $dOp'$ was linearized sometime between $C'$ and $D(t,C)$. 
Prior to invoking \textsc{DeleteBinaryTrie} (and hence prior to $D(t,C)$), it reads  $\iNode.\target = \dNode$.

If $dOp'$ invokes \textsc{DeleteBinaryTrie} before $D(t,C)$, then it sets $\dNode.\stopflag = \textsc{True}$ (on line~\ref{ln2:delete:setStop}). So $s$ occurs before $D(t,C)$.
If not, then since the latest update operations for all keys in $U_t$ in $D(t,C)$ are \textsc{TrieDelete} operations that have invoked \textsc{DeleteBinaryTrie}, there must be an $S$-modifying \textsc{TrieInsert}$(x')$ operation and $S$-modifying \textsc{TrieDelete}$(x')$ operation linearized after  $dOp'$ but before $D(t,C)$.
Let  $iOp''$ be this \textsc{TrieInsert}$(x')$ operation linearized after $dOp'$ but before $D(t,C)$.
Then before $iOp''$ is linearized, $iOp''$ helps $dOp'$ set $\dNode.\stopflag = \textsc{True}$ (on line~\ref{ln2:insert:helpStop}). 
So $s$ occurs before $D(t,C)$.
\end{proof}

The following lemma is a stronger version of Lemma~\ref{lemma:no_minwrite} that applies to the DEL node of a \textsc{TrieDelete}$(x)$ operation in $OP(t,C)$.

\begin{lemma}\label{lemma:no_minwrite_OP}
Consider a configuration $C$ and a binary trie node $t$ such that the latest update operations for all keys in $U_t$ are \textsc{TrieDelete} operations that have invoked \textsc{DeleteBinaryTrie}.
Let $\dNode$ be the DEL node created by a \textsc{TrieDelete}$(x)$ operation, $dOp$, in $OP(t,C)$, where $x \in U_t$.
For any key $x' \in U_t$,  
no update operation with key $x'$ writes to $\dNode.\insthreshold$ (on line~\ref{ln:insert:minwrite}) or $\dNode.\stopflag$ (on line~\ref{ln2:delete:setStop} or line~\ref{ln2:insert:helpStop}) prior to $C$.
\end{lemma}

\begin{proof}
Suppose, for contradiction, that during a step $s$, some update operation with key $x'$ writes to  $\dNode.\threshold$ (on line~\ref{ln:delete:threshold}) or $\dNode.\stopflag$ (on line~\ref{ln2:delete:setStop} or line~\ref{ln2:insert:helpStop}) prior to $C$.
By Lemma~\ref{lemma:no_minwrite}, $s$ occurs before $D(t,C)$. Since $dOp$ is stopped by $s$ before $D(t,C)$, $dOp$ does not have a potential update to $t$ in $D(t,C)$. So $dOp \notin OP(t,C)$, a contradiction. 
\end{proof}

The most important property of the \textsc{TrieDelete} operations in $OP(t,C)$, which we prove in the following lemma, is that they are not stopped in $C$ provided they still have a potential update to $t$.

\begin{lemma}\label{lemma:no_flag}
Consider a configuration $C$ and a binary trie node $t$ such that the latest update operations for all keys in $U_t$ are \textsc{TrieDelete} operations that have invoked \textsc{DeleteBinaryTrie}.
For any \textsc{TrieDelete}$(x)$ operation, $dOp$, in $OP(t,C)$ that still has a potential update to $t$ in $C$, $dOp$ is not stopped by any step prior to $C$.
\end{lemma}

\begin{proof}
Suppose, for contradiction, that $dOp$ is stopped by a step $s$, where $s$ occurs prior to $C$.
So $s$ is a write to  $\dNode.\insthreshold$ (on line~\ref{ln:insert:minwrite}) or $\dNode.\stopflag$ (on line~\ref{ln2:delete:setStop} or line~\ref{ln2:insert:helpStop}) prior to $C$.
In either case, some $S$-modifying \textsc{TrieInsert}$(x')$ operation, $iOp$, has set $\iNode.\target = \dNode$ (on line~\ref{ln:insert:setTarget}), where $\iNode$ is the INS node that it created.
Suppose this write to $\iNode.\target$ occurs during iteration $t'$ of $iOp$'s instance of \textsc{InsertBinaryTrie}.
In the line of code prior, $iOp$ reads that $t'.\dNodePtr = \dNode$ or $t'.\height\leq \dNode.\threshold$ on line~\ref{ln:insert:IfThreshold}. 
If $t'.\height\leq \dNode.\threshold$, then Lemma~\ref{lemma:delCompleteIteration} implies that $dOp$ must have completed iteration $t'$ of \textsc{DeleteBinaryTrie} to successfully increment $\dNode.\threshold$ to equal $t'.\height$.
So in either case, $dOp$ has previously performed a successful CAS on  $t'.\dNodePtr$ to change it to point to $\dNode$. 
Since $dOp$ has a potential update to $t$ in $C$ and $dOp$ completes iterations in order from leaf to root, this implies $t'$ is a descendant of $t$ (i.e. either $t' = t$ or $t'$ is a proper descendant of $t$).
So $x' \in U_t$. By Lemma~\ref{lemma:no_minwrite}, $s$ does not occur prior to $C$, a contradiction.


\end{proof}

The next lemma is used to show property IB0 is satisfied by the interpreted bits.

\begin{lemma}\label{lemma:ib0_technical}

Consider any configuration $C$ and binary trie node $t$.
If the latest update operation for all keys in $U_t$ in $C$ are \textsc{TrieDelete} operations
that have invoked \textsc{DeleteBinaryTrie}
and no \textsc{TrieDelete} operation in $OP(t,C)$ has a potential update to $t$ in $C$, then
\begin{enumerate}
\item[(a)] $t$ has interpreted bit 0 in $C$, and
\item[(b)] when $t$ is not a leaf, $t.\dNodePtr$ points to a DEL node created by a \textsc{TrieDelete} operation in $OP(t,C)$.
\end{enumerate}
\end{lemma}

\begin{proof}

The proof is by induction on the configurations of the execution and the height of binary trie nodes.

In the initial configuration, the latest update operation of every key in $U$ is a dummy \textsc{TrieDelete} operation and each entry in $\latest$ points to a dummy update node.
The fields of the dummy update nodes are initialized so that the interpreted bits of all binary trie nodes are 0.
Since the dummy operations do not have a potential update to any binary trie node,
the lemma is true in the initial configuration.


Consider any other configuration $C$ of the execution, and assume the lemma holds for all binary trie nodes in all configurations before $C$.
Let $C'$ be the configuration immediately before $C$ and let $s$ be the step performed immediately following $C'$ that results in $C$. 

Consider a leaf $\ell$ of the binary trie with key $x$. 
Consider the latest update operation with key $x$ in $C$, which is a \textsc{TrieDelete}$(x)$ operation, $dOp$, that has invoked \textsc{DeleteBinaryTrie}. By definition, it has completed its update to $\ell$ immediately after it is linearized. By definition, $\latest[x]$ points to the DEL node, $\dNode$, created by $dOp$ in $C$.
Since it always holds that $\dNode.\threshold \geq \ell.\height = 0$ and $\dNode.\insthreshold > 0$, the interpreted bit of $t$ is 0.


Now suppose $t$ is an internal binary trie node. We assume the lemma is true in $C'$ for all binary trie nodes that are proper descendants of $t$.
Only steps by $S$-modifying update  operations with keys in $U_t$  can affect the truth of the lemma.
So suppose $s$ is performed by an $S$-modifying update operation.


By definition,
step $s$ can change the interpreted bit of $t$ if and only if it is a successful CAS that changes $t.\dNodePtr$, an update to $\dNode.\insthreshold$ or $\dNode.\threshold$, where 
$\dNode$ is the DEL node pointed to by $\latest[t.\dNodePtr.\key]$, or a successful CAS that changes $\latest[t.\dNodePtr.\key]$ to point to a different update node. Changing $t.\dNodePtr$ can also cause a \textsc{TrieDelete} operation to no longer have a potential update to $t$.


Step $s$ can stop a \textsc{TrieDelete} operation if and only if it is a write to $\dNode.\stopflag$ or $\dNode.\insthreshold$ for some DEL node $\dNode$. This step results in the operation that created $\dNode$ from no longer having a potential update to $t$.

Step $s$ can cause a \textsc{TrieDelete} operation to return from \textsc{DeleteBinaryTrie} if it reads $\dNode.\stopflag = \textsc{True}$, $\dNode.\insthreshold \neq b+1$, \textsc{FirstActivated}$(\dNode)$ returns \textsc{False}, or is an unsuccessful CAS on $t.\dNodePtr$.


We consider each type of step $s$ performed by an update operation with key $x \in U_t$, and argue the lemma holds for $t$ in $C$. We first consider steps in which $S$-modifying update operations with key $x$ are linearized or invoke \textsc{DeleteBinaryTrie}.

\begin{itemize}
\item Suppose $s$ is a successful CAS (on line~\ref{ln2:insert:cas_latest_head}) that changes $\latest[x]$ to point to an INS node, $\iNode$. 
So $s$ is performed by a \textsc{TrieInsert}$(x)$ operation, $iOp$, that created $\iNode$.
In $C$, $iOp$ is the latest update operation with key $x \in U_t$.
Since the latest update operations for all keys in $U_t$ are not  \textsc{TrieDelete} operations, the lemma is true for $t$ in $C$.

\item Suppose $s$ is a successful CAS (on line~\ref{ln2:delete:cas_latest_head}) that changes $\latest[x]$ to point to a DEL node, $\dNode$.
So $s$ is performed by a \textsc{TrieDelete}$(x)$ operation, $dOp$, that created $\dNode$. 
In $C$, $dOp$ is the latest update operation with key $x$.
It follows from the code that $dOp$ has not yet invoked \textsc{DeleteBinaryTrie}, so the lemma is true for $t$ in $C$.

\item Suppose $s$ is a step where a latest \textsc{TrieDelete}$(x)$ operation, $dOp$, invokes its instance of \textsc{DeleteBinaryTrie}. 
If there is another latest update operation with key in $U_t$ that has not yet invoked \textsc{DeleteBinaryTrie}, then the lemma is true for $t$ in $C$.
So, after $s$, the latest update operation for all keys in $U_t$ in $C$ are \textsc{TrieDelete} operations
that have invoked \textsc{DeleteBinaryTrie}. 
So $C = D(t,C)$.

We show that $dOp$ has a potential update to $t$ in $C$, and hence the lemma is true for $t$ in $C$.
If $dOp$ is not the latest update operation with key $x$, then $s$ does not affect the truth of the lemma. Then since the lemma is true for $t$ in $C'$, it is also true for $t$ in $C$.
So $dOp$ is the latest update operation with key $x$.
Since $dOp$ has not yet completed any iterations of \textsc{DeleteBinaryTrie}, $\dNode.\threshold = 0$.
By line~\ref{ln:insert:IfThreshold} of \textsc{InsertBinaryTrie}, no \textsc{TrieInsert} operation writes to $\dNode.\insthreshold$ or sets $\iNode.\target = \dNode$, where $\iNode$ is the INS node created by the  \textsc{TrieInsert} operation prior to $C$. So $dOp$ is  not stopped in $C$.
It follows by definition that $dOp$ has a potential update to $t$ in $C$ and $dOp \in OP(t,C)$.


\end{itemize}

In the remainder of cases for the types of step $s$, we assume that in $C$ (and $C'$), the latest update operation for all keys in $U_t$ are \textsc{TrieDelete} operations that have invoked \textsc{DeleteBinarayTrie}. If not, the lemma will be true for $t$ in $C$.

\begin{itemize}

\item Suppose $s$ is a write that sets $\dNode'.\stopflag$ to \textsc{True} of \textsc{TrieDelete} to some DEL node, $\dNode'$, with key $x$.
Then $s$ is either performed by a \textsc{TrieDelete}$(x)$ operation on line~\ref{ln2:delete:setStop}, or by a \textsc{TrieInsert}$(x)$ operation on line~\ref{ln2:insert:helpStop}.
Let $dOp'$ be \textsc{TrieDelete}$(x)$ operation that created $\dNode'$. 

Suppose $s$ falsifies the lemma by causing $dOp'$ to no longer have a potential update to $t$ in $C$ and $dOp'$ is the only latest update operation with key in $U_t$ that has a potential update to $t$ in $C'$.
Since $dOp'$ has a potential update to $t$ in $C'$, it by definition has a potential update to $t$ in $D(t,C)$, so $dOp' \in OP(t,C)$.
By Lemma~\ref{lemma:no_flag}, $dOp'$ is not stopped by any step prior to $C$, a contradiction.

\item Suppose $s$ is a $\textsc{MinWrite}$ of $t.\height$ to $\dNode.\insthreshold$ on line~\ref{ln:insert:minwrite} of \textsc{InsertBinaryTrie}, for some DEL node $\dNode'$, where $x' = \dNode.\key$. 
Let $dOp'$ be the \textsc{TrieDelete}$(x')$ that created $\dNode'$, and let $iOp$ be the \textsc{TrieInsert}$(x)$ operation that performed $s$.

Suppose $s$ falsifies the lemma by causing $dOp'$ to no longer have a potential update to $t$ in $C$ and $dOp'$ is the only latest update operation with key in $U_t$ that has a potential update to $t$ in $C'$.
Since $dOp'$ has a potential update to $t$ in $C'$, it by definition has a potential update to $t$ in $D(t,C)$, so $dOp' \in OP(t,C)$.
By Lemma~\ref{lemma:no_flag}, $dOp'$ is not stopped by any step prior to $C$, a contradiction.

The only other way $s$ can falsify the lemma is by changing the interpreted bit of $t$ to 1.
So in $C'$ and $C$, the latest update operation for all keys in $U_t$ in $C$ are \textsc{TrieDelete} operations
that have invoked \textsc{DeleteBinaryTrie}, and none of these have a potential update to $t$. 
The interpreted bit of $t$ is 0 in $C'$, but is changed to 1 by $s$.
The interpreted bit of $t$ is not affected by $s$ unless it depends on $\dNode'$, so $t.\dNodePtr = \dNode'$ and $\latest[x']$ points to  $\dNode'$.
By the induction hypothesis (Lemma~\ref{lemma:ib0_technical}(b)), $dOp' \in OP(t,C)$. By Lemma~\ref{lemma:no_minwrite}, no update operation with key in $x$ writes to $\dNode'.\insthreshold$, so $s$ does not occur, a contradiction.


\item  Suppose $s$ is a step where a \textsc{TrieDelete}$(x)$ operation, $dOp$, in $OP(t,C)$ returns because some instance of \textsc{FirstActicated}$(\dNode)$ returns \textsc{False} (on line~\ref{ln:delete:firstactivated1} or line~\ref{ln:delete:firstactivated2}) during iteration $t$ of \textsc{DeleteBinaryTrie}.

By Observation~\ref{lemma:firstActivatedRelaxed_false}, there is a configuration during this instance of \textsc{FirstActivated}$(\dNode)$ in which $dOp$ is not the latest update operation with key $x$. So $dOp$ is not the latest update operation with key $x$ in $C'$. 
So $s$ does not affect the truth of the lemma.
Since the lemma is true for $t$ in $C'$, it is also true for $t$ in $C$.

\item  Suppose $s$ is a step where a \textsc{TrieDelete}$(x)$ operation, $dOp$, returns because it reads $\dNode.\stopflag = \textsc{True}$ or $\dNode.\threshold \neq b+1$ on line~\ref{ln:delete:readstop1} or line~\ref{ln:delete:readstop2} during iteration $t$ of \textsc{DeleteBinaryTrie}, where $\dNode$ is the DEL node $dOp$ created.

By definition, $dOp$ is stopped in $C'$ and $C$ and does not have a potential update to $t$ in $C'$ and $C$. 
So $s$ does not affect the truth of the lemma.
Since the lemma is true for $t$ in $C'$, it is also true for $t$ in $C$.

\item  Suppose $s$ is the step where a \textsc{TrieDelete}$(x)$ operation, $dOp$, returns from \textsc{DeleteBinaryTrie} because \textsc{InterpretedBit}$(t.\mathit{left})$ or \textsc{InterpretedBit}$(t.\mathit{right})$ returns 1.
Without loss of generality, suppose  \textsc{InterpretedBit}$(t.\mathit{left})$ returns 1.

By Lemma~\ref{lemma:IB_return0}, there is a configuration during this instance of  \textsc{InterpretedBit}$(t.\mathit{left})$ in which the interpreted bit of $t.\mathit{left}$ is 1.
The induction hypothesis for the children of $t$ implies that either there is a latest update operation with key in $U_{t.\mathit{left}}$ that is a \textsc{TrieInsert} operation, or there is a \textsc{TrieDelete} operation with key in $U_t$ that has a potential update to $t.\mathit{left}$. Since $U_{t.\mathit{left}} \subsetneq U_t$, the lemma is true for $t$ in $C$.

\item Suppose $s$ is the second unsuccessful CAS performed by a \textsc{TrieDelete}$(x)$ operation, $dOp$, during iteration $t$ of \textsc{DeleteBinaryTrie}.
Then $t.\dNodePtr$ has changed since $dOp$'s last read of $t.\dNodePtr$.
By definition, $dOp$ does not have a potential update to $t$ in $C'$ and $C$. Since the lemma is true for $t$ in $C'$, it must also be true for $t$ in $C$.

\item Suppose $s$ is a successful CAS that changes $t.\dNodePtr$ to point to DEL node, $\dNode$. So $s$ is performed by a \textsc{TrieDelete}$(x)$ operation, $dOp$, that created $\dNode$.

If there are still latest \textsc{TrieDelete} operations with keys in $U_t$ that have a potential update to $t$ in $C$, the lemma is true for $t$ in $C$. 
So in $C$, suppose no latest \textsc{TrieDelete} operations have a potential update to $t$. 
Note that if there is a latest \textsc{TrieDelete} operation with a key in $U_t$ with a potential update to $t$ in $C'$ that is in iteration $t'$ of \textsc{DeleteBinaryTrie}, where $t'$ is a descendant of $t$, then this \textsc{TrieDelete} operation will still have a potential update to $t$ in $C$.
So all latest update operations with keys in $U_t$ are in iteration $t$ of \textsc{DeleteBinaryTrie}.
Furthermore, these operations either already performed a successful CAS during iteration $t$ of \textsc{DeleteBinaryTrie} or have read $t.\dNodePtr$ on line~\ref{ln:delete:readPtr2} but not yet attempted their second CAS during iteration $t$ of \textsc{DeleteBinaryTrie}.


Suppose that in $C'$, $dOp$ is the latest update operation with key $x$ and $dOp$ is not stopped. 
Then $dOp$ has a potential update to $t$ in $C$. So the lemma is true for $t$ in $C$.


So suppose that, in $C'$, $dOp$ is stopped or is not the  latest update operation with key $x$. For each of these cases, we first argue that $t.\dNodePtr$ does not change between $D(t,C)$ and $C'$.
Suppose that in $C'$, $dOp$ is not the latest update operation with key $x$ in $C'$. 
Prior to $dOp$ performing $s$, $dOp$ reads $t.\dNodePtr$ (on line~\ref{ln:delete:readPtr} or \ref{ln:delete:readPtr2}), and then reads that $\textsc{FirstActivated}(\dNode) = \textsc{True}$ in the following line of code on line~\ref{ln:delete:firstactivated1} or line~\ref{ln:delete:firstactivated2}. 
So there is a latest $\textsc{TrieDelete}(x)$ operation, $dOp'$, that updates $\latest[x]$ to point to its own DEL node, $\dNode'$, sometime after $dOp$ reads that $\textsc{FirstActivated}(\dNode) = \textsc{True}$. 
So $D(t,C)$ occurs after $dOp$ reads $t.\dNodePtr$ (on line~\ref{ln:delete:readPtr} or \ref{ln:delete:readPtr2}). Since $s$ is successful, $t.\dNodePtr$ does not change between $D(t,C)$ and $C'$. 

Now we consider the case when $dOp$ is  stopped in $C'$.
Since $dOp$ is in iteration $t$ of \textsc{DeleteBinaryTrie}, it may only be stopped by update operations with keys in $U_t$.
Since $dOp$ is stopped in $C$, it follows by Lemma~\ref{lemma:no_minwrite} that $dOp \notin OP(t,C)$ and $dOp$ is stopped in a configuration before $D(t,C)$.
Prior to $dOp$ performing $s$, $dOp$ reads $t.\dNodePtr$ (on line~\ref{ln:delete:readPtr} or \ref{ln:delete:readPtr2}), and then reads that it is not stopped on line~\ref{ln:delete:readstop1} or line~\ref{ln:delete:readstop2}. 
So $D(t,C)$ occurs after $dOp$ reads $t.\dNodePtr$ (on line~\ref{ln:delete:readPtr} or \ref{ln:delete:readPtr2}). Since $s$ is successful, $t.\dNodePtr$ does not change between $D(t,C)$ and $C'$. 
Therefore, in either case, $t.\dNodePtr$ does not change between $D(t,C)$ and $C'$, so no successful CAS is performed on  $t.\dNodePtr$  between $D(t,C)$ and $C'$.

We now show that there is still a \textsc{TrieDelete}$(x)$ operation in $OP(t,C)$ with a potential update to $t$ in $C$.
Let $t_\ell$ and $t_r$ be the left and right children of $t$, respectively. Note that both $D(t_\ell,C)$ and $D(t_r,C)$ occur at or before $D(t,C)$.
By the induction hypothesis (Lemma~\ref{lemma:ib0_technical}(a)), $t_\ell$ and $t_r$ have interpreted bit 0 in $C'$. 
By the induction hypothesis (Lemma~\ref{lemma:ib0_technical}(b),  $t_\ell.\dNodePtr$ points to a \textsc{TrieDelete} operation in $OP(t_\ell,C)$ in $C'$ and  $t_r.\dNodePtr$ points to a \textsc{TrieDelete} operation in $OP(t_r,C)$ in $C'$.
Let $dOp'$ be the last operation in $OP(t,C)$ to start iteration $t$ of \textsc{DeleteBinaryTrie} prior to $C'$.
So $dOp'$ starts iteration $t$ of \textsc{DeleteBinaryTrie} sometime after $D(t,C)$.
Since $t_\ell$ and $t_r$ have interpreted bit 0 from when $dOp'$ starts iteration $t$ of \textsc{DeleteBinaryTrie} to $C'$, it follows by Lemma~\ref{lemma:IB_return0} that $dOp'$ does not return from \textsc{DeleteBinaryTrie} on line~\ref{ln:delete:checkIB1} or line~\ref{ln:delete:checkIB2} prior to $C'$.
Since $dOp'$ is a latest update operation from when it is linearized to $C'$, it follows by Observation~\ref{lemma:firstActivatedRelaxed_false} that $dOp'$ does not return from \textsc{DeleteBinaryTrie} on line~\ref{ln:delete:firstactivated1} or line~\ref{ln:delete:firstactivated1} prior to $C'$.
By Lemma~\ref{lemma:no_flag},  $dOp'$ is not stopped in $C'$. 
It follows from the code that $dOp'$ may only return from iteration $t$ of \textsc{DeleteBinaryTrie} by performing an unsuccessful CAS (on line~\ref{ln:delete:trieCAS} or \ref{ln:delete:trieCAS2}).

Suppose $dOp'$ performs a CAS on $t$ (on line~\ref{ln:delete:trieCAS})  before $C'$. Since $dOp'$'s last read of $t.\dNodePtr$ (on line~\ref{ln:delete:readPtr}) occurs between  $D(t,C)$ and $C'$, this CAS will be successful. 
This contradicts the fact that $t.\dNodePtr$ does not change between $D(t,C)$ and $C'$.
So $dOp'$ has not yet performed a CAS on $t$ (on line~\ref{ln:delete:trieCAS}).		
So $dOp'$ still has another attempt to change $t.\dNodePtr$ to point to its own DEL node (on line~\ref{ln:delete:trieCAS2}) in $C'$.
It follows from the code that $dOp'$ has a potential update to $t$ in $C$. So the lemma is true for $t$ in $C$.





\item Suppose $s$ is a step where a \textsc{TrieDelete}$(x)$ operation, $dOp$, completes iteration $t$ of \textsc{DeleteBinaryTrie} because it writes $t.\height$ to $\dNode.\threshold$ (on line~\ref{ln:delete:threshold}), where $\dNode$ is the DEL node created by $dOp$.


If there are still latest \textsc{TrieDelete} operations in $OP(t,C)$ with keys in $U_t$ that have a potential update to $t$ in $C$, the lemma is true for $t$ in $C$. 
So in $C$, suppose no latest \textsc{TrieDelete} operations in $OP(t,C)$ have a potential update to $t$.
Since $s$ can only cause $dOp$ to no longer have a potential update to $t$, it follows that, in $C'$, $dOp$ is the only 
\textsc{TrieDelete}$(x)$ operation in $OP(t,C)$ that has a potential update to $t$.
By definition, $dOp$ is the last operation prior to $C'$ to perform a successful CAS on $t.\dNodePtr$, changing it to point to $\dNode$.
Since $dOp$ is the latest update operation with key $x$ in $C'$, the interpreted bit of $t$ depends on $\dNode$. After $dOp$ performs $s$, $\dNode.\threshold = t.\height$. By definition, the interpreted bit of $t$ is 1, and $t.\dNodePtr$ points to a DEL node created by a \textsc{TrieDelete} operation in  $OP(t,C)$. So the lemma is true for $t$ in $C$.

\end{itemize}
\end{proof}

The following lemma states that property IB0 of the interpreted bits is satisfied. It is proved using Lemma~\ref{lemma:ib0_technical}.

\begin{lemma}\label{lemma:ib0}
For all binary trie nodes $t$ and configurations $C$,
if $U_t \cap S = \emptyset$ and for all $x \in U_t$, 
either there has been no $S$-modifying \textsc{TrieDelete}$(x)$ operation or
the last $S$-modifying \textsc{TrieDelete}$(x)$ operation linearized prior to $C$ is no longer active,
then the interpreted bit of $t$ is 0 in $C$.
\end{lemma}

\begin{proof}
Since $U_t \cap S = \emptyset$, it follows by definition that the latest update operation for all keys in $U_t$ are \textsc{TrieDelete} operations.
For any  $x \in U_t$, 
when there has been no $S$-modifying \textsc{TrieDelete}$(x)$ operation linearized, the latest update operation with key $x$ is the dummy \textsc{TrieDelete}$(x)$ operation. By definition, it does not have a potential update to any update node in any configuration.
If the last $S$-modifying \textsc{TrieDelete}$(x)$ operation linearized prior to $C$ is no longer active, then by definition it does  not have a potential update to $t$ in $C$.
It follows that the latest update operation for all keys in $U_t$ do not have a potential update to $t$ in $C$. By Lemma~\ref{lemma:ib0_technical}, the interpreted bit of $t$ is 0 in $C$.
\end{proof}

\subsubsection{Correctness of RelaxedPredecessor}

In this section, we prove that the output of \textsc{RelaxedPredecessor} satisfies the specification outlined in Section~\ref{section_relaxed_binary_trie_specification}.
Let $pOp$ be a completed instance of \textsc{RelaxedPredecessor}$(y)$. Let $k$ be the largest key that is completely present throughout $pOp$ that is less than $y$, or $-1$ if no such key exists.

\begin{lemma}\label{lemma:path_ones_from_k}
In all configurations during $pOp$, for each binary trie node $t$ on the path from the leaf with key $k$ to the root, $t$ has interpreted bit 1.
\end{lemma}

\begin{proof}
Since $k$ is completely present throughout $pOp$, 
for each binary trie node $t$ on the path from the leaf with key $k$ to the root, $k \in U_t \cap S$ in all configurations during $pOp$.
Moreover, the last $S$-modifying \textsc{TrieInsert}$(k)$ operation linearized prior to the end of $pOp$ is 
not concurrent with $pOp$.
Hence, in all configurations during $pOp$, the last $S$-modifying \textsc{TrieInsert}$(k)$ operation is not active.
So, by Property IB1 (i.e. Lemma~\ref{lemma:ib1}), $t$ has interpreted bit 1 in all configurations during $pOp$.
\end{proof}

The next two lemmas prove that the specification of \textsc{RelaxedPredecessor} is satisfied. 

\begin{lemma}\label{lemma:relaxed_trie_kxy}
Suppose $pOp$ returns a key $x \in U$. Then $k \leq x <y$ and $x \in S$ in some configuration during $pOp$.
\end{lemma}

\begin{proof}
Before $pOp$ returns a key $x \in U$, it verified that the leaf with key $x$ had interpreted bit 1 by reading $\latest[x]$. In the configuration immediately after this read, $\latest[x]$ points to an INS node, so $x \in S$.

Since $pOp$ begins its upward traversal starting at the leaf with key $y$ and then performs a downward traversal starting from the left child of a binary trie node $t$ on this path and ending at the leaf with key $x$, it follows that $x < y$. 

Consider the path of binary trie nodes from the leaf with key $y$ to $t.\mathit{right}$. 
Any node that is not on this path and is a left child of a node on this path has
interpreted bit 0 when encountered by $pOp$ during its upward traversal.
By Lemma~\ref{lemma:path_ones_from_k}, each node on the path from the leaf with key $k$ to $t$ has interpreted bit 1.
If $k$ is in the left subtrie of a proper ancestor of $t$, then $k < x$.
Otherwise $k$ is in the left subtrie of $t$.
Since $pOp$ traverses the right-most path of binary trie nodes with interpreted bit 1 starting from $t.\mathit{left}$, $pOp$ reaches a leaf with key at least $k$. Therefore, $k \leq x$.
\end{proof}


\begin{lemma}\label{lemma_relaxed_trie_k}


If $pOp$ returns $\bot$, then there exists a key $x$, where $k < x < y$, such that the last $S$-modifying update operation with key $x$ linearized prior to the end of $pOp$ is concurrent with $pOp$.
\end{lemma}

\begin{proof}
Assume that, for all $k < x < y$, the $S$-modifying update operation with key $x$ that was last linearized prior to the end of $pOp$ 
is not concurrent with $pOp$. 
We will prove that $pOp$ returns $k \neq \bot$.

By definition of $k$, 
there are no keys greater than $k$ and less than $y$ that are completely present throughout $pOp$.
By assumption, it follows that, throughout $pOp$,
there are no keys in $S$ that are greater than $k$ and smaller than $y$.

First suppose $k = -1$. 
Recall that $pOp$  begins by traversing up the relaxed binary trie starting from the leaf with key $y$.
Consider any node on this path whose left child $t$ is not on this path.
Every key in $U_t$ is less than $y$, so $U_t \cap S = \emptyset$.
By Property IB0 (i.e. Lemma~\ref{lemma:ib0}), $t$ has interpreted bit 0 in all configurations during $pOp$. It follows that the while-loop on line~\ref{ln:traverseTrie:whileUp} always evaluates to \textsc{True}. So $pOp$ eventually reaches the root, and
returns $-1$ on line~\ref{ln:traverseTrie:return_minus}.

Now suppose $k \in U$. 
Consider any binary trie node $t$ such that $k < \min U_t < \max U_t < y$. Since  $U_t \cap S = \emptyset$, it follows from Property IB0 (i.e. Lemma~\ref{lemma:ib0}) that $t$ has interpreted bit 0 in all configurations during $pOp$.

Let $t$ be the lowest common ancestor of the leaf with key $k$ and the leaf with key $y$.
Then the leaf with key $k$ is in the subtree rooted at 
$t.\textit{left}$
and the leaf with key $y$ is in the subtree rooted at 
$t.\textit{right}$.
Consider any node on the path from $y$ to
$t.\textit{right}$ whose left child $t'$ is not on this path.
Note that $U_{t'} \cap S = \emptyset$, since $k < \min U_{t'} < \max U_{t'} < y$.
By Property IB0, $t'$ has interpreted bit 0 in all configurations during $pOp$. 
It follows from Lemma~\ref{lemma:path_ones_from_k} that $pOp$ reaches $t$ and  traverses down the right-most path of binary trie nodes with interpreted bit 1 starting from $t.\textit{left}$.

Consider any node on the path from the leaf with key $k$ to
$t.\textit{left}$ whose right child $t'$ is not on this path. Note that $U_{t'} \cap S = \emptyset$, since $k < \min U_{t'} < \max U_{t'} < y$.
By Property IB0 (i.e. Lemma~\ref{lemma:ib0}), $t'$ has interpreted bit 0 in all configurations during $pOp$. 
By Lemma~\ref{lemma:path_ones_from_k}, each binary trie node on the path from the leaf with key $k$ to $t$ has interpreted bit 1 throughout $pOp$, so $pOp$ reaches the leaf with key $k$ and $pOp$ returns $k$. 
\end{proof}

\section{Lock-free Binary Trie}\label{section_implementation}

In this section, we give the implementation of the linearizable, lock-free binary trie, which uses the relaxed binary trie as one of its components. 	This implementation supports a linearizable \textsc{Predecessor} operation, unlike the \textsc{RelaxedPredecessor} operation of the relaxed binary trie.

In Section~\ref{section_full_impl_high_level}, we give the high-level description of our algorithms.  We then describe the algorithms in detail and present the pseudocode in Section~\ref{section_detailed_algorithm}. We then prove that the implementation is linearizable in Section~\ref{section_predecessor_linearization}.







\subsection{High-level Algorithm Description}\label{section_full_impl_high_level}

One component of our lock-free binary trie is a relaxed binary trie, described in Section~\ref{section_relaxed_binary_trie}.
The other components are linked lists, which enable \textsc{Insert}, \textsc{Delete}, and \textsc{Predecessor} operations to help one another make progress. 

To ensure update operations are announced at the same time as they are linearized, each update node has a status.
It is initially inactive and can later change to active by the operation that created it, at which point the update node is \textit{activated}.
An $S$-modifying update operation will be linearized when the status of its update node changes to active.
For each key $x \in U$, $\latest[x]$ points to an update node, $\uNode$, with key $x$ whose $\latestNext$ field is either $\bot$, or points to an update node whose $\latestNext$ field is $\bot$.
We let the \textit{$\latest[x]$ list} be the sequence of length at most 2 consisting of $\uNode$, followed by the update node pointed to by $\uNode.\latestNext$ if it is not $\bot$.
When an update node with key $x$ is added to this list, it is inactive and it is only added to the beginning of the list. 
This list contains at least 1 activated update node
and only its first update node can be inactive.
The sequence of update nodes pointed to by $\latest[x]$ in an execution is the history of  $S$-modifying \textsc{TrieInsert}$(x)$ and \textsc{TrieDelete}$(x)$ operations performed.
So the types of the update nodes added to the $\latest[x]$ list alternate between INS and DEL.
The first activated update node in the $\latest[x]$ list is an INS node if and only if $x \in S$. The interpreted bit of a binary trie node that stores key $x$ now depends on the first activated update node in the $\latest[x]$ list.



The \textit{update announcement linked list}, called the $\UALL$,
is a lock-free linked list of update nodes sorted by key. 
An update operation can add an inactive update node to the $\UALL$,
and is added after every update node with the same key.
Then it can announce itself by activating its update node. Just before the update operation completes, it removes its update node from the $\UALL$.
Whenever update nodes are added and removed from the $\UALL$, we additionally modify a lock-free linked list called the \textit{reverse update announcement linked list} or $\RUALL$. It contains a copy of all update nodes in the $\UALL$, except it is sorted by keys in descending order and then by the order in which they were added. 
For simplicity, we assume that both the $\UALL$ and $\RUALL$ contain two sentinel nodes with keys $\infty$ and $-\infty$. So $\UALL.\head$ always points to the sentinel node with key $\infty$ and $\RUALL.\head$ always points to the sentinel node with key $-\infty$.

The \textit{predecessor announcement linked list}, or the $\PALL$, is an unsorted lock-free linked list of \textit{predecessor nodes}. 
Each predecessor node contains a key and an insert-only linked list, called its $\notifyList$.
An update operation can notify a predecessor operation by adding a \textit{notify node} to the beginning of the $\notifyList$ of the predecessor operation's predecessor node.
Each \textsc{Predecessor} operation begins by creating a predecessor node and then announces itself by adding this predecessor node to the beginning of the $\PALL$. Just before a \textsc{Predecessor} operation completes, it removes its predecessor node from the $\PALL$.

Figure~\ref{figure:trie} shows an example of the data structure for $U = \{0,1,2,3\}$. White circles represent nodes of the relaxed binary trie. Blue rectangles represent activated INS nodes, red rectangles represent activated DEL nodes, and light red rectangles represent inactive DEL nodes. Yellow diamonds represent predecessor nodes.
This example depicts 5 concurrent operations, \textsc{Insert}$(0)$, \textsc{Insert}$(1)$, \textsc{Delete}$(3)$, and two \textsc{Predecessor} operations.
The data structure represents the set $S = \{0, 1, 3\}$ because the first activated update node in each of the $\latest[0]$, $\latest[1]$, and $\latest[3]$ lists is an INS node. 

\begin{figure}[bt]
	\centering
	\includegraphics[width=0.6\textwidth]{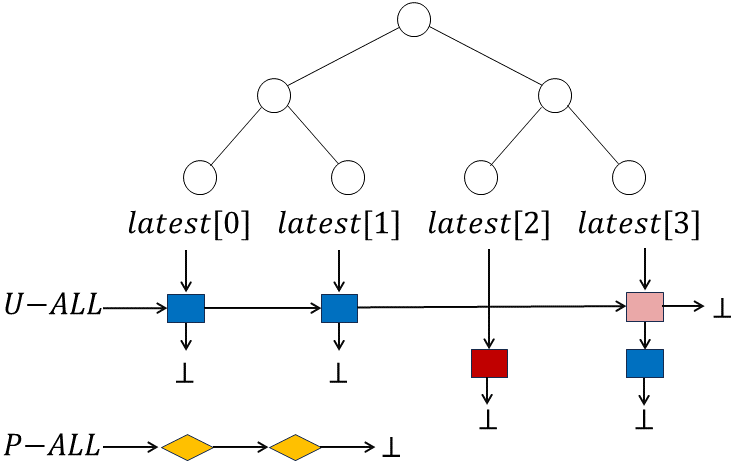}
	\caption{An example of the lock-free binary trie representing $S = \{0, 1, 3\}$.}
	\label{figure:trie}
\end{figure}



\subsubsection*{Search Operations}
A \textsc{Search}$(x)$ operation finds the first activated update node in the $\latest[x]$ list,
returns \textsc{True} if it is an INS node, and returns \textsc{False} if it is a DEL node.

\subsubsection*{Insert and Delete Operations}

An \textsc{Insert}$(x)$ or \textsc{Delete}$(x)$ operation, $uOp$, is similar to an update operation of the relaxed binary trie, with a few modifications.

The operation $uOp$ begins by finding the first activated update node in the $\latest[x]$ list. 
If it has the same type as $uOp$, then $uOp$ can return because $S$ does not need to be modified. Otherwise $uOp$ creates a new inactive update node, $\uNode$, with key $x$ and attempts to add it to the beginning of the $\latest[x]$ list.
If successful, $\uNode$ is then added to the $\UALL$ and $\RUALL$. Next, $uOp$ changes the status of $\uNode$ from inactive to active, which announces $uOp$. Additionally, $uOp$ is linearized at this step.
The field $\uNode.\latestNext$ is set to $\bot$, allowing another update node to be added to the $\latest[x]$ list.
If multiple update operations with key $x$ concurrently attempt to add an update node to the beginning of the $\latest[x]$ list, exactly one will succeed. 
Update operations that are unsuccessful instead help the update operation that succeeded become linearized. Inserting into the $\UALL$ and $\RUALL$ has an amortized cost of $O(\dot{c}(op))$ because their lengths are at most $\dot{c}(op)$.

Another modification is to notify predecessor operations after the relaxed binary trie is updated.
For each predecessor node $\pNode$ in the $\PALL$, $uOp$ creates a notify node containing information about its update node and adds it to the beginning of $\pNode$'s $\notifyList$, provided the update node created by $uOp$ is still the first activated update node in the $\latest[x]$ list. After notifying the predecessor operations announced in the $\PALL$, $uOp$ removes its update node from the $\UALL$ and $\RUALL$ before returning. 

There are at most $\dot{c}(uOp)$ predecessor nodes in the $\PALL$ when $uOp$ is invoked. Adding a notify node to the beginning of a $\notifyList$ has an amortized cost of $O(\dot{c}(uOp))$. So the total amortized cost charged to $uOp$ for notifying these predecessor nodes is $O(\dot{c}(uOp)^2)$. If a predecessor node $\pNode$ is added to the $\PALL$ after the start of $uOp$, the predecessor operation $pOp$ that created $\pNode$ pays the amortized cost of notifying $\pNode$ on $uOp$'s behalf. Since there are  $O(\dot{c}(pOp))$ update operations concurrent with $pOp$ when it is invoked, the total amortized cost charged to $pOp$ is $O(\dot{c}(pOp)^2)$.

When $uOp$ is a \textsc{Delete}$(x)$ operation, it also performs two embedded \textsc{Predecessor}$(x)$ operations, one just before $uOp$ is announced and one just before $uOp$ begins to update the relaxed binary trie.
The announcement of these embedded predecessor operations remain in the $\PALL$ until just before $uOp$ returns. 
Pointers to the predecessor nodes of these embedded \textsc{Predecessor}$(x)$ operations and their return values are stored in $\uNode$. 
This information is used by other \textsc{Predecessor} operations.

\subsubsection*{Predecessor Operations}

A \textsc{Predecessor}$(y)$ operation, $pOp$, begins by adding a predecessor node to the beginning of the $\PALL$ so that it can be notified by update operations. 
It continues to traverses the $\PALL$ to determine embedded predecessor operations belonging to \textsc{Delete} operations that have not yet completed.
It then traverses the $\RUALL$ to identify DEL nodes corresponding to \textsc{Delete} operations that may have been linearized before the start of $pOp$.

Following this, $pOp$ determines a number of \textit{candidate return values}, which are keys in $S$ sometime during $pOp$, and returns the largest of these. The exact properties satisfied by the candidate return values are stated in Section~\ref{section_crv_properties}.
Some candidate return values are determined from a traversal of the relaxed binary trie, a traversal of the $\UALL$, and a traversal of its notify list.
The linearization point of $pOp$ depends on when $pOp$ encountered the value it eventually returns.
The amortized cost for $pOp$ to perform these traversals is $O(\bar{c}(pOp)) = O(\dot{c}(pOp))$.

When $pOp$ encounters an update node with key $x < y$ during its traversal of the $\UALL$ or its notify list, 
it verifies that the update node was the first activated update node in $\latest[x]$ at some point during $pOp$.
If a verified update node is an INS node, then $x \in S$ sometime during $pOp$ and $x$ is a candidate return value.
If it is a DEL node and the \textsc{Delete}$(x)$ operation that created it is linearized  during $pOp$, then $x \in S$ immediately before this linearization point and $x$ is a candidate return value.
The traversal of the $\RUALL$ is used to identify \textsc{Delete} operations that may have been linearized before the start of $pOp$. The keys of these operations are not added to the set of candidate return values during $pOp$'s traversal of the $\UALL$ and its notify list.


If the result returned by traversing the relaxed binary trie using \textsc{RelaxedPredecessor}$(y)$ is not $\bot$, then it is a candidate return value. 
Now suppose \textsc{RelaxedPredecessor}$(y)$ returns $\bot$.
Let $k$ be the largest key less than $y$ that is completely present throughout $pOp$'s traversal of the relaxed binary trie, or $-1$ if no such key exists. 
By the specification of the \textsc{RelaxedPredecessor}$(y)$ when it returns $\bot$ (Lemma~\ref{lemma_relaxed_trie_k}), there is an $S$-modifying update operation $uOp$ with key $x$, where $k < x < y$, whose update to the relaxed binary trie is concurrent with $pOp$'s traversal of the relaxed binary trie.
The update node created by $uOp$ is encountered by $pOp$ either in the $\UALL$ or in its notify list.
This is because either $pOp$ will traverse the $\UALL$ before $uOp$ can remove its update node from the $\UALL$, or $uOp$ will notify $pOp$ before $pOp$ removes its predecessor node from the $\PALL$. 
Unless $uOp$ is a \textsc{Delete}$(x)$ operation linearized before the start of $pOp$ (or shortly after the start of $pOp$), $x$ is a candidate return value.

Now suppose $uOp$ is a \textsc{Delete}$(x)$ operation linearized before the start of $pOp$.
For simplicity, suppose $uOp$ is the only update operation concurrent with $pOp$. 
Since $uOp$ is concurrent with $pOp$'s traversal of the relaxed binary trie, its DEL node is the only update node that $pOp$ encounters during its traversal of the $\RUALL$.
Let $pOp'$ be the first embedded \textsc{Predecessor}$(x)$ of $uOp$, which was completed before $uOp$ was announced in the $\RUALL$.
The result returned by $pOp'$ may be added as a candidate return value for $pOp$. In addition, $pOp$ traverses the notify list of $pOp'$ to possibly obtain other candidate return values.
Note that $k$ is the predecessor of $y$ throughout $pOp$, so $pOp$ must return $k$. 
Let $iOp$ be the completed \textsc{Insert}$(k)$ operation that last added $k$ to $S$ prior to the start of $pOp$.
First, suppose $iOp$ is linearized after $pOp'$ was announced. Then $iOp$ will notify $pOp'$ because $iOp$ is completed before the start of $pOp$.
When $pOp$ traverses the notify list of $pOp'$, it adds $k$ to its set of candidate return values.

Now suppose $iOp$ is linearized before $pOp'$ was announced. Then $k \in S$ throughout $pOp'$, and $pOp'$ returns a value $k'$ where $k \leq k' < x$. If $k' = k$, then $k$ is added to $pOp$'s candidate return values.
If $k' \neq k$, then $k'$ is removed from $S$ by a \textsc{Delete}$(k')$ operation prior to the start of $pOp$ because $k$ is the predecessor of $y$ at the start of $pOp$.
More generally, consider any \textsc{Delete} operation, $dOp$, with key strictly between $k$ and $y$ that is linearized after $pOp'$ is announced and is completed before the start of $pOp$. In particular, $dOp$'s 
second embedded \textsc{Predecessor} is completed before the start of $pOp$ and returns a value which is at least $k$. Before $dOp$ completed, it notified $pOp'$ of this result. 
We prove that there is a \textsc{Delete} operation that notifies $pOp'$ and whose second embedded \textsc{Predecessor} returns a value exactly $k$. 
By traversing the notify list of $pOp'$, $pOp$ can determine the largest key less than $y$ that is in $S$ at the start of $pOp$'s traversal of the relaxed binary trie.
So $k$ is added to $pOp$'s set of candidate return values.
The cost for $pOp$ to traverse the notify list of $pOp'$ is $O(\tilde{c}(pOp))$.

\subsection{Detailed Algorithm and Pseudocode}\label{section_detailed_algorithm}

In this section, we present a detailed description of \textsc{Insert}, \textsc{Delete} and \textsc{Predecessor}, as well as present their pseudocode.

Figure~\ref{data_records} gives a summary of the fields used by each type of node, and classifies each field as immutable, update-once, or mutable.
An \textit{immutable field} is set when the field is initialized and is never changed. 
A \textit{mutable field} may change its value an arbitrary number of times.
The possible transitions of other fields is specified.

\begin{figure}[h!]
	\begin{algorithmic}[1]
		\alglinenoPop{alg1}
		\State \textbf{Update Node}
		\Indent
		\State $\key$ (Immutable) \Comment{A key in $U$}
		\State $\type$ (Immutable) \Comment{Either INS or DEL}
		\State $\status$ (Initially \textsc{Inactive}, only changes from \textsc{Inactive} to \textsc{Active}) \Comment{One of $\{\textsc{Inactive}, \textsc{Active}\}$}
		\State $\latestNext$ (Initialized to point to an update node. Changes once to $\bot$)\label{ln:init:latestNext}
		\State $\target$ (Mutable, initially $\bot$) \Comment{pointer to update node}
		\State $\stopflag$ (From 0 to 1) \Comment(Boolean value)
		\State $\done$ (From 0 to 1) \Comment(Boolean value)
		
		\State $\triangleright$ Additional fields when $\type = \text{DEL}$
		\State $\threshold$ (Mutable, initially $0$)\label{ln:init:threshold} \Comment{An integer in $\{0,\dots,b\}$}
		\State $\insthreshold$ (Mutable min-register, initially $b+1$) \Comment{An integer in $\{0,\dots,b+1\}$}
		\State $\delprednode$  (Immutable) \Comment{A pointer to a predecessor node}
		\State $\delpred$ (Immutable) \Comment{A key in $U$}
		\State $\delpredsecond$ (From $\bot$ to a key in $U$) \Comment{A key in $U$}
		\EndIndent
		\State \textbf{Predecessor Node} of $\PALL$
		\Indent
		\State $\key$ (Immutable) \Comment{A key in $U$}
		\State $\notifyList$ (Mutable, initially an empty linked list)  \Comment{A linked list of notify nodes}
		\State $\RuallPosition$ (Mutable, initially a pointer to sentinel update node with key $\infty$) 
		\EndIndent
		
		\State \textbf{Notify Node} of $\pNode.\notifyList$ for each $\pNode \in \PALL$
		\Indent
		\State $\key$ (Immutable) \Comment{A key in $U$}
		\State $\updateNode$ (Immutable) \Comment{A pointer to an update node}
		\State $\updateNodeMax$ (Immutable) \Comment{A pointer to an update node}
		\State $\notifyThreshold$ (Immutable) \Comment{A key in $U$}
		\EndIndent

		\State \textbf{Binary Trie Node}
		\Indent
		\State $\dNodePtr$ (Mutuable, initially points to a dummy DEL node) \Comment{A pointer to a DEL node}
		\EndIndent
		
		\alglinenoPush{alg1}
	\end{algorithmic}
	\caption{Summary of the fields and initial values of each node used by the data structure.}
	\label{data_records}
\end{figure}

\subsubsection*{Search Operations}

The \textsc{Search}$(x)$ algorithm finds the first activated update node in the $\latest[x]$ list by calling \textsc{FindLatest}$(x)$. 
It returns \textsc{True} if this update node has type INS, and \textsc{False} if this update node has type DEL.

\begin{figure}[h!]
	\begin{algorithmic}[1]
		\alglinenoPop{alg1}
		\State \textbf{Algorithm} \textsc{FindLatest}$(x)$
		\Indent
		\State $\uNode \leftarrow \latest[x]$\label{ln:findLatest:read_head}
		\If {($\uNode.\status = \textsc{Inactive}$)} \label{ln:findLatest:if}
		\State $\uNode' \leftarrow \uNode.\latestNext$ \label{ln:findLatest:read_latestNext}
		\If {$\uNode' \neq \bot$} \Return $\uNode'$ \label{ln:findLatest:if2} 
		\EndIf
		\EndIf
		\Return $\uNode$ \label{ln:findLatest:return}
		\EndIndent
		\alglinenoPush{alg1}
	\end{algorithmic}
\end{figure}

\begin{figure}[h!]
	\begin{algorithmic}[1]
		\alglinenoPop{alg1}
		\State \textbf{Algorithm} \textsc{Search}$(x)$
		\Indent 
		\State $\uNode \leftarrow \textsc{FindLatest}(x)$\label{ln:search:findLatest}
		\If {$\uNode.\type = \textsc{INS}$} \Return \textsc{True}
		\Else ~\Return \textsc{False}
		\EndIf 
		\EndIndent
		\alglinenoPush{alg1}
	\end{algorithmic}
\end{figure}

\begin{figure}[h!]
	\begin{algorithmic}[1]
		\alglinenoPop{alg1}
		\State \textbf{Algorithm} \textsc{FirstActivated}$(\uNode)$
		\Indent 
		\State $\uNode' \leftarrow \latest[\uNode.\key]$ \label{ln:firstactivated:head}
		\State \Return $\uNode = \uNode'$ or ($\uNode'.status = \textsc{Inactive}$ and $\uNode = \uNode'.\latestNext$)\label{ln:firstactivated:return}
		\EndIndent
		\alglinenoPush{alg1}
	\end{algorithmic}
\end{figure}

The helper function \textsc{FindLatest}$(x)$ returns the first activated update node in the $\latest[x]$ list.
It 
first reads the update node $\uNode$ pointed to by $\latest[x]$.
If $\uNode$ is inactive, then the update node, $\uNode'$, pointed to by its next pointer is read. If $\uNode'$ is $\bot$, then $\uNode$ was activated sometime between when $\uNode$ was read to be inactive and when its next pointer was read. If $\uNode' \neq \bot$, then $\uNode'$ was active when $\uNode$ was read to be inactive.
If \textsc{FindLatest}$(x)$ returns an update node $\uNode$, we prove that there is a configuration during the instance of \textsc{FindLatest} in which $\uNode$ is the first activated update node in the $\latest[x]$ list.

The helper function \textsc{FirstActivated}$(\uNode)$ takes a pointer to an activated update node, $\uNode$, and checks if $\uNode$ is the first activated update node in the $\latest[\uNode.key]$ list. 
It first reads the update node, $\uNode'$, pointed to by $\latest[x]$. If $\uNode' = \uNode$, then the algorithm returns \textsc{True} because $\uNode$ is the first activated update node in the $\latest[\uNode.key]$ list.
If $\uNode'$ is inactive when it is read and its $\latestNext$ pointer points to $\uNode$, then $\uNode$ is the first activated update node in the $\latest[\uNode.\key]$ list.

\subsubsection*{Insert Operations}

We next describe the algorithm for an \textsc{Insert}$(x)$ operation $iOp$. It is roughly divided into the main parts of inserting a new INS node, $\iNode$, into $\latest[x]$, adding $\iNode$ to $\UALL$, updating the relaxed binary trie, notifying predecessor operations, and removing $\iNode$ from the $\UALL$.

It begins by finding the first activated update node in the $\latest[x]$ list (on line~\ref{ln:insert:find_latest}). If this is an INS node then $iOp$ returns, because $x$ is already in $S$.
Otherwise it is a DEL node, $\dNode$. 
A new inactive update node, $\iNode$ is created, where $\iNode.\latestNext$ points to $\dNode$, and CAS is used to try to change $\latest[x]$ to point from $\dNode$ to $\iNode$.
If the CAS is unsuccessful, then some other \textsc{Insert}$(x)$ operation, $iOp'$, successfully updated $\latest[x]$ to point to $\iNode'$.
In this case, $iOp$ invokes \textsc{HelpActivate}$(\iNode')$ to help activate $\iNode'$, and hence linearize $iOp'$.
First, $iOp$ helps $iOp'$ add $\iNode'$ to the $\UALL$, and changes the status of $\iNode'$ from \textsc{Inactive} to \textsc{Active}. 
Recall from the relaxed binary trie that \textsc{Delete} operations help stop the target of the previously linearized $S$-modifying \textsc{Insert} operation (i.e. line~\ref{ln2:delete:setStop} from the relaxed binary trie).
This is done when helping on line~\ref{ln:help_activate:help_stop}.
On line~\ref{ln:help_activate:set_bot}, $iOp$ sets $\iNode'.\latestNext$ to $\bot$, allowing another update node to be added to the $\latest[x]$ list.
Then $iOp$ checks if $\iNode'.\done$ is set to \textsc{True} (on line~\ref{ln:help_activate:check_remove}). This indicates that  $iOp'$ has completed updating the relaxed binary trie, and $\iNode'$ no longer needs to be in $\UALL$ or $\RUALL$. It is possible that $iOp'$ already removed $\iNode'$ from the $\UALL$ and $\RUALL$, but $iOp$ helped add it back sometime it was removed. So $iOp$ must remove it again from the $\UALL$ and $\RUALL$.
After $iOp$ removes $\iNode'$ from the $\UALL$ and $\RUALL$, it returns.

\begin{figure}[h!]
	\begin{algorithmic}[1]
		\alglinenoPop{alg1}
		\State \textbf{Algorithm} \textsc{HelpActivate}$(\uNode)$
		\Indent
		\If {$\uNode.status = \textsc{Inactive}$}
		\State Insert $\uNode$ into $\UALL$  and $\RUALL$
		\State $\uNode.\status \leftarrow \textsc{Active}$\label{ln:help_activate:CAS}
		\If {$\uNode.\type = \text{DEL}$}
		\State $\uNode.\latestNext.\target.\stopflag \leftarrow \textsc{True}$ \Comment{Ignore if any field reads $\bot$}\label{ln:help_activate:help_stop}
		\EndIf
		\State $\uNode.\latestNext \leftarrow \bot$\label{ln:help_activate:set_bot}
		\If {$\uNode.\done = \textsc{True}$}\label{ln:help_activate:check_remove} \Comment{$\uNode$ no longer needed in $\UALL$ or $\RUALL$}
		\State Remove $\uNode$ from $\UALL$ and $\RUALL$
		\EndIf 
		\EndIf
		\EndIndent
		\alglinenoPush{alg1}
	\end{algorithmic}
\end{figure}

\begin{figure}[h!]
	\begin{algorithmic}[1]
		\alglinenoPop{alg1}
		\State \textbf{Algorithm} \textsc{TraverseUall}$(x)$
		\Indent
		\State Initialize local variables $I \leftarrow \emptyset$ and  $D \leftarrow \emptyset$
		\State $\uNode \leftarrow \UALL.\head$
		\While {$\uNode \neq \bot$ and $\uNode.\key < x$}\label{ln:traverseUall:while}
		\If {($\uNode.\status \neq \textsc{Inactive}$ and $\textsc{FirstActivated}(\uNode)$)} \label{ln:traverseUALL:findLatest}
		\If {$\uNode.type = \textsc{INS}$} $I \leftarrow I \cup \{\uNode\}$
		\Else  ~$D \leftarrow D \cup \{\uNode\}$
		\EndIf
		\EndIf
		\State $\uNode \leftarrow \uNode.next$
		\EndWhile
		\State \Return $I, D$
		\EndIndent
		\alglinenoPush{alg1}
	\end{algorithmic}
\end{figure}

\begin{figure}[h!]
	\begin{algorithmic}[1]
		\alglinenoPop{alg1}
		\State \textbf{Algorithm} \textsc{NotifyPredOps}$(\uNode)$
		\Indent
		\State $I,D \leftarrow \textsc{TraverseUall}(\infty)$\label{ln:notify:travuall}
		\For {each node $\pNode \in \PALL$}\label{ln:notify:for_loop}
		
		
		\If {$\textsc{FirstActivated}(\uNode) = \textsc{False}$} \Return \label{ln:notify:first_activated} 
		\EndIf
		
		\State Create a notify node $\nNode$:
		\State \quad $\nNode.\key \leftarrow \uNode.\key$
		\State \quad $\nNode.\updateNode \leftarrow \uNode$
		\State \quad $\nNode.\updateNodeMax \leftarrow $ INS node in $I$ with largest key less than $\pNode.\key$
		\State \quad $\nNode.\notifyThreshold \leftarrow \pNode.\RuallPosition.\key$ \label{ln:notify:notifyThreshold}
		\If {$\textsc{SendNotification}(\nNode, \pNode) = \textsc{False}$} \Return
		\EndIf
		\EndFor
		\EndIndent
		\alglinenoPush{alg1}
	\end{algorithmic}
\end{figure}

\begin{figure}[h!]
	\begin{algorithmic}[1]
		\alglinenoPop{alg1}
		\State \textbf{Algorithm} \textsc{SendNotification}$(\mathit{nNodeNew},\pNode)$
		\Indent
		\While {$\textsc{True}$}
		\State $\nNode \leftarrow \pNode.\notifyList.\head$\label{ln:sendNotification:readHead}
		\State $\mathit{nNodeNew}.\nextptr \leftarrow \nNode$ 
		\If {\textsc{FirstActivated}$(\mathit{nNodeNew}.\updateNode) = \textsc{False}$} \Return \textsc{False}\label{ln:sendNotification:firstActivated}
		\EndIf
		\If {CAS$(\pNode.\notifyList.\head,  \nNode, \mathit{nNodeNew}) = \textsc{True}$} \Return \textsc{True}\label{ln:sendNotification:CAS}
		\EndIf
		\EndWhile
		\EndIndent
		\alglinenoPush{alg1}
	\end{algorithmic}
\end{figure}

Otherwise the CAS is successful, and $iOp$ inserts $\iNode$ into the $\UALL$ and $\RUALL$. The status of $\iNode$ is changed from \textsc{Inactive} to \textsc{Active}, which announces the operation. The operation $iOp$ is linearized during the write that changes $\iNode.\status$ from \textsc{Inactive} to \textsc{Active}, which may be performed by $iOp$ or an \textsc{Insert}$(x)$ operation helping $iOp$. It then performs \textsc{InsertBinaryTrie} (on line~\ref{ln:insert:binaryTrie}), which is the same as in the relaxed binary trie, to update the interpreted bits of the binary trie.

\begin{figure}[h!]
	\begin{algorithmic}[1]
		\alglinenoPop{alg1}
		\State \textbf{Algorithm} \textsc{Insert}$(x)$
		\Indent
		\State $\dNode \leftarrow \textsc{FindLatest}(x)$\label{ln:insert:find_latest}
		\If {$\dNode.\type \neq \text{DEL}$} \Return \Comment{$x$ is already in $S$}\label{ln:insert:return_early}
		\EndIf
		
		\State Let $\iNode$ be a pointer to a new update node: 
		\State \quad $\iNode.\key \leftarrow x$, $\iNode.\type \leftarrow \text{INS}$
		\State \quad $\iNode.\latestNext \leftarrow \dNode$
		\State $\dNode.\latestNext.\target.\stopflag \leftarrow \textsc{True}$ \Comment{Ignore if any field reads $\bot$}
		\State $\dNode.\latestNext \leftarrow \bot$\label{ln:insert:set_latestNext_bot}
		\If {CAS$(\latest[x], \dNode, \iNode) = \textsc{False}$}\label{ln:insert:cas_latest_head}
		\State \textsc{HelpActivate}$(\latest[x])$\label{ln:insert:help_activate}
		\State \Return\label{ln:insert:return_help}
		\EndIf
		
		
		\State Insert $\iNode$ into $\UALL$ and $\RUALL$\label{ln:insert:add_uall}
		\State $\iNode.\status \leftarrow \textsc{Active}$ \label{ln:insert:set_active}
		\State $\iNode.\latestNext \leftarrow \bot$\label{ln:insert:set_bot}
		\State \textsc{InserstBinaryTrie}$(\iNode)$\label{ln:insert:binaryTrie}
		\State $\textsc{NotifyPredOps}(\iNode)$ \label{ln:insert:notify}
		\State $\iNode.\done \leftarrow \textsc{True}$
		\State Remove $\iNode$ from $\UALL$  and $\RUALL$\label{ln:insert:removeUall}
		\State \Return
		\EndIndent
		\alglinenoPush{alg1}
	\end{algorithmic}
\end{figure}




The helper function \textsc{NotifyPredOps}$(\iNode)$ attempts to notify to all predecessor nodes in $\PALL$ about $\iNode$ as follows. 
First, $iOp$ invokes \textsc{TraverseUall}$(x)$ (on line~\ref{ln:notify:travuall}). It finds update nodes with key less than $x$ that are the first activated update node in their respective latest lists.
INS nodes are put into $iOp$'s local set $I$ and DEL nodes into $iOp$'s local set $D$.
It simply traverses the $\UALL$, checking if each update node $\uNode$ visited has key less than $x$ and is the first activated update node in the $\latest[\uNode.\key]$ list. If so, $\uNode$ is added to $I$ if it is an INS node, or it is added to $D$ if it is a DEL node.
We prove that the keys of the nodes in $I$ are in the set $S$ sometime during the traversal, while the keys in $D$ are not in the set $S$ sometime during the traversal.

Then $iOp$ notifies each predecessor node $\pNode$ in $\PALL$ with $\pNode.\key > x$. 
It creates a new notify node, $\nNode$, containing a pointer to $iOp$'s update node
and a pointer to the INS node in $I$ with the largest key less than $\pNode.\key$. It also reads the value stored in $\pNode.\RuallPosition$, which is a pointer to an update node.
The key of this update node is written into $\nNode.\notifyThreshold$ (on line~\ref{ln:notify:notifyThreshold}).
This information is used by the predecessor operation to determine if the notification sent by $iOp$ should be used to determine a candidate return value.
Prior to adding $\nNode$ to the head of $\pNode.\notifyList$ using CAS, $iOp$ checks that $\iNode$ is still the first activated update node in the $\latest[x]$ list (on line~\ref{ln:sendNotification:firstActivated}). If not, $iOp$ stops sending notifications.

Before $iOp$ returns from its \textsc{Insert}$(x)$ operation, $\iNode$ is removed from the $\UALL$ and $\RUALL$ (on line~\ref{ln:insert:removeUall}).


\subsubsection*{Delete Operations}

We next describe the algorithm for a \textsc{Delete}$(x)$ operation, $dOp$. The algorithm is similar to an \textsc{Insert}$(x)$ operation, but with a few more parts. Most importantly, $dOp$ may perform embedded predecessor operations, which we describe later. The main parts of a \textsc{Delete}$(x)$ operation are performing an embedded predecessor operation, notifying \textsc{Predecessor} operations about the previous $S$-modifying \textsc{Insert}$(x)$ operation, adding a new DEL node into the $\latest[x]$ list and then into the $\UALL$, performing a second embedded predecessor operation, updating the relaxed binary trie, and then notifying \textsc{Predecessor} operations about its own operation.

\begin{figure}[h!]
	\begin{algorithmic}[1]
		\alglinenoPop{alg1}
		\State \textbf{Algorithm} \textsc{Delete}$(x)$
		\Indent
		\State $\iNode \leftarrow \textsc{FindLatest}(x)$\label{ln:delete:findLatest}
		\If {$\iNode.\type \neq \text{INS}$} \Return \Comment{$x$ is not in $S$}\label{ln:delete:return_early}
		\EndIf
		
		\State $\delpred, \pNode1 \leftarrow \textsc{PredHelper}(x)$\label{ln:delete:first_em_pred}
		\State Let $\dNode$ be a pointer to a new update node:
		\State \quad $\dNode.\key \leftarrow x$, $\dNode.\type \leftarrow \text{DEL}$
		\State \quad $\dNode.\latestNext \leftarrow \iNode$
		\State \quad $\dNode\delpred \leftarrow \delpred$
		\State \quad $\dNode.\delprednode \leftarrow \pNode1$
		\State $\iNode.\latestNext \leftarrow \bot$\label{ln:delete:set_latestNext_bot}
		\State \textsc{NotifyPredOps}$(\iNode)$\Comment{Help previous \textsc{Insert} send notifications}\label{ln:delete:help_notify}
		
		\If {CAS$(\latest[x], \iNode, \dNode) = \textsc{False}$}\label{ln:delete:cas_latest_head}
		\State \textsc{HelpActivate}$(\latest[x])$\label{ln:delete:help_activate}
		\State Delete $\pNode1$ from $\PALL$.
		\State \Return\label{ln:delete:return_help}
		\EndIf
		
		\State Insert $\dNode$ into $\UALL$ and $\RUALL$
		\State $\dNode.\status \leftarrow \textsc{Active}$\label{ln:delete:set_active}
		
		\State $\iNode.\target.\stopflag \leftarrow \textsc{True}$\label{ln:delete:setStop}
		\State $\dNode.\latestNext \leftarrow \bot$\label{ln:delete:set_bot}
		
		\State $\delpredsecond, \pNode2 \leftarrow \textsc{PredHelper}(x)$\label{ln:delete:second_em_pred}
		\State $\dNode.\delpredsecond \leftarrow \delpredsecond$
		
		\State \textsc{DeleteBinaryTrie}$(\dNode)$\label{ln:delete:binaryTrie}
		\State \textsc{NotifyPredOps}$(\dNode)$\label{ln:delete:notifyPredOps}
		
		\State $\dNode.\done \leftarrow \textsc{True}$
		
		\State Delete $\dNode$ from $\UALL$  and $\RUALL$\label{ln:delete:remove_uall_ruall}
		\State Delete $\pNode1$ and $\pNode2$ from $\PALL$\label{ln:delete:remove_pall}
		\EndIndent
		\alglinenoPush{alg1}
	\end{algorithmic}
\end{figure}

The \textsc{Delete}$(x)$ operation, $dOp$, begins by finding the first activated update node in the $\latest[x]$ list. It immediately returns if this is a DEL node, since $x \notin S$. Otherwise, the $\latest[x]$ list begins with an INS node, $\iNode$.
Next, $dOp$ performs an embedded predecessor operation with key $x$ (on line~\ref{ln:delete:first_em_pred}).
The result of the embedded predecessor operation is saved in a new DEL node $\dNode$, along with other information. Recall that this embedded predecessor operation will be used by concurrent \textsc{Predecessor} operations in the case that $dOp$ prevents them from traversing the relaxed binary trie.

Additionally, $dOp$ performs \textsc{NotifyPredOps}$(\iNode)$ (on line~\ref{ln:delete:help_notify}) to help notify predecessor operations about $\iNode$. Since an \textsc{Insert}$(x)$ operation does not send notifications when its update node is no longer the first activated update node in the $\latest[x]$ list, we ensure at least one update operation with key $x$ notifies all predecessor nodes about $\iNode$ before a \textsc{Delete}$(x)$ operation is linearized.

Then $dOp$ attempts to add $\dNode$ to the $\latest[x]$ list, by using CAS  to change $\latest[x]$ to point from $\iNode$ to $\dNode$ (on line~\ref{ln:delete:cas_latest_head}).
If the CAS is unsuccessful, 
a concurrent \textsc{Delete}$(x)$ operation, $dOp'$, successfully updated  $\latest[x]$ to point to some DEL node, $\dNode' \neq \dNode$. So $dOp$ performs \textsc{HelpActivate}$(\dNode')$ to help $dOp'$ linearize and then $dOp$ returns.
If the CAS is successful, $dOp$ inserts $\dNode$ into the $\UALL$ and $\RUALL$ and changes $\dNode.\status$ from \textsc{Inactive} to \textsc{Active} using a write. 
This write is the linearization point of $dOp$.
As in the case of the relaxed binary trie (on line~\ref{ln2:delete:setStop}), $dOp$ helps stop the previously linearized \textsc{Insert}$(x)$ operation (on line~\ref{ln:delete:setStop}).
On line~\ref{ln:delete:set_bot}, $dOp$ sets $\dNode'.\latestNext$ to $\bot$, allowing another update node to be added to the $\latest[x]$ list.
Next, $dOp$ performs a second embedded predecessor operation (on line~\ref{ln:delete:second_em_pred}) and 
records the result in $\dNode$.

The interpreted bits of the binary trie are updated in \textsc{DeleteBinaryTrie}$(x)$, which is the same algorithm as in the implementation of the relaxed binary trie.
Once the interpreted bits have been updated, $dOp$ notifies predecessor operations using \textsc{NotifyPredOps}$(\dNode)$ described previously.
Finally, $dOp$ removes its update node $\dNode$ from the $\UALL$ and $\RUALL$ and removes the predecessor nodes it created for its embedded predecessor operations from the $\PALL$. 




\subsubsection*{Predecessor Operations}\label{section_lock_free_trie_pred}

A \textsc{Predecessor}$(y)$ operation begins by calling an instance, $pOp$, of $\textsc{PredHelper}(y)$, which does all the steps of the operation except for removing the announcement from the $\PALL$. This helper function is also used by \textsc{Delete}$(y)$ operations to perform their embedded predecessor operations. Recall that these embedded predecessor operations do not remove their announcements from the $\PALL$ until the end of their \textsc{Delete}$(y)$ operation. The helper function $\textsc{PredHelper}(y)$ is complicated, so its description 
is divided into six main parts: announcing the operation in the $\PALL$, traversing the $\RUALL$,
traversing the relaxed binary trie,
traversing the $\UALL$,
collecting notifications,
and handling the case when the traversal the relaxed binary trie returns $\bot$.
These parts are performed one after the other in this order.

\begin{figure}[htbp]
	\begin{algorithmic}[1]
		\alglinenoPop{alg1}
		\State \textbf{Algorithm} \textsc{PredHelper}$(y)$
		\Indent
		\State Create a predecessor node $\pNode$ with key $y$
		\State Insert $\pNode$ at the head of $\PALL$ \label{ln:pred:insertPALL}
		\State $\pNode' \leftarrow \pNode.\nextptr$\label{ln:pred:travPALL}
		\State  $Q \leftarrow ()$ \Comment{Initialize empty sequence}
		\While {$\pNode' \neq \bot$} 
		\State prepend $\pNode'$ to $Q$
		\State $\pNode' \leftarrow \pNode'.\nextptr$ \label{ln:pred:travPALLend}
		\EndWhile
		
		
		
		\State $(\Iruall,\Druall) \leftarrow \textsc{TraverseRUall}(\pNode)$\label{ln:pred:uall1}
		
		
		\State $r_0 \leftarrow \textsc{RelaxedPredecessor}(y)$\label{ln:pred:travtrie}
		\State $(\Iuall,\Duall) \leftarrow \textsc{TraverseUall}(y)$\label{ln:pred:uall2}
		
		\State $(\Inotify, \Dnotify) \leftarrow (\emptyset,\emptyset)$
		
		\For {each notify node $\nNode$ in $\pNode.\notifyList$ with key less than $y$} \label{ln:pred:newlist}
		
		\If {$\nNode.\updateNode.\type = \text{INS}$}
		\If {$\nNode.\notifyThreshold \leq \nNode.\key$}\label{ln:pred:notifyThresholdINS}
		\State $\Inotify \leftarrow \Inotify \cup \{ \nNode.\updateNode \}$\label{ln:pred:I2}
		\EndIf
		\Else
		\If {$\nNode.\notifyThreshold < \nNode.\key$}\label{ln:pred:notifyThreshold}
		\State $\Dnotify \leftarrow \Dnotify \cup \{ \nNode.\updateNode \}$\label{ln:pred:D2}
		\EndIf
		\EndIf
		
		\If  {$\nNode.\notifyThreshold = -\infty$ and $\nNode.\updateNode \notin (\Iruall \cup \Druall)$}\label{ln:pred:updateNodeMaxCheck}
		\State $\Inotify \leftarrow \Inotify \cup \{ \nNode.\updateNodeMax \}$\label{ln:pred:updateNodeMax}
		\EndIf
		
		\EndFor
		
		
		\State $r_1 \leftarrow \max \left(\{-1\} \cup \{\uNode.\key \mid \uNode \in \Iuall \cup \Inotify \cup (\Duall - \Druall) \cup (\Dnotify - \Druall)\} \right)$\label{ln:pred:candidate}
		
		
		
		\State $\triangleright$ Unsuccessful traversal of relaxed binary trie
		\If {$r_0 = \bot$ and $\Druall \neq \emptyset$}\label{ln:pred:fix_binary_trie_pred} 
		\State $L_1 \leftarrow ()$ \Comment{Initialize empty sequence}
		\State $\mathit{predNodes} \leftarrow \{ \pNode' \mid \exists \dNode \in \Druall \text{ where } \dNode.\delprednode = \pNode' \}$\label{ln:pred:predNodes}
		\If {$Q$ contains a predecessor node in $\mathit{predNodes}$}
		\State $\pNode' \leftarrow$ predecessor node in $\mathit{predNodes}$ that occurs earliest in sequence $Q$\label{ln:pred:pNode'}
		\For {each notify node $\nNode$ in $\mathit{pNode'}.\notifyList$ with key less than $y$}\label{ln:pred:L}
		\State prepend $\nNode.\updateNode$ to $L_1$ if not already in $L_1$\label{ln:pred:prependL'1}
		\EndFor 
		\EndIf
		\State $L_2 \leftarrow ()$  \Comment{Initialize empty sequence}
		\For {each notify node $\nNode$ in $\pNode.\notifyList$ with key less than $y$}\label{ln:pred:L'}
		\State remove $\nNode.\updateNode$ from $L_1$\label{ln:pred:removeL'1}
		\If  {$\nNode.\notifyThreshold \geq \nNode.\key$} \label{ln:pred:thresholdL'2}
		\State prepend $\nNode.\updateNode$ to $L_2$  if not already in $L_2$\label{ln:pred:prependL'2}
		\EndIf 
		\EndFor
		
		\State $L \leftarrow $ sequence of update nodes in $L_1$ followed by $L_2$\label{ln:pred:L-L'}
		\State $L \leftarrow L - \{\dNode \in L \mid \dNode.\type = \text{DEL} \text{ and is not the last update node in } L \text{ with key } \dNode.\key \}$\label{ln:pred:LremoveDEL}
		
		
		\State Let $T_L = (V,E)$ be the following directed graph (i.e. Definition~\ref{definition:T_L})\label{ln:pred:T_L}:
		\State \hspace{0.5cm} $V = \{ u.\key \mid u \in L \} \cup \{ d.\delpredsecond \mid d \in L \} \cup \{ d.\delpred \mid d \in \Druall \}$
		\State \hspace{0.5cm} $E = \{ (u,v) \mid \exists \dNode \in L \text{ where } \dNode.\key = u \text{ and  }  \dNode.\delpredsecond = v \}$
		

		\State $X \leftarrow \{ x \in U \cup \{-1\} \mid \exists \dNode \in \Druall \text{ where } \dNode.\delpred = x \}$ \label{ln:pred:X_init}
		\State $X \leftarrow X \cup \{  x \in U \mid \exists \iNode \in L \text{ where } \iNode.\type = \text{INS} \text{ and } \iNode.\key = x \}$ \label{ln:pred:X_init_INS}
		\State $R \leftarrow \{ w \text{ is a sink in } T_L \mid w \text{ is reachable from some key in } X \}$ \label{ln:pred:R_from_TL}
		
		
		\State $R \leftarrow R - \{ x \in R \mid \exists \dNode \in \Druall \text{ where } \dNode.\key = x \}$ \label{ln:pred:R_remove_Druall}
		\State $r_0 \leftarrow \max\{R\}$ \label{ln:pred:set_pred0}

		\EndIf
		\State \Return $\max\{r_0, r_1 \}$,  $\pNode$\label{ln:pred:return}
		\EndIndent
		\alglinenoPush{alg1}
	\end{algorithmic} 
\end{figure}

\begin{figure}[htbp]
	\begin{algorithmic}[1]
		\alglinenoPop{alg1}
		\State \textbf{Algorithm} \textsc{Predecessor}$(y)$
		\Indent
		\State $\pred, \pNode \leftarrow \textsc{PredHelper}(y)$\label{ln:pred:PredHelper}
		\State Remove $\pNode$ from $\PALL$\label{ln:pred:removePall}
		\State \Return $\pred$
		\EndIndent
		\alglinenoPush{alg1}
	\end{algorithmic} 
\end{figure}

\noindent\\
\textbf{Announcing the operation in the $\PALL$:}
An instance, $pOp$, of \textsc{PredHelper}$(y)$ announces itself
by creating a new predecessor node, $\pNode$, with key $y$ and inserts it at the head of $\PALL$ (on line \ref{ln:pred:insertPALL}).
Then it traverses the $\PALL$ starting from $\pNode$ (on line \ref{ln:pred:travPALL} to \ref{ln:pred:travPALLend}), locally storing the sequence of predecessor nodes it encounters into a sequence $Q$.

\noindent\\
\textbf{Traversing the $\RUALL$:}
The traversal of the $\RUALL$ is done by a call to  \textsc{TraverseRUall}$(\pNode)$ (on line~\ref{ln:pred:uall1}). During this traversal, $pOp$ identifies each update node with key less than $y$ that is the first activated update node in its latest list.
Those with type INS are put into $pOp$'s local set $\Iruall$, while those with type DEL are put into $pOp$'s local set $\Druall$.
The sets $\Iruall$ and $\Druall$ include the update nodes of all $S$-modifying update operations linearized before the start of $pOp$ and are still active at the start of $pOp$'s traversal of the relaxed binary trie.
They may additionally contain update nodes of update operations linearized shortly after the start of $pOp$, because it is difficult to distinguish them from those that were linearized before the start of $pOp$. 
Since the \textsc{Delete} operations of DEL nodes in $\Druall$ may be linearized before the start of $pOp$, they are not used to determine candidate return values.
Instead, they are used to eliminate announcements or notifications later seen by $pOp$ when it traverses the $\UALL$ or $\pNode.\notifyList$.

While $pOp$ is traversing the $\RUALL$, it makes available the key of the update node it is currently visiting in the $\RUALL$. This is done by maintaining $\pNode.\RuallPosition$, which contains a pointer to the update node in the $\RUALL$ that $pOp$ is currently visiting. Recall that the key field of an update node is immutable, so a pointer to this node is sufficient to obtain its key.
Initially, $\pNode.\RuallPosition$ points to the sentinel at the head of the $\RUALL$, which has key $\infty$.
Only $pOp$ modifies $\pNode.\RuallPosition$. Each time $pOp$ reads a pointer to the next node in the $\RUALL$, $pOp$ atomically copies this pointer into $\pNode.\RuallPosition$. 
Single-writer atomic copy can be implemented from CAS with $O(1)$ worst-case step complexity~\cite{BlellochW20}.

The pseudocode for \textsc{TraverseRUall}$(\pNode)$ is as follows. The local variable $\uNode$ is initialized to point to the sentinel node with key $\infty$ at the head of the $\RUALL$ (on line~\ref{ln:travRUALL:head}). 
Then $pOp$ traverses the $\RUALL$ one update node at a time, atomically copying $\uNode.\nextptr$ into  $\pNode.\RuallPosition$ (on line~\ref{ln:travRUALL:atomicCopy}) before progressing to the next node. It traverses the list until it first reaches an update node with key less than $y$. From this point on, it checks whether the update node it is pointing to is the first activated update node in the  $\latest[\uNode.\key]$ list (on line~\ref{ln:traverseRevUALL:findLatest}).
If so, the update node is added to $\Iruall$ or $\Druall$, depending on its type.
When $\uNode.\key = -\infty$, $\pNode.\RuallPosition$ points to the sentinel node at the end of the $\RUALL$. 

\begin{figure}[!htbp]
	\begin{algorithmic}[1]
		\alglinenoPop{alg1}
		\State \textbf{Algorithm} \textsc{TraverseRUall}$(\pNode)$
		\Indent
		\State Initialize local variables $I \leftarrow \emptyset$ and  $D \leftarrow \emptyset$
		\State $y \leftarrow \pNode.\key$
		\State $\uNode \leftarrow \pNode.\RuallPosition$\label{ln:travRUALL:head}
		\Do
		\State atomic copy $\uNode.\nextptr$ to $\pNode.\RuallPosition$  \Comment{Atomic read and write}\label{ln:travRUALL:atomicCopy}
		\State $\uNode \leftarrow \pNode.\RuallPosition$ \label{ln:travRUALL:updateuNode}
		\If {$\uNode.\key < y$}\label{ln:TraverseRUall:keyCheck}
		\If {($\uNode.\status \neq \textsc{Inactive}$ and $\textsc{FirstActivated}(\uNode)$)} \label{ln:traverseRevUALL:findLatest}
		\If {$\uNode.type = \textsc{INS}$} $I \leftarrow I \cup \{\uNode\}$
		\Else  ~$D \leftarrow D \cup \{\uNode\}$\label{ln:TraverseRUall:addD}
		\EndIf
		\EndIf
		\EndIf
		\EndDo {$(\uNode.\key \neq -\infty)$} \label{ln:TraverseRUall:whileCheck}
		
		\State \Return $I, D$
		\EndIndent
		\alglinenoPush{alg1}
	\end{algorithmic}
\end{figure}

We now explain the purpose of having available the key of the update node $pOp$ is currently visiting in the $\RUALL$.
Recall that when an update operation creates a notify node, $\nNode$, to add to $\pNode.\notifyList$, it reads $\pNode.\RuallPosition.\key$ and writes it into $\nNode.\notifyThreshold$. This is used by $pOp$ to determine if $\nNode.\key$ should be used as a candidate return value. 
For example, consider a \textsc{Delete}$(w)$ operation, for some where $w < y$, linearized before the start of $pOp$, that notifies $pOp$. 
This \textsc{Delete}$(w)$ operation's DEL node should not be used to determine a candidate return value, because $w$ was removed from $S$ before the start of $pOp$.
If $pOp$ sees this DEL node 
when it traverses  the $\RUALL$, then it is added to $\Druall$ and hence will not be used. 
Otherwise this DEL is removed from the $\RUALL$ before $pOp$ can see it during its traversal of the $\RUALL$.
The notification from a \textsc{Delete} operation should only be used to determine a candidate return value if $pOp$ can guarantee that the \textsc{Delete} operation was linearized sometime during $pOp$.
In particular, when $pOp$ does not add a DEL node into $\Druall$ and $pOp$ is currently at an update node with key strictly less than $w$, the \textsc{Delete}$(w)$ operation must have added its DEL node into the $\RUALL$ before this update node. 
So only after $pOp$ has encountered an update node with key less than $w$ during its traversal of the $\RUALL$ does $pOp$ begin accepting the notifications of \textsc{Delete}$(w)$ operations. 
Note that $pOp$ cannot accept notifications from \textsc{Insert}$(w)$ operations until $pOp$ also accepts notifications from \textsc{Delete} operations with key larger than $w$. Otherwise $pOp$ may miss a candidate return value larger than $w$.
So only after $pOp$ has encountered an update node with key less than or equal to $w$ during its traversal of the $\RUALL$ does $pOp$ begin accepting the notifications of \textsc{Insert}$(w)$ operations. 




It is important that the $\RUALL$ is sorted by decreasing key. 	Since the $\RUALL$ is sorted by decreasing key, as $pOp$ traverses the $\RUALL$, it begins accepting notifications from update operations with progressively smaller keys.  This is so that if $pOp$ determines accepts the notification from an update operation, it does not miss notifications of update operations with larger key.
For example, consider the execution depicted in Figure~\ref{figure:ruall_example_2}.
\begin{figure}[!h]
	\centering
	\includegraphics[scale=0.5]{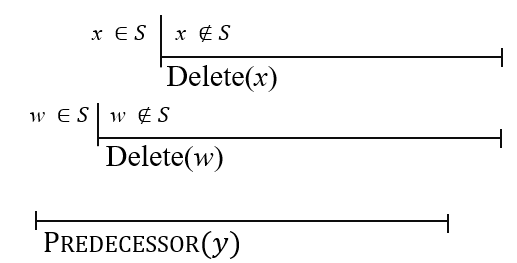}
	\caption{An example execution of a \textsc{Delete}$(x)$, \textsc{Delete}$(w)$, and \textsc{Predecessor}$(y)$ operation, where $w < x < y$.}
	\label{figure:ruall_example_2}
\end{figure}
There are three concurrent operations: a \textsc{Predecessor}$(y)$ operation, a \textsc{Delete}$(x)$ operation, and a \textsc{Delete}$(w)$ operation, where $w < x < y$. The \textsc{Delete}$(x)$ operation is linearized after the \textsc{Delete}$(w)$ operation, and both are linearized after the start of the \textsc{Predecessor}$(y)$ operation, $pOp$.
Notice that $x \in S$ for all configurations during $pOp$ in which $w \in S$. So if $w$ is a candidate return value of $pOp$, $pOp$ must also determine a candidate return value which is at least $x$. 
Hence,
If $pOp$ accepts notifications from \textsc{Delete}$(w)$ operations, 
it also accepts notifications from update operations with keys larger than $w$, and hence $x$. 

Atomic copy is used to make sure that no update operations modify the next pointer of an update node in the $\RUALL$ between when $pOp$ reads this next pointer and when $pOp$ writes a copy it into $\pNode.\RuallPosition$.
Otherwise $pOp$ may miss such an update operation whose key should be used as a candidate return value. 	
To see why, we consider the following execution where $pOp$ does not use atomic copy, which is depicted in Figure~\ref{figure:ruall_example}. 
The $\RUALL$ contains an update node, $\uNode_{20}$, with key $20$, and the two sentinel nodes with keys $\infty$ and $-\infty$. 
A \textsc{Predecessor}$(40)$ operation, $pOp$, reads a pointer to  $\uNode_{20}$ during the first step of its traversal of the $\RUALL$, but it does not yet write it into its predecessor node, $\pNode$.
An $S$-modifying \textsc{Delete}$(25)$ operation, $dOp_{25}$, is linearized, and then an $S$-modifying \textsc{Delete}$(29)$ operation,  $dOp_{29}$, is linearized. 
These DEL nodes of these \textsc{Delete} operations are not seen by $pOp$ because they are added to the $\RUALL$ before $pOp$'s current location in the $\RUALL$. 
Then $dOp_{29}$ attempts to notify $pOp$. It reads that $\pNode.\RuallPosition$ points to the sentinel node with key $\infty$, so $dOp_{29}$ writes $\infty$ to the notify threshold of its notification.  Hence, the notification is rejected by $pOp$.
Now $pOp$ writes the pointer to $\uNode_{20}$ to $\pNode.\RuallPosition$. So $pOp$ begins accepting the notifications of \textsc{Delete} operations with key greater than 20. 
When  $dOp_{25}$ attempts to notify $pOp$, it reads that $\pNode.\RuallPosition$ points to $\uNode_{25}$, so it writes $25$ to the notify threshold of its notification. Hence, the notification is accepted by $pOp$, and 25 is a candidate return value of $pOp$. In all configurations in which 25 is in $S$, the key $29$ is also in $S$. So $pOp$ should not return 25.
By using atomic copy, either both 25 and 29 will be added as candidate return values, or neither will be.

\begin{figure}[!h]
	\centering
	\includegraphics[scale=0.5]{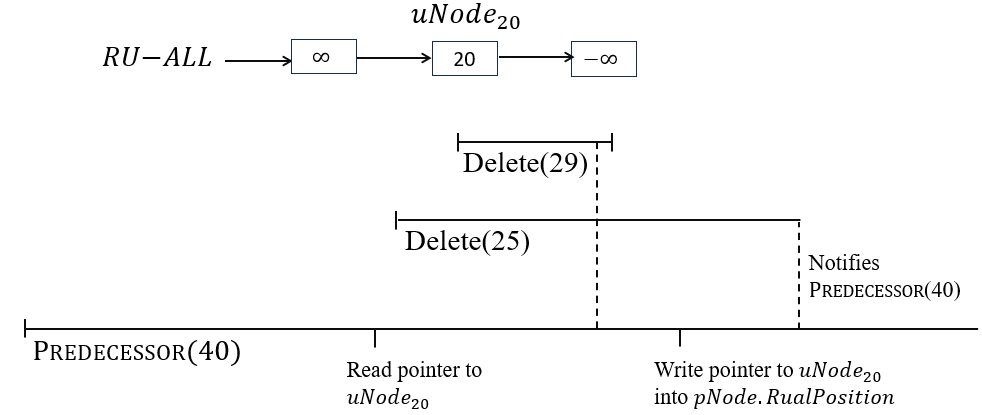}
	\caption{Example execution of a \textsc{Predecessor}$(40)$, \textsc{Delete}$(25)$, and \textsc{Delete}$(29)$. The vertical dashed lines indicate when a \textsc{Delete} operation notifies the \textsc{Predecessor}$(40)$ operation.}
	\label{figure:ruall_example}
\end{figure}

\noindent\\
\textbf{Traversing the relaxed binary trie:} 
Following $pOp$'s traversal of the $\RUALL$, $pOp$ traverses the relaxed binary trie using \textsc{RelaxedPredecessor}$(y)$ (on line~\ref{ln:pred:travtrie}). This has been described in Section~4. If it returns a value other than $\bot$, the value is a candidate return value of $pOp$. 

\noindent\\
\textbf{Traversing the $\UALL$:}
The traversal of the $\UALL$ is done in \textsc{TraverseUall}$(y)$  (on line~\ref{ln:pred:uall2}). 
Recall that it returns two sets of update nodes $\Iuall$ and $\Duall$.
The keys of INS nodes in $\Iuall$ are candidate return values. The keys of DEL nodes in $\Duall$ not seen during $pOp$'s traversal of $\RUALL$ are candidate return values. This is because the \textsc{Delete} operations that created these DEL nodes are linearized sometime during $pOp$.

\noindent\\
\textbf{Collecting notifications:}
The collection of $pOp$'s notifications is done  on line~\ref{ln:pred:newlist} to line~\ref{ln:pred:updateNodeMax}.
Consider a notify node, $\nNode$, created by an update operation, $uOp$, with key less than $y$ that $pOp$ encounters in its $\notifyList$ at the beginning of the for-loop on line~\ref{ln:pred:newlist}.
If $\nNode$ was created by an \textsc{Insert} operation, then $\nNode$ is \textit{accepted} if $\nNode.\notifyThreshold$ is less than or equal to $\nNode.\key$. 
In this case, the update node created by the  \textsc{Insert} operation is put into $pOp$'s local set $\Inotify$.
If $\nNode$ was created by a \textsc{Delete} operation, then $\nNode$ is \textit{accepted} if $\nNode.\notifyThreshold$ is less than $\nNode.\key$. In this case, the update node created by the  \textsc{Delete} operation is put into $pOp$'s local set $\Dnotify$.
Note that, if $\nNode$ is not accepted, it may still be used in the sixth part of the algorithm.

Recall that $\nNode.\updateNodeMax$ 
is the INS node with largest key less than $y$ that $uOp$ identified during its traversal of the $\UALL$. This is performed before $uOp$ notifies any \textsc{Predecessor} operations. Operation $pOp$ also determines if the key of this INS node should be used as a candidate return value:
If $\nNode.\notifyThreshold = -\infty$ (indicating that $pOp$ had completed its traversal of the $\RUALL$ when $pOp$ was notified) and
$\nNode.\updateNode$ is not an update node in $\Iruall$ or $\Druall$ (checked on line~\ref{ln:pred:updateNodeMaxCheck}), then $\nNode.\updateNodeMax$ is also added to $\Inotify$. 
We need to ensure that $pOp$ does not miss relevant \textsc{Insert} operations linearized after $pOp$  completed its traversal of the $\UALL$ and before $uOp$ is linearized. These  \textsc{Insert} operations might not notify $pOp$, and their announcements are not seen by $pOp$ when it traverses the $\UALL$. We guarantee that $\nNode.\updateNodeMax$ is the INS node with largest key less than $y$ that falls into this category.
For example, consider the execution shown in Figure~\ref{figure:ruall_example_3}. Let $w$, $x$, and $y$ be three keys where $w < x < y$.
An \textsc{Insert}$(x)$ operation, $iOp_x$ is linearized before an \textsc{Insert}$(w)$ operation, $iOp_w$, and both are linearized after a \textsc{Predecessor}$(y)$ operation, $pOp$, has completed its traversal of the $\UALL$. 
Suppose $iOp_w$ notifies $pOp$, but $iOp_x$ does not. Then $w$ is a candidate return value of $pOp$.
Note that $pOp$ does not see the announcement of $iOp_x$ when it traverses the $\UALL$. 
In this execution, $x \in S$ whenever $w \in S$. Since $pOp$ returns its largest candidate return value and $w$ is a candidate return value, 
$pOp$ must determine a candidate return value at least $x$.
The INS node of $iOp_x$ is in the $\UALL$ throughout $iOp_w$'s traversal of the $\UALL$ and, hence, is seen by $iOp_w$.
So, when $iOp_w$ notifies $pOp$, $iOp_w$ will set $\updateNodeMax$ to point to this INS node.
Hence, $x$ is a candidate return value of $pOp$.

\begin{figure}[h]
	\centering
	\includegraphics[scale=0.5]{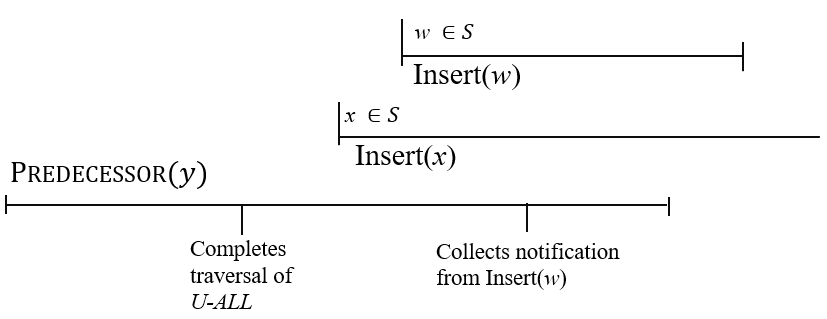}
	\caption{Example execution of a \textsc{Predecessor}$(y)$, \textsc{Insert}$(w)$, and \textsc{Insert}$(x)$, where $w < x < y$.}
	\label{figure:ruall_example_3}
\end{figure}


\noindent\\
\textbf{When the traversal of the relaxed binary trie returns $\bot$:}
Let $k$ be the largest key less than $y$ that is completely present throughout $pOp$'s traversal of the relaxed binary trie, or $-1$ if no such key exists. If $pOp$'s traversal returns $\bot$,
then by the specification of the relaxed binary trie,
there is an $S$-modifying update operation $uOp$ with key $x$, where $k < x < y$, whose update to the relaxed binary trie is concurrent with $pOp$'s traversal of the relaxed binary trie.
The update node created by $uOp$ is encountered by $pOp$ either in the $\UALL$ or in its notify list. This is because either $pOp$ will traverse the $\UALL$ before $uOp$ can remove its update node from the $\UALL$, or $uOp$ will notify $pOp$ before $pOp$ removes its predecessor node from the $\PALL$. 
Unless $uOp$ is a \textsc{Delete}$(x)$ operation whose DEL node is in $\Druall$, $x$ is a candidate return value.
This gives the following observation:
If $r_1 < k$ (on line~\ref{ln:pred:candidate}), then there is a DEL node in $\Druall$ with key $x$ such that $k < x < y$.


When the traversal of the relaxed binary trie returns $\bot$ and $\Druall$ is non-empty, $pOp$ takes additional steps to guarantee it has a candidate return value at least $k$ (by executing lines~\ref{ln:pred:fix_binary_trie_pred} to \ref{ln:pred:set_pred0}).
This is done by using the keys and results of embedded predecessor operations of update operations linearized before the start of $pOp$'s traversal of the relaxed binary trie, and possibly before the start of $pOp$.
First, $pOp$ determines the predecessor nodes created by the first embedded predecessor operations of DEL nodes in $\Druall$. 
If $pOp$ encounters one of these predecessor nodes when it traversed the $\PALL$, $pOp$ sets $\pNode'$ to be the one it encountered the latest
in the $\PALL$ (on line~\ref{ln:pred:pNode'}). Note that $\pNode'$ was announced the earliest among these predecessor nodes and also announced earlier than $\pNode$.
Then $pOp$ traverses $\pNode'.\notifyList$ to determine the update nodes of update operations with key less than $y$ that notified $\pNode'$ (on lines~\ref{ln:pred:L} to \ref{ln:pred:prependL'1}). These update nodes are stored in a local sequence, $L_1$, and appear in the order in which their notifications were added to $\pNode'.\notifyList$.

Next, $pOp$ traverses $\pNode.\notifyList$ to determine the update nodes of update operations with key less than $y$ that notified $\pNode$. Those belonging to notifications whose $\notifyThreshold$ are greater than or equal to the key of the notification
are stored in a local sequence, $L_2$ (on line~\ref{ln:pred:prependL'2}).
The update nodes in $L_2$ appear in the order in which their notifications were added to $\pNode.\notifyList$, and were added to $\pNode.\notifyList$ before $pOp$ completed its traversal of the $\RUALL$. It includes the update nodes of update operations whose notifications were rejected by $pOp$, and may include some INS nodes of \textsc{Insert} operations whose notifications were accepted by $pOp$. 
The local sequence $L$ is $L_1$ followed by $L_2$ (computed on line~\ref{ln:pred:L-L'}). 
On line~\ref{ln:pred:LremoveDEL}, all DEL nodes, $\dNode \in L$ that are not the last update node with key $\dNode.\key$ are removed from $L$. From the resulting sequence of update nodes, a directed graph is locally determined by $pOp$ according to the following definition.

\begin{definition}\label{definition:T_L}
	Let $L$ be the sequence of update nodes computed on line~\ref{ln:pred:LremoveDEL}.
	Let $V = \{ \uNode.\key \mid \uNode \in L \} \cup \{ \dNode.\delpredsecond \mid \dNode \in L \} \cup \{ \dNode.\delpred \mid \dNode \in \Druall \}$.
	Let $E$ be the set of directed edges $(u,v)$ where
	$\dNode.\key = u$ and $\dNode.\delpredsecond = v$ for some $\dNode \in L$.
	Let $T_L = (V,E)$.
\end{definition}

\noindent 
From line~\ref{ln:pred:LremoveDEL}, there is at most one DEL node with key $v$ in $L$, so any vertex $v$ in $T_L$ has at most one outgoing edge. 
Moreover, $v < u$ for any edge $(u,v)$ in $T_L$,
because $v$ is the return value of an embedded predecessor operation with key $u$.

A candidate return value is then computed from $L$, $T_L$, and $\Druall$ (on lines~\ref{ln:pred:X_init} to line~\ref{ln:pred:set_pred0}).
The set $X$ is initialized to be the set of keys of INS nodes in $L$ (on line~\ref{ln:pred:X_init}) together with the set of keys that are the return values of the first embedded predecessor operations of DEL nodes in $\Druall$ (on line~\ref{ln:pred:X_init_INS}).
Each key $w \in X$ is in $S$ sometime between $C_{ann}$ and the start of $pOp$'s traversal of the relaxed binary trie. However,
it may be deleted from $S$ before $pOp$ begins its traversal of the relaxed binary trie. Each edge in $(w,w')$ in $T_L$ represents a key $w$ that is removed from a \textsc{Delete}$(w)$ operation, and whose second embedded predecessor operation returns the key $w'$. So the sinks of $T_L$ represent keys that may not have been deleted before the start of $pOp$'s traversal of the relaxed binary trie.
The keys of DEL nodes in $\Druall$ are removed from $R$ (on line~\ref{ln:pred:R_remove_Druall}). 
The largest remaining key in $R$, which we can guarantee is non-empty, is a candidate return value of $pOp$.

\subsection{Linearizability}\label{section_predecessor_linearization}

This section shows that our implementation of the lock-free binary trie is linearizable.
We first prove basic properties about the $\latest$ lists in Section~\ref{section_latest_properties}.
We then show that our implementation is linearizable with respect to \textsc{Search}, \textsc{Insert}, and \textsc{Delete} operations in Section~\ref{section_ins_del_search_lin}.

Recall that
in a configuration $C$, the predecessor of $y$ is the key $w$ such that  $w \in S$ and there is no key $x \in S$ such that $w < x < y$, otherwise it is $-1$ if there is no key in $S$ smaller than $y$.
We show that if a \textsc{Predecessor}$(y)$ operation returns $w$, then there is a configuration $C$ during its execution interval in which $w$ is the predecessor of $y$. 
Each completed \textsc{Predecessor}$(y)$ operation can be linearized at any such configuration. This shows that our implementation is linearizable.





Recall that our implementation of \textsc{Predecessor} determines a number of candidate return values, and returns the largest of them.
In Section~\ref{section_crv_properties}, we first define three properties, denoted Properties~\ref{prop:largest}, \ref{prop:configC}, and \ref{prop:trie}, that these candidate return values will satisfy.
Additionally, we prove that any implementation of \textsc{Predecessor} whose candidate return values satisfies these properties, 
together with our implementations of \textsc{Search}, \textsc{Insert}, and \textsc{Delete},
results in a linearizable implementation of a lock-free binary trie.  
Additionally, we show that the candidate return values of our implementation satisfy Property~\ref{prop:largest}.
In Section~\ref{section_crv_2}, we prove properties of the candidate return values that are keys of update nodes in the set $\Iuall \cup \Inotify \cup (\Duall - \Druall) \cup (\Dnotify - \Druall)$.
In Section~\ref{section_crv_3}, we prove that Properties~\ref{prop:configC} and \ref{prop:trie} are satisfied.


\subsubsection{Properties of the Latest Lists}\label{section_latest_properties}

In this section, we prove basic facts about the $\latest[x]$ lists, for each $x \in U$. They are used to show the linearizability of \textsc{Insert}, \textsc{Delete}, and \textsc{Search} operations.



\begin{lemma}\label{lemma:latest:length2}
	An update node with key $x$ is only added to the beginning of the $\latest[x]$ list when $\latest[x]$ points to an update node whose $\latestNext$ field is $\bot$.
\end{lemma}

\begin{proof}
	Let $\uNode$ be the update node pointed to by $\latest[x]$.
	New update nodes are only added to the $\latest[x]$ list by updating $\latest[x]$ to point to a different update node, $\uNode'$, using CAS$(\latest[x], \uNode, \uNode')$ on line~\ref{ln:insert:cas_latest_head} (for \textsc{Insert} operations) or \ref{ln:delete:cas_latest_head} (for \textsc{Delete} operations), where $\uNode'.\latestNext$ points to $\uNode$. Immediately before this CAS on line~\ref{ln:insert:set_latestNext_bot} (for \textsc{Insert} operations) or \ref{ln:delete:set_latestNext_bot} (for \textsc{Delete} operations), $\uNode.\latestNext$ is set to $\bot$.
\end{proof}

This implies that the $\latest[x]$ list is a sequence of at most 2 update nodes.

\begin{corollary}\label{coro:latest:length2}
	Either $\latest[x]$ points to an update node whose $\latestNext$ field is $\bot$, or $\latest[x]$ points to an update node whose $\latestNext$ field points to an update node whose $\latestNext$ field is $\bot$.
\end{corollary}

We next prove that the helper function \textsc{FindLatest}$(x)$ correctly returns an update node that was the first activated update node in the $\latest[x]$ list sometime during its execution.

\begin{lemma}\label{lemma:findLatestConfig}
	Let $\tau$ be a completed instance of \textsc{FindLatest}$(x)$.
	Then there is a configuration during $\tau$ in which the update node returned is the first activated node in the $\latest[x]$ list.
\end{lemma}

\begin{proof}
	We prove by induction on the sequence of instances of \textsc{FindLatest}$(x)$ in the order that they are completed. 
	
	Suppose $\tau$ is the first instance of \textsc{FindLatest}$(x)$ that is completed. Note that before $\latest[x]$ is changed by a CAS on line~\ref{ln:insert:cas_latest_head} or line~\ref{ln:delete:cas_latest_head}, an instance of  \textsc{FindLatest}$(x)$ must be previously completed (on line~\ref{ln:insert:find_latest} or \ref{ln:delete:findLatest}). Therefore,  throughout $\tau$, $\latest[x]$ does not change and it points to the dummy DEL node with key $x$. By definition, this dummy DEL node is active, so it is the first activated update node in the $\latest[x]$ list throughout $\tau$. It follows from the code of \textsc{FindLatest}$(x)$ that $\tau$ returns this dummy DEL node.

	So suppose the lemma is true for all instances of \textsc{FindLatest}$(x)$ completed before $\tau$.
	We first show that if $\latest[x]$ changes during $\tau$, it changes from pointing to an active update node to an inactive update node. The pointer $\latest[x]$ only changes by a CAS on line~\ref{ln:insert:cas_latest_head} or \ref{ln:delete:cas_latest_head}.
	These CAS steps change $\latest[x]$ to point from an update node returned by an instance of \textsc{FindLatest}$(x)$ (on line~\ref{ln:insert:find_latest} or line~\ref{ln:delete:findLatest}), to a newly created inactive update node.
	By the induction hypothesis, the update node returned by \textsc{FindLatest}$(x)$ is active. So  $\latest[x]$ only changes from pointing to an active update node to an inactive update node. 
	
	Let $\uNode$ be the update node pointed to by $\latest[x]$ that is read by $\tau$ on line~\ref{ln:findLatest:read_head}.
	Suppose that $\tau$ reads $\uNode.\status = \textsc{Active}$ on line~\ref{ln:findLatest:if}. Then $\uNode$ is the first activated node in the $\latest[x]$ list in some configuration between the read of $\latest[x]$ on line~\ref{ln:findLatest:read_head} and the read of $\uNode.\status$ on line~\ref{ln:findLatest:if}. Since $\tau$ returns $\uNode$, the lemma holds.
	
	Suppose $\tau$ reads that $\uNode.\status = \textsc{Inactive}$ on line~\ref{ln:findLatest:if}. Since $\latest[x]$ only changes when it points to an active update node, $\latest[x]$ still points to $\uNode$ when this read occurs.
	Suppose $\tau$ reads that $\uNode.\latestNext = \uNode' \neq \bot$ on line~\ref{ln:findLatest:if2}. 
	Once $\uNode.\latestNext$ is initialized to point to $\uNode'$, it does not change to any other value except for $\bot$. 
	So in the configuration immediately after $\tau$ reads $\uNode.\status = \textsc{Inactive}$,
	$\uNode.\latestNext$ points to $\uNode'$. Since $\latestNext$ is always initialized to point to an update node returned by \textsc{FindLatest}$(x)$  (which is active by the induction hypothesis), $\uNode'$ is the first activated update node in the $\latest[x]$ list in this configuration.
	Since $\tau$ returns $\uNode'$, the lemma holds.

	So suppose $\tau$ reads that $\uNode.\latestNext = \bot$ on line~\ref{ln:findLatest:read_latestNext}.
	We argue that $\uNode$ was activated sometime before $\uNode.\latestNext$ was set to $\bot$.
	\begin{itemize}
		\item 	Suppose $\uNode.\latestNext$ was changed to $\bot$ on line~\ref{ln:insert:set_bot} (for \textsc{Insert} operations) or line~\ref{ln:delete:set_bot} (for \textsc{Delete} operations). On the preceding line (line~\ref{ln:insert:set_active} for \textsc{Insert} operations, or line~\ref{ln:delete:set_active} for \textsc{Delete} operations), $\uNode.\status$ was set to \textsc{Active}.	
		
		\item Suppose $\uNode.\latestNext$ was changed to $\bot$ on line~\ref{ln:insert:set_latestNext_bot} (for \textsc{Insert} operations) or line~\ref{ln:delete:set_latestNext_bot}  (for \textsc{Delete} operations). 
		Since $\uNode$ was returned by \textsc{FindLatest}$(x)$ on line~\ref{ln:insert:find_latest} or line~\ref{ln:delete:findLatest}, it was activated, by the induction hypothesis.
		
		\item Suppose $\uNode.\latestNext$ was changed to $\bot$ on line~\ref{ln:help_activate:set_bot}. Then $\uNode.\status$ was previously set to \textsc{Active} on line~\ref{ln:help_activate:CAS}.
	\end{itemize}
	So $\uNode$ was activated sometime before $\uNode.\latestNext$ was set to $\bot$. Since $\latest[x]$ points to $\uNode$ until after it is activated, $\uNode$ is the first activated update node in the $\latest[x]$ list in the configuration immediately after its status changed from \textsc{Inactive} to \textsc{Active}. Since $\tau$ returns $\uNode$, the lemma holds.
\end{proof}

The next two lemmas, together with the previous lemmas in this section, argue how update nodes are added to the $\latest[x]$ list. In particular, only inactive update nodes are added to the $\latest[x]$ list. Furthermore, their status must be set to \textsc{Active} and their $\latestNext$ field must be set to $\bot$ before another update node is added. 

\begin{lemma}\label{lemma:latest_change}
	When $\latest[x]$ changes, it changes from pointing to an active INS node to an inactive DEL node, or from an active DEL node to an inactive INS node.
\end{lemma}

\begin{proof}
	The pointer $\latest[x]$ only changes by a CAS on line~\ref{ln:insert:cas_latest_head} or \ref{ln:delete:cas_latest_head}.
	These CAS steps change $\latest[x]$ to point from an update node returned by an instance of \textsc{FindLatest}$(x)$ (on line~\ref{ln:insert:find_latest} or line~\ref{ln:delete:findLatest}), to a newly created inactive update node.
	By Lemma~\ref{lemma:findLatestConfig}, the update node returned by \textsc{FindLatest}$(x)$ is active. The if-statements on lines~\ref{ln:insert:return_early} and \ref{ln:delete:return_early} guarantee that the type of the newly created inactive update node is different from the type of the update node returned by \textsc{FindLatest}$(x)$.
\end{proof}

\begin{lemma}\label{lemma:change_active}
	When the status of an update node, $\uNode$, with key $x$ changes from \textsc{Inactive} to \textsc{Active}, $\uNode$ is the first activated update node in the $\latest[x]$ list.
\end{lemma}

\begin{proof}
	
	We first prove $\latest[x]$ was changed to point to $\uNode$ in some previous configuration.
	Suppose $\uNode.\status$ first changes to \textsc{Active} on line~\ref{ln:insert:set_active} (for \textsc{Insert} operations) or line~\ref{ln:delete:set_active} (for \textsc{Delete} operations).
	Then, in a previous step, a CAS on line~\ref{ln:insert:cas_latest_head} (for \textsc{Insert} operations)  or \ref{ln:delete:cas_latest_head} (for \textsc{Insert} operations) successfully changed $\latest[x]$ to point to $\uNode$ when it was inactive. 
	The only other case is if $\uNode.\status$ first changes to \textsc{Active} on line~\ref{ln:help_activate:CAS} during an instance of \textsc{HelpActivate}$(\uNode)$. Then $\uNode$ is pointed to by $\latest[x]$ when it is read on line~\ref{ln:insert:help_activate} (for \textsc{Insert} operations) or line~\ref{ln:delete:help_activate} (for \textsc{Delete} operations).
	
	By Lemma~\ref{lemma:latest_change}, $\latest[x]$ points to $\uNode$ while its status is \textsc{Inactive}. 
	Therefore, when $\uNode.\status$ changes from \textsc{Inactive} to \textsc{Active}, $\latest[x]$ points to $\uNode$. So $\uNode$ is the first activated update node in the $\latest[x]$ list immediately after this change.
\end{proof}

We prove in the following two lemmas that the helper function \textsc{FirstActivated}$(\uNode)$ correctly identifies whether $\uNode$ is the first activated update node in the $\latest[x]$ list.

\begin{lemma}\label{lemma:firstActivated}
	Let $\tau$ be a completed instance of \textsc{FirstActivated}$(\uNode)$, for some activated update node, $\uNode$, with key $x$. If $\tau$ returned \textsc{True}, there is a configuration during $\tau$ in which $\uNode$ was the first activated update node in the $\latest[x]$ list.
\end{lemma}

\begin{proof}
	Let $\uNode'$ be the update node pointed to by $\latest[x]$ when it was read on line~\ref{ln:firstactivated:head}.
	Suppose $\tau$ returned \textsc{True} because $\uNode' = \uNode$. Since, by assumption, $\uNode$ was activated before the start of $\tau$, $\uNode$ was the first activated node in the $\latest[x]$ list immediately after this read on line~\ref{ln:firstactivated:head}.
	
	So suppose $\tau$ returned \textsc{True} because $\uNode'.\status = \textsc{Inactive}$ and $\uNode = \uNode'.\latestNext$. 
	By Lemma~\ref{lemma:latest_change}, $\latest[x]$ still pointed to  $\uNode'$ when $\tau$ read $\uNode'.\status = \textsc{Inactive}$. 
	Once $\uNode'.\latestNext$ was initialized to point to $\uNode$, it does not change to any other value except for $\bot$. 
	So in the configuration  immediately after $\tau$ read $\uNode'.\status = \textsc{Inactive}$,
	$\uNode'.\latestNext$ pointed to $\uNode$ and $\uNode$ was the first activated update node in the $\latest[x]$ list.

	
\end{proof}

\begin{lemma}\label{lemma:firstActivated_False}
	Let $\tau$ be a completed instance of \textsc{FirstActivated}$(\uNode)$, for some activated update node, $\uNode$, with key $x$. If $\tau$ returned \textsc{False}, there is a configuration during $\tau$ in which $\uNode$ was not the first activated node in the $\latest[x]$ list.
\end{lemma}

\begin{proof}
	Let $\uNode'$ be the update node pointed to by $\latest[x]$ when it was read on line~\ref{ln:firstactivated:head}.
	Suppose $\tau$ returned \textsc{False} because $\uNode' \neq \uNode$ and $\uNode'.\status = \textsc{Active}$. 
	If $\uNode'.\status = \textsc{Active}$ when $\latest[x]$ was read on line~\ref{ln:firstactivated:head}, then $\uNode$ was not the first activated node in the $\latest[x]$ list immediately after this read.
	So suppose $\uNode'.\status = \textsc{Inactive}$ when $\latest[x]$ was read on line~\ref{ln:firstactivated:head}. So $\uNode'.\status$ changed to $\textsc{Active}$ sometime before $\tau$ read that $\uNode'.\status = \textsc{Active}$. 
	By Lemma~\ref{lemma:change_active}, when $\uNode'$ was activated, it was the first activated update node in the $\latest[x]$ list. So $\uNode$ was not the first activated node in the $\latest[x]$ list immediately after $\uNode'$ was activated.
	
	So suppose $\tau$ returned \textsc{False} because $\uNode' \neq \uNode$ and $\uNode'.\latestNext \neq \uNode$. 
	If $\uNode'.\latestNext$ did not point to $\uNode$ when $\latest[x]$ was read on line~\ref{ln:firstactivated:head}, then Corollary~\ref{coro:latest:length2} implies that $\uNode$ was not in the $\latest[x]$ list.
	So suppose $\uNode'.\latestNext$ pointed to $\uNode$ when $\latest[x]$ was read on line~\ref{ln:firstactivated:head}. 
	Because  $\uNode'.\latestNext$ does not point to $\uNode$ when it is read on line~\ref{ln:firstactivated:return} and $\uNode'.\latestNext$ only changes to $\bot$ after it is initialized, 
	$\uNode'.\latestNext = \bot$.
	By Lemma~\ref{lemma:latest:length2}, $\latest[x]$ only changes to point to an update node other than $\uNode'$ when $\uNode'.\latestNext = \bot$. 
	So there is a configuration between the read on line~\ref{ln:firstactivated:head} and the read of $\uNode'.\latestNext$ on line~\ref{ln:firstactivated:return} when the $\latest[x]$ list only contains $\uNode'$. Therefore, $\uNode$ was not the first activated node in the $\latest[x]$ list in this configuration.

\end{proof}

\subsubsection{Linearizability of Insert, Delete, and Search}\label{section_ins_del_search_lin}

A \textsc{Search}$(x)$ operation that returns \textsc{True} is linearized in any configuration during its execution interval in which $x \in S$. The next lemma proves such a configuration exists.

\begin{lemma}
Suppose $op$ is a \textsc{Search}$(x)$ operation that returns \textsc{True}. Then there exists a configuration during $op$ in which $x \in S$.
\end{lemma}

\begin{proof}
\textsc{Search}$(x)$ begins by calling \textsc{FindLatest}$(x)$ on line~\ref{ln:search:findLatest}, which returns an update node $\uNode$.
By Lemma~\ref{lemma:findLatestConfig}, there is a configuration $C$ during this instance of \textsc{FindLatest}$(x)$ in which $\uNode$ is the first activated update node in the $\latest[x]$ list. Since $op$ returned \textsc{True}, it read that $\uNode.\type = \textsc{INS}$. By definition, $x \in S$ in $C$.
\end{proof}

A \textsc{Search}$(x)$ operation that returns \textsc{False} is linearized in any configuration during its execution interval in which $x \notin S$. 

\begin{lemma}
Suppose $op$ is a \textsc{Search}$(x)$ operation that returns \textsc{False}. Then there exists a configuration during $op$ in which $x \notin S$.
\end{lemma}

\begin{proof}
\textsc{Search}$(x)$ begins by calling \textsc{FindLatest}$(x)$ on line~\ref{ln:search:findLatest}, which returns an update node $\uNode$.
By Lemma~\ref{lemma:findLatestConfig}, there is a configuration $C$ during this instance of \textsc{FindLatest}$(x)$ in which $\uNode$ is the first activated update node in the $\latest[x]$ list. Since $op$ returned \textsc{False}, it read that $\uNode.\type = \textsc{DEL}$. By definition, $x \notin S$ in $C$.
\end{proof}

The next lemma shows how to choose the linearization points of \textsc{Insert} operations. 
An $S$-modifying \textsc{Insert}$(x)$ operation is linearized when the status of the update node it created is changed from \textsc{Inactive} to \textsc{Active}. The lemma shows $x\in S$ in this configuration.
An \textsc{Insert}$(x)$ operation that is not $S$-modifying does not update $\latest[x]$ to point to its own update node because it returns early on line~\ref{ln:insert:return_early} or line~\ref{ln:insert:return_help}.
This happens when it sees that the first activated update node in the $\latest[x]$ is already an INS node, or when it unsuccessfully attempts to change $\latest[x]$ to point to its own INS node (by the CAS on line~\ref{ln:insert:cas_latest_head}).
In each of these cases, we prove that there is a configuration during the \textsc{Insert}$(x)$ operation in which $x \in S$, and, therefore, the operation does not need to add $x$ to $S$. 

\begin{lemma}
Let $iOp$ be a completed \textsc{Insert}$(x)$ operation. Then there was a configuration during $iOp$ in which $x \in S$.
\end{lemma}

\begin{proof}
\textsc{Insert}$(x)$ begins by calling \textsc{FindLatest}$(x)$ on line~\ref{ln:insert:find_latest}, which returns an update node, $\dNode$. By Lemma~\ref{lemma:findLatestConfig}, there is a configuration $C$ during this instance of \textsc{FindLatest}$(x)$ in which $\dNode$ is the first activated node in the $\latest[x]$ list.
By Corollary~\ref{coro:latest:length2}, in $C$, either $\latest[x]$ points to $\dNode$, or it points to an inactive update node whose $\latestNext$ field points to $\dNode$. 

If  $iOp$ returns on line~\ref{ln:insert:return_early}, then $\dNode$ is an INS node and $x\in S$. So suppose  $iOp$ does not return on line~\ref{ln:insert:return_early} and $\dNode$ is a DEL node.
So $iOp$ creates a new inactive INS node, $\iNode$. 
Consider the CAS that $iOp$ performs on line~\ref{ln:insert:cas_latest_head}, which attempts to change $\latest[x]$ to point from $\dNode$ to $\iNode$. 
Suppose this CAS is successful, so $iOp$ is an $S$-modifying \textsc{Insert}$(x)$ operation.
The status of $\iNode$ must change to \textsc{Active} by the time $iOp$ sets $\iNode.\status$ to \textsc{Active} on line~\ref{ln:insert:set_active}.
By Lemma~\ref{lemma:change_active},  $\iNode$ is the first activated update node in the $\latest[x]$ list when $\iNode.\status$ changes from \textsc{Inactive} to \textsc{Active}, so $x \in S$ immediately after this change.

So suppose the CAS on line~\ref{ln:insert:cas_latest_head} is unsuccessful. 
Then either $\latest[x]$ pointed to an inactive INS node in $C$ or, between $C$ and this unsuccessful CAS, some other \textsc{Insert}$(x)$ operation changed $\latest[x]$ to point from $\dNode$ to its own INS node.
Let $\iNode'$ be this INS node.
By Lemma~\ref{lemma:latest_change}, $\latest[x]$ does not change to point to a different update node until $\iNode'$ is activated. On line~\ref{ln:insert:help_activate}, $iOp$ reads the update node, $\uNode$, pointed to by $\latest[x]$, and performs an instance of \textsc{HelpActivate}$(\uNode)$. If $\uNode = \iNode'$, then $iOp$ sets $\iNode'.\status$ to \textsc{Active} on line~\ref{ln:help_activate:CAS} before returning. 
If $\uNode \neq \iNode'$, then  $\iNode'$ was activated sometime before $\latest[x]$ was changed to point to $\uNode$.
In either case, in the configuration immediately after $\iNode'.\status$ changed to \textsc{Active}, $\iNode'$ was the first activated update node in the $\latest[x]$ and $x \in S$.
\end{proof}

By symmetric arguments, we can prove a similar lemma for \textsc{Delete}$(x)$ operations. 

\begin{lemma}
Let $dOp$ be a completed \textsc{Delete}$(x)$ operation. Then there was a configuration during $dOp$ in which $x \notin S$.
\end{lemma}

\subsubsection{Candidate Return Values and Their Properties}\label{section_crv_properties}


Recall that our implementation of \textsc{Predecessor}$(y)$ performs a single instance of \textsc{PredHelper}$(y)$ and then returns the result of this instance.
An instance of \textsc{PredHelper}$(y)$ produces a number of candidate return values, and returns the largest of these. In this section, we state properties of the candidate return values of a \textsc{PredHelper}$(y)$ operation, and prove that the value returned by this operation is correct, assuming these properties hold.

Let $pOp$ be a completed instance of \textsc{PredHelper}$(y)$, and let $\pNode$ be the predecessor node it created. 
During $pOp$, it computes six sets of update nodes: 
$\Iruall$ and $\Druall$ are the sets of INS nodes and DEL nodes obtained from $pOp$'s traversal of the $\RUALL$, 
$\Iuall$ and $\Duall$ are the sets of INS nodes and DEL nodes obtained from $pOp$'s traversal of the $\UALL$, and $\Inotify$ and $\Dnotify$ are the sets of INS nodes and DEL nodes obtained from
$pOp$'s traversal of its notify list.
The keys of update nodes in $\Iuall \cup \Inotify \cup (\Duall-\Druall) \cup (\Dnotify - \Druall)$ are \textit{candidate return values} of $pOp$. 

There may be one additional candidate return value from $pOp$'s traversal of the relaxed binary trie:
If $pOp$'s traversal of the relaxed binary trie returns a value $x \neq \bot$, then $x$ is a \textit{candidate return value}. 
If the traversal of the relaxed binary trie returns $\bot$ and $\Druall \neq \emptyset$, then $pOp$ assigns a value $x$ to $r_0$ on line~\ref{ln:pred:set_pred0}. 
If $\Iuall \cup \Inotify \cup (\Duall-\Druall) \cup (\Dnotify - \Druall)$ does not contain an update node with key at least $x$ (i.e. $pOp$ does not already have a candidate return value at least $x$), $x$ is a candidate return value.

\begin{property}\label{prop:largest}


All candidate return values of $pOp$ are less than $y$, and $pOp$ returns its largest candidate return value.
\end{property}


For each candidate return value $w \neq -1$, we show that there is a configuration $C$ during $pOp$ in which $w \in S$. In configuration $C$, if there exists a key $x$, where $w < x < y$, that is also in $S$ (so $w$ is not the predecessor of $y$ in $C$), then $pOp$ will determine a candidate return value at least $x$ and, hence, not return $w$. 
Let $C_T$ be the configuration immediately before $pOp$ begins its traversal of the relaxed binary trie.
How $pOp$ determines a candidate return value at least $x$ is different depending on when $C$ occurs relative to $C_T$ and when the $S$-modifying \textsc{Insert}$(x)$ operation, $iOp$, that last added $x$ to $S$ prior to $C$ was linearized.
Suppose $C$ occurs at or after $C_T$. If $iOp$ is linearized after $C_T$, then Property~2(c) states $pOp$ has a candidate return value at least $x$. If $iOp$ is linearized before $C_T$, Property~3 states that $pOp$ has a candidate return value at least $x$.
Now suppose $C$ occurs before $C_T$. If $x$ is removed from $S$ after $iOp$ is linearized, but before $C_T$, Property~2(b) states $pOp$ has a candidate return value at least $x$. If $x$ is in $S$ in all configurations between when $iOp$ is linearized and $C_T$, Property~3 states $pOp$ has a candidate return value at least $x$. 

\begin{property}\label{prop:configC}
Suppose $w \neq -1$ is a candidate return value of $pOp$. 
Then 
there is a configuration $C$ after $pOp$ is announced but before the end of $pOp$ such that

\begin{enumerate}
	\item[(a)] 
	$w \in S$ in $C$,
	
	
	\item[(b)]
	if $C$ occurs before $C_T$ and
	there exists an $S$-modifying \textsc{Delete}$(x)$ operation linearized between $C$ and $C_T$
	with $w < x < y$,  then $pOp$ has a candidate return value which is at least $x$, and
	
	
	
	\item[(c)]     
	if $C$ occurs after $C_T$ and 
	there exists an $S$-modifying \textsc{Insert}$(x)$ operation 
	linearized between $C_T$ and $C$ with  $w < x < y$, then $pOp$ has a candidate return value which is at least $x$.
	
	
\end{enumerate}
\end{property}


\begin{property}\label{prop:trie}
Suppose an $S$-modifying \textsc{Insert}$(x)$ operation $iOp$ is linearized before $C_T$, $x < y$, 
and there are no $S$-modifying \textsc{Delete}$(x)$ operations linearized after $iOp$ and before $C_T$.
Then $pOp$ has a candidate return value which is at least $x$.
\end{property}

We use Property~\ref{prop:largest}, Property~\ref{prop:configC}, and Property~\ref{prop:trie} to prove the following theorem, which states that the value returned by $pOp$ is a predecessor of $y$ in some configuration during $pOp$. 
If $pOp$ is an instance of \textsc{PredHelper} invoked by a \textsc{Predecessor}$(y)$ operation, this operation can be linearized at any such configuration.

\begin{theorem}\label{thm:predecessor}
If $pOp$ returns $w \in U \cup \{-1\}$, then there exists a configuration  after $pOp$ is announced but before the end of $pOp$  in which $w$ is the predecessor of $y$.
\end{theorem}

\begin{proof}
First suppose that $pOp$ returns $-1$. 
To obtain a contradiction, suppose
that there is a key $x \in S$ in $C_T$, where $x < y$.
Let $iOp$ be the \textsc{Insert}$(x)$ operation that last added $x$ to $S$ before $C_T$. So there are no \textsc{Delete}$(x)$ operations linearized after $iOp$ but before $C_T$. By Property~\ref{prop:trie}, there is a key $x'$ where $x \leq x' < y$ that is a candidate return value of $pOp$. This contradicts Property~\ref{prop:largest}.

Otherwise, $pOp$ returns $w \in U$.
By Property~\ref{prop:largest}, $w$ is a candidate return value of $pOp$.
Let $C$ be any configuration during $pOp$ satisfying Property~\ref{prop:configC}. By Property~\ref{prop:configC}(a), $w \in S$ in $C$. 
Suppose that $w$ is not the predecessor of $y$ in $C$. Then there is a key $x \in S$ in $C$, where $w < x < y$.
Let $iOp$ be the \textsc{Insert}$(x)$ operation that last added $x$ to $S$ before $C$.
If $iOp$ is linearized after $C_T$, then
it follows from Property~\ref{prop:configC}(c) that  $pOp$ has a candidate return value that is at least $x$. This contradicts Property~\ref{prop:largest}.

Therefore $iOp$ is linearized before $C_T$. If $x \in S$ in all configurations from the linearization point of $iOp$ to $C_T$, then Property~\ref{prop:trie} states there is a key $x'$ that is a candidate return value of $pOp$, where $w < x \leq x' < y$. This contradicts Property~\ref{prop:largest}.
So  $x \notin S$ in some configuration between the linearization point of $iOp$ and $C_T$.
Since $x \in S$ in all configurations from the linearization point of $iOp$ to $C$, 
it follows that $x$ is removed from $S$ by a \textsc{Delete}$(x)$ operation $dOp$ linearized sometime between $C$ and $C_T$. By Property~\ref{prop:configC}(b),  $pOp$ has a candidate return value that is at least $x$. This contradicts Property~\ref{prop:largest}.
Therefore, in any case, $w$ is the predecessor of $y$ in $C$.
\end{proof}



For the remainder of Section~\ref{section_predecessor_linearization}, let $\alpha$ be an arbitrary execution of our implementation. 
We prove that every instance of \textsc{PredHelper} that is completed in $\alpha$ satisfies Properties~\ref{prop:largest}, \ref{prop:configC}, and \ref{prop:trie}. The proof is by induction on the sequence of these instances in the order in which they were completed.
Let $pOp$ be an arbitrary completed instance of \textsc{PredHelper}$(y)$ in $\alpha$  and
let $\pNode$ be the predecessor node it created. 
We assume that every instance of \textsc{PredHelper} that is completed prior to  $pOp$ satisfies Properties~\ref{prop:largest}, \ref{prop:configC}, and \ref{prop:trie}, and, hence, Theorem~\ref{thm:predecessor}.

It is easy to show that $pOp$ satisfies Property~\ref{prop:largest}.






\begin{lemma}\label{lemma:prop:largest}
All candidate return values of $pOp$ are less than $y$, and $pOp$ returns its largest candidate return value.
\end{lemma}

\begin{proof}

The maximum candidate return value of $pOp$ is returned on line~\ref{ln:pred:return}. It is either a key of an update node in $\Iuall \cup \Inotify \cup (\Duall - \Druall) \cup (\Dnotify - \Druall)$, or the value of $pOp$'s local variable $r_0$. It remains to show that these keys are less than $y$.

The update nodes in $\Iuall$ and $\Duall$ are those returned by \textsc{TraverseUAll}$(y)$ on line~\ref{ln:pred:uall2}. By the check on line~\ref{ln:traverseUall:while}, these update nodes have key less $y$. Update nodes in $\Duall$ and $\Dnotify$ have keys less than $y$ by the check in the while loop on line~\ref{ln:pred:newlist}. 
By the specification of \textsc{RelaxedPredecessor}$(y)$, the value $r_0$ returned by \textsc{RelaxedPredecessor}$(y)$ is either less than $y$, or $\bot$. 

When \textsc{RelaxedPredecessor}$(y)$ returns $\bot$, $r_0$ is calculated from the return values of the embedded predecessor operations of DEL nodes in $\Druall$, or from the keys of update nodes in a list $L$. 
DEL nodes in $\Druall$ have keys less than $y$ according to line~\ref{ln:TraverseRUall:keyCheck}. The embedded predecessor operations of DEL nodes in $\Druall$ are the return values of completed instances of \textsc{PredHelper}$(x)$, for some key in $x < y$. By assumption that all completed instances of \textsc{PredHelper} satisfy Property~\ref{prop:largest}, \textsc{PredHelper}$(x)$ returns a value less than $x$, which is also less than $y$. Only notifications of these embedded predecessor operations with input keys less than $y$ are considered (on lines~\ref{ln:pred:L} and \ref{ln:pred:L'}). So $r_0$ is a value less than $y$.
It follows from the code on lines~\ref{ln:pred:L} and line~\ref{ln:pred:L'} that the update nodes that added to $L$ have keys less than $y$.
So all candidate return values of $pOp$ are less than $y$.
\end{proof}



\subsubsection{Properties of $\Iuall \cup \Inotify \cup (\Duall-\Druall) \cup (\Dnotify - \Druall)$}\label{section_crv_2}

In this section, we prove properties of the update nodes in the set $\Iuall \cup \Inotify \cup (\Duall-\Druall) \cup (\Dnotify - \Druall)$. These are update nodes created by $S$-modifying update operations.
Most notably, we prove 
that the candidate return values determined from keys of update nodes in this set satisfy Property~\ref{prop:configC}.



We define several important configurations that occur during $pOp$.
Recall that the update nodes in the $\RUALL$ are arranged in decreasing order of their keys.
During its traversal of the $\RUALL$, when $pOp$ first encounters an update node, it 
atomically reads the pointer to that update node and writes it into $\pNode.\RuallPosition$.
Let $C_{\leq x}$ be the configuration immediately after $pOp$ first encounters an update node with key less than or equal to $x$ and let $C_{<x}$ be the configuration immediately after $pOp$ first encounters an update node with key less than $x$.
Note that, for any two keys $w < x$, 
$C_{<x}$ occurs at or before $C_{\leq w}$ 
and $C_{\leq w}$ occurs at or before $C_{<w}$.
The configuration $C_T$ occurs immediately before $pOp$ starts its traversal of the relaxed binary trie. 
Since $pOp$ performs \textsc{TraverseRUAll} before the start of its traversal of the relaxed binary trie, it follows that for any key $x$, $C_{<x}$ occurs before $C_T$. 
Let $C_{\mathit{notify}}$ be the configuration immediately after $pOp$ reads the head pointer to the first notify node in $\pNode.\notifyList$ (on line~\ref{ln:pred:newlist}). 
Since $pOp$ performs its traversal of the relaxed binary trie before it reads the head pointer to the first notify node in $\pNode.\notifyList$, $C_{T}$ occurs before $C_{\mathit{notify}}$.
A notify node is seen by $pOp$ when it traverses $\pNode.\notifyList$ (on line~\ref{ln:pred:newlist}) if and only if it was added to $\pNode.\notifyList$ before $C_{\mathit{notify}}$.
The order of these configurations is summarized in the following observation.

\begin{observation}\label{lemma:C_x}
Let $w$ and $x$ be any two keys in $U$ where $w < x$. Then the following statements hold:
\begin{enumerate}
	\item[(a)] $C_{<x}$ occurs at or before $C_{\leq w}$.
	\item[(b)] $C_{\leq w}$ occurs at or before $C_{<w}$.
	\item[(c)] $C_{<w}$ occurs before $C_T$.
	\item[(d)] $C_{T}$ occurs before $C_{\mathit{notify}}$.
\end{enumerate}
\end{observation}





The next lemmas considers several cases depending on how $pOp$ learns about each of its candidate return values.
In each case, we show that the candidate return value is in $S$ sometime during $pOp$. Recall that $x \in S$ if and only if the first activated update node in the $\latest[x]$ list is an INS node.

Recall that \textsc{TraverseUall} returns two sets of update nodes, $\Iuall$ and $\Duall$.
The next lemma states that the INS nodes in $\Iuall$  have keys in $S$ sometime during the traversal, while the DEL nodes in $\Duall$ have keys not in $S$ sometime during the traversal. 

\begin{lemma}\label{lemma:traverseUALL}
For each $\uNode \in \Iuall \cup \Duall$, there is a configuration $C$ during $pOp$'s traversal of the $\UALL$ (in its instance of \textsc{TraverseUall}) in which $\uNode$ is the first activated update node in the $\latest[\uNode.\key]$ list. Furthermore, $C$ occurs before $pOp$ encounters any update nodes with key greater than $\uNode.\key$ during this traversal of the $\UALL$.
\end{lemma}

\begin{proof}
Consider an update node $\uNode \in \Iuall \cup \Duall$.
It follows from the code that there is an iteration during \textsc{TraverseUall} where $\textsc{FirstActivated}(\uNode) = \textsc{True}$ on line~\ref{ln:traverseUALL:findLatest}. By Lemma~\ref{lemma:firstActivated}, there is a configuration $C$ during this instance of $\textsc{FirstActivated}$ in which $\uNode$ is the first activated update node in the $\latest[\uNode.\key]$ list. Since $\UALL$ is sorted by increasing key, $C$ occurs before $pOp$ encounters any update nodes with key greater than $\uNode.\key$ during this instance of \textsc{TraverseUall}.

\end{proof}



An \textsc{Insert}$(x)$ operation may notify $pOp$ about the INS node, $\iNode$, it created. In this case, $\iNode$ is added to $\Inotify$ on line~\ref{ln:pred:I2}.
We show that $x \in S$ some time after $pOp$ begins accepting notifications with key $x$, but before $pOp$ begins collecting its notifications.

\begin{lemma}\label{lemma:I2_1}		

Consider an INS node $\iNode \in \Inotify$ with key $x$. Suppose the  \textsc{Insert}$(x)$ operation that created $\iNode$ notified $pOp$. Then there is a configuration between $C_{\leq x}$ and $C_{\mathit{notify}}$ in which $x \in S$. Furthermore, $\iNode$ is the first activated update node in the $\latest[x]$ list in this configuration.
\end{lemma}

\begin{proof}
Let $iOp$ be the \textsc{Insert}$(x)$ operation that created $\iNode$. 
The INS node $\iNode$ is added to $\Inotify$ on line~\ref{ln:pred:I2} because $iOp$ notifies $pOp$ about its own operation.
So $iOp$ or a \textsc{Delete}$(x)$ operation helping $iOp$ (on line~\ref{ln:delete:help_notify}) notified $pOp$ by adding a notify node $\nNode$ into $\pNode.\notifyList$ where $\nNode.\updateNode = \iNode$. Let $uOp$ be the update operation that successfully added $\nNode$ into $\pNode.\notifyList$ using CAS on line~\ref{ln:sendNotification:CAS}.
In the line of code prior, $uOp$ successfully checks that $\iNode$ is the first activated update node in the $\latest[x]$ list during \textsc{FirstActivated}$(\iNode)$ on line~\ref{ln:sendNotification:firstActivated}.
By Lemma~\ref{lemma:firstActivated}, there is a configuration $C$ during \textsc{FirstActivated}$(\iNode)$ in which $x \in S$. 
Furthermore, $x \in S$ in all configurations from when $uOp$ is linearized to $C$. In particular, $x \in S$ in the configuration $uOp$ reads that $\pNode$ is in $\PALL$ on line~\ref{ln:notify:for_loop}.
Since $pOp$ is active when it has a predecessor node in $\UALL$, $x \in S$ sometime during $pOp$ and when $uOp$ is traversing $\PALL$.

Finally, since $\iNode \in \Inotify$, 
it follows from line~\ref{ln:pred:notifyThresholdINS} that $\nNode.\notifyThreshold \leq \nNode.\key$.
This means that when $uOp$ reads $\pNode.\RuallPosition$ on line~\ref{ln:notify:notifyThreshold}, $\pNode.\RuallPosition$ is a pointer to an update node with key less than or equal $x$, and hence the read occurs after $C_{\leq x}$. This read occurs before $uOp$'s instance of \textsc{FirstActivated}$(\iNode)$
(the one it performed on line~\ref{ln:sendNotification:firstActivated}),
and hence before $C$. So $C$ occurs sometime after $C_{\leq x}$.
\end{proof}


Alternatively, $pOp$ may be notified about $\iNode$ by
an \textsc{Insert}$(w)$ or \textsc{Delete}$(w)$ operation, $uOp$, for some $w < x$, that included a pointer to $\iNode$ in the notify node it added to $pOp$'s notify list.
This happens when 
$\iNode$ has the largest key less than $y$ among the INS nodes returned by $uOp$'s instance of \textsc{TraverseUall} on line~\ref{ln:notify:travuall}.
Later, when $pOp$ reads the notify node $uOp$ added to $\pNode.\notifyList$, $pOp$ will add $\iNode$ to $\Inotify$ on line~\ref{ln:pred:updateNodeMax}.


\begin{lemma}\label{lemma:I2_2}		
Consider an INS node, $\iNode \in \Inotify$, with key $x$. Suppose an update operation with key $w$ notified $pOp$ about $\iNode$, where $w < x < y$. Then there is a configuration between $C_{\leq x}$ and $C_{\mathit{notify}}$ in which $x \in S$. Furthermore, $\iNode$ is the first activated update node in the $\latest[x]$ list in this configuration.
\end{lemma}

\begin{proof}
From the code, $\iNode$ is added to $\Inotify$ on line~\ref{ln:pred:updateNodeMax}. So there exists a notify node, $\nNode$, in $\pNode.\notifyList$ where $\iNode = \nNode.\updateNodeMax$. 
Let $uOp$ be the update operation with key $w$ that created $\nNode$, and hence is the update operation that notified $pOp$ about $\iNode$. Let $\uNode$ be the update node created by $uOp$.

By the code on line~\ref{ln:pred:updateNodeMax}, $\nNode.\notifyThreshold = -\infty$ and $\uNode \notin \Iruall \cup \Druall$.
We first argue that $uOp$ is linearized sometime at or after $C_{\leq x}$. Suppose, for contradiction, that $uOp$ is linearized before $C_{\leq x}$. Since $\nNode.\notifyThreshold = -\infty$, $uOp$ reads that $\pNode.\RuallPosition$ points to an update node with key $-\infty$ on line~\ref{ln:notify:notifyThreshold}, which means $pOp$ has completed its traversal of the $\RUALL$ by this point.
Later, on line~\ref{ln:sendNotification:firstActivated}, $uOp$'s instance of \textsc{FirstActivated}$(\uNode)$ returns \textsc{True}. By Lemma~\ref{lemma:firstActivated}, $\uNode$ is the first activated update node in the $\latest[w]$ list sometime during this instance of \textsc{FirstActivated}.
Therefore, $\uNode$ is first activated update node in the $\latest[w]$ list and is in the $\RUALL$ from $C_{\leq x}$ to at least the end of $pOp$'s traversal of the $\RUALL$. By Lemma~\ref{lemma:firstActivated_False}, any instance of \textsc{FirstActivated}$(\uNode)$ performed by $pOp$ returns \textsc{True} during its traversal of the $\RUALL$. It follows from the code of \textsc{TraverseRUall} that $pOp$ will encounter $\uNode$ when it traverses the $\RUALL$ and add it to $\Iruall \cup \Druall$, a contradiction.

By definition, $\iNode = \nNode.\updateNodeMax$ is the update node with largest key less than $y$ returned by $uOp$'s instance of \textsc{TraverseUall} on line~\ref{ln:notify:travuall}, which occurs sometime between $C_{\leq x}$ and $C_{\mathit{notify}}$.
It follows from the code that there is an iteration during \textsc{TraverseUall} where $\textsc{FirstActivated}(\iNode) = \textsc{True}$ on line~\ref{ln:traverseUALL:findLatest}. By Lemma~\ref{lemma:firstActivated}, there is a configuration $C$ during this instance of $\textsc{FirstActivated}$ in which $\iNode$ is the first activated update node in the $\latest[x]$ list, so  $x \in S$. 
Since \textsc{TraverseUall} occurs after $uOp$ is linearized but before it begins sending notifications, $C$ is between $C_{\leq x}$ and $C_{\mathit{notify}}$.
\end{proof}



Recall that prior to traversing the relaxed binary trie, $pOp$ first traverses the $\RUALL$ to find DEL nodes of \textsc{Delete} operations that may have been linearized before the start of $pOp$. 
Suppose  $\dNode$ is the DEL node of a \textsc{Delete} operation with key less than $y$ that is linearized before $pOp$.
If $pOp$ does not encounter $\dNode$ when it traverses the $\RUALL$, then $\dNode$ was removed from $\RUALL$ before $pOp$ could encounter it. 
In this case, $pOp$ will not accept any notifications about $\dNode$ and $pOp$ will not encounter $\dNode$ in the $\UALL$. 




\begin{lemma}\label{lemma:D0}
Let $dOp$ be an $S$-modifying \textsc{Delete}$(x)$ operation for some key $x < y$, and let $\dNode$ be the DEL node created by $dOp$. 
If $dOp$ is linearized before $C_{< x}$, then either $\dNode \in \Druall$ or $\dNode \notin \Duall \cup \Dnotify$.
\end{lemma}

\begin{proof}
Since $dOp$ is linearized before $C_{< x}$, $\dNode$ is in $\RUALL$
before $C_{< x}$.
Suppose $pOp$ does not encounter $\dNode$ when it traverses the $\RUALL$. So $dOp$ removed its $\dNode$ from the $\RUALL$ before 
$pOp$ sets $\pNode.\RuallPosition$ to an update node with key less than $x$, and hence before $C_{< x}$.
Hence, if $uOp$ adds a notify node $\nNode$ to $\pNode.\notifyList$, it will write
$\nNode.\notifyThreshold \geq x$ on line~\ref{ln:notify:notifyThreshold}.
By line~\ref{ln:pred:notifyThreshold}, $\dNode \notin \Dnotify$. Furthermore, $\dNode$ will not be encountered when $pOp$ traverses the $\UALL$ because $\dNode$ is removed from the $\UALL$ before the $\RUALL$ (by line~\ref{ln:delete:remove_uall_ruall}), so $\dNode \notin \Duall$.

Suppose  $pOp$ encounters $\dNode$ when it traverses the $\RUALL$. If $pOp$'s instance of \textsc{FirstActivated}$(\dNode)$ on line~\ref{ln:traverseRevUALL:findLatest} returns \textsc{True}. Then $\dNode \in \Druall$.
So $pOp$'s instance of \textsc{FirstActivated}$(\dNode)$ on line~\ref{ln:traverseRevUALL:findLatest} returns \textsc{False}. By Lemma~\ref{lemma:firstActivated_False}, there is a configuration $C$ during this instance of \textsc{FirstActivated}$(\dNode)$ in which $\dNode$ is not the first activated update node in the $\latest[x]$ list. 
So there is an \textsc{Insert}$(x)$ operation linearized sometime between when $dOp$ is linearized and $C$, and hence before $pOp$ completes \textsc{TraverseRUAll}. Since $\dNode$ is no longer the first activated update node in the $\latest[x]$ list when it traverses the $\UALL$, Lemma~\ref{lemma:traverseUALL} implies that $\dNode \notin \Duall$.

Before $dOp$ adds a notify node $\nNode$ to $\pNode.\notifyList$, its instance of \textsc{FirstActivated}$(\dNode)$ on line~\ref{ln:sendNotification:firstActivated} returns \textsc{True}. By Lemma~\ref{lemma:firstActivated}, $\dNode$ is the first activated update node in the $\latest[x]$ list sometime before $C$, and hence before $C_{< x}$. From the code, $dOp$ previously wrote
$\nNode.\notifyThreshold \geq x$ on line~\ref{ln:notify:notifyThreshold}.
By line~\ref{ln:pred:notifyThreshold}, $\dNode \notin \Dnotify$.
\end{proof}

The next two lemmas show that  the key of each DEL node, $\dNode$, in $\Duall - \Druall$ or $\Dnotify - \Druall$, is in $S$ sometime during $pOp$. 
Both of these results use Lemma~\ref{lemma:D0} to argue that the \textsc{Delete}$(x)$ operation that created $\dNode$ was linearized sometime after $C_{< x}$. In the configuration immediately before this \textsc{Delete} operation was linearized, $x \in S$. 

\begin{lemma}\label{lemma:D1}
Consider a DEL node $\dNode \in  \Duall - \Druall$ with key $x$. There is a configuration $C$ after $C_{< x}$ in which $x \in S$. Furthermore, $C$ occurs before $pOp$ encounters any update nodes with key greater than $x$ during its traversal of the $\UALL$.
\end{lemma}

\begin{proof}
Let $dOp$ be the creator of $\dNode$.
By Lemma~\ref{lemma:traverseUALL}, there is a configuration $C$ during \textsc{TraverseUall} of $pOp$ in which $dOp$ is the latest \textsc{Delete}$(x)$ operation, and hence $x \notin S$. 
Suppose, for contradiction, that $x \notin S$ in all configurations from $C_{< x}$ to $C$. This implies that $dOp$ was linearized before $C_{< x}$. By Lemma~\ref{lemma:D0}, $\dNode \in \Druall$ or $\dNode \notin \Duall \cup \Dnotify$. This contradicts the fact that $\dNode \in \Duall - \Druall$.
So there exists a configuration after $C_{< x}$ in which $x \in S$.

\end{proof}




\begin{lemma}\label{lemma:D2}

Consider a DEL node $\dNode \in \Dnotify - \Druall$ with key $x$. There is a configuration between $C_{< x}$ and $C_{\mathit{notify}}$ in which $x \in S$.
\end{lemma}

\begin{proof}
Let $dOp$ be the \textsc{Delete}$(x)$ operation that created $\dNode$. Since $\dNode \in \Dnotify$, $dOp$ successfully added a notify node $\nNode$ to $\pNode.\notifyList$. By line~\ref{ln:pred:notifyThreshold}, $\nNode.\notifyThreshold < x$. 
This means that when $dOp$ read $\pNode.\RuallPosition$ on line~\ref{ln:notify:notifyThreshold} it pointed to an update node with key less than $x$, and \textsc{FirstActivated}$(\dNode)$ on line~\ref{ln:sendNotification:firstActivated} returned \textsc{True}. By Lemma~\ref{lemma:firstActivated}, there is a configuration $C$ during \textsc{FirstActivated} in which $\dNode$ is the first activated update node in the $\latest[x]$ list.
Since $\dNode \notin \Druall$ and $\dNode$ is the first activated update node the $\latest[x]$ list in $C$, 
it follows 
by Lemma~\ref{lemma:D0} that $dOp$ is linearized sometime after $C_{< x}$.
In the configuration immediately before $dOp$ is linearized, $x \in S$.
\end{proof}





The next two observations describe intervals in which update operations notify predecessor operations.
An \textsc{Insert} operation may notify concurrent predecessor operations directly, or may be helped to do so by a \textsc{Delete} operation with the same key that is linearized after it.



\begin{lemma}\label{lemma:ins_notify}
Let $iOp$ be an $S$-modifying \textsc{Insert}$(x)$ operation for some key $x < y$, and let $\iNode$ be the INS node created by $iOp$.
Then $\iNode$ is in the $\UALL$ and is the first activated update node in the $\latest[x]$ list in all configurations from when $iOp$ is linearized until some update operation completes an instance of \textsc{NotifyPredOps}$(\iNode)$ (on line~\ref{ln:insert:notify} or \ref{ln:delete:help_notify}).
\end{lemma}

\begin{proof}
By definition, $iOp$ is linearized when the status of $\iNode$ changes from inactive to active by the CAS on line~\ref{ln:help_activate:CAS}. Immediately after this CAS, $\iNode$ is the first activated update node in the $\latest[x]$ list and is an activated update node in $\UALL$. From the code $iOp$ does not remove $\iNode$ from $\UALL$ until after it completes \textsc{NotifyPredOps}$(\iNode)$ (on line~\ref{ln:insert:notify}).

Let $dOp$ be any \textsc{Delete}$(x)$ operation whose instance of \textsc{FindLatest}$(x)$ on line~\ref{ln:delete:findLatest} returns $\iNode$. 
By Lemma~\ref{lemma:findLatestConfig}, $iOp$ is the latest update operation with key $x$ sometime after the start of $dOp$.
Before $dOp$ is linearized, $dOp$ helps perform \textsc{NotifyPredOps}$(\iNode)$ on line~\ref{ln:delete:help_notify}. 
Once $dOp$ is linearized, $\iNode$ is no longer the first activated update node in the $\latest[x]$ list.
\end{proof}

\textsc{Delete} operations notify concurrent predecessor operations directly. They are not helped by \textsc{Insert} operations.


\begin{lemma}\label{lemma:del_notify}
Let $dOp$ be an $S$-modifying \textsc{Delete}$(x)$ operation for some key $x < y$, and let $\dNode$ be the DEL node created by $dOp$.
Then $\dNode$ is in the $\UALL$ and is the first activated update node in the $\latest[x]$ list in all configurations from when $dOp$ is linearized until either
\begin{itemize}
	\item $dOp$ completes an instance of \textsc{NotifyPredOps}$(\dNode)$ on line~\ref{ln:delete:notifyPredOps}, or
	\item an $S$-modifying \textsc{Insert}$(x)$ operation is linearized after $dOp$.
\end{itemize}
\end{lemma}

\begin{proof}
By definition, $dOp$ is linearized when it the status of $\iNode$ changes from inactive to active by the  line~\ref{ln:help_activate:CAS}. Immediately after this CAS, $\dNode$ is the first activated update node in the $\latest[x]$ list and is an activated update node in $\UALL$.

Suppose an $S$-modifying \textsc{Insert}$(x)$ operation, $iOp$, is linearized after $dOp$ before $dOp$ invokes and completes \textsc{NotifyPredOps}$(\dNode)$ (on line~\ref{ln:delete:notifyPredOps}). Only after $iOp$ is linearized is $dOp$ no longer the first activated update node in the $\latest[x]$ list. Furthermore, $dOp$ does not remove $\dNode$ from $\UALL$ until after it invokes and completes \textsc{NotifyPredOps}$(\dNode)$.
Then $\dNode$ is the first activated update node in the $\latest[x]$ list and is an activated update node in $\UALL$ for all configurations starting when $dOp$ is linearized and ending when $iOp$ is linearized.

So suppose $dOp$ invokes and completes \textsc{NotifyPredOps}$(\dNode)$ (on line~\ref{ln:delete:notifyPredOps}) before an $S$-modifying \textsc{Insert}$(x)$ operation is linearized after $dOp$.
So $\dNode$ remains the first activated update node in the $\latest[x]$ list until after $dOp$ invokes and completes \textsc{NotifyPredOps}$(\dNode)$.
Furthermore, $dOp$ does not remove $dNode$ from $\UALL$ until after it invokes and completes \textsc{NotifyPredOps}$(\dNode)$. Then $\dNode$ is the first activated update node in the $\latest[x]$ list and is an activated update node in $\UALL$ for all configurations starting when $dOp$ is linearized and ending when $iOp$ is linearized.
\end{proof}

Next, we consider two different scenarios in which an update operation $uOp$ with key $x$ is linearized during $pOp$. In both, we show that $x$ is a candidate return value of $pOp$.
During the proof, we use the previous two lemmas to guarantee that an update node with  key $x$ will either be seen when $pOp$ traverses the $\UALL$, or some update operation with key $x$ notifies $pOp$ sometime before $pOp$ completes its traversal of the $\UALL$. 

\begin{lemma}\label{lemma:del_before_CT}
Let $dOp$ be an $S$-modifying \textsc{Delete}$(x)$ operation, for some key $x < y$, that is linearized sometime between $C_{<x}$ and $C_T$. Then $\Iuall \cup \Inotify \cup (\Duall-\Druall) \cup (\Dnotify - \Druall)$ contains an update node with key $x$ and $x$ is a candidate return value of $pOp$.
\end{lemma}

\begin{proof}
It follows from  Lemma~\ref{lemma:del_notify} and Lemma~\ref{lemma:ins_notify} that $\UALL$ contains an update node with key $x$ that is the first activated update node in the $\latest[x]$ list from when $dOp$ is linearized until either $dOp$ completes an instance of $\textsc{NotifyPredOps}(\dNode)$ for its own DEL node $\dNode$ with key $x$, or some update operation completes $\textsc{NotifyPredOps}(\iNode)$ for some INS node with key $\iNode$ of an \textsc{Insert}$(x)$ operation linearized after $dOp$. Let $uOp'$ be the update operation with key $x$ that completes first completes an instance of \textsc{NotifyPredOps} (for either $\dNode$ or $\iNode$). Note that since $uOp'$ is linearized after $C_{<x}$, $\dNode$ or $\iNode$ are not encountered when $pOp$ traverses the $\RUALL$, and hence are not in $\Druall$.

Suppose $uOp'$ completes the instance of \textsc{NotifyPredOps} before $pOp$ completes its traversal of the $\UALL$. 
Note that $\pNode$ is inserted into $\PALL$ before the start of $pOp$, and is not removed from $\PALL$ until sometime after $pOp$ completes its traversal of the $\UALL$. So $\pNode$ is in $\PALL$ throughout  $uOp'$'s traversal of the $\PALL$ on line~\ref{ln:notify:for_loop}.
By lemma~\ref{lemma:firstActivated_False}, all instances of \textsc{FirstActivated} on line~\ref{ln:notify:first_activated} and \ref{ln:sendNotification:firstActivated} return \textsc{True}.
So consider the notify node $\nNode$ that $uOp'$ eventually adds to $\pNode.\notifyList$. 
Since $uOp'$ is linearized after $C_{<x}$,
it follows that $\pNode.\RuallPosition$ points to an update node with key less than $x$ throughout  $uOp'$  traversal of the $\PALL$. So $\nNode.\notifyThreshold < x$.
So when $pOp$ encounters $\nNode$ on line~\ref{ln:pred:newlist}, either $\dNode$ or $\iNode$ is added to  $\Inotify \cup \Dnotify$ on line~\ref{ln:pred:I2} or \ref{ln:pred:D2}.
So $\Inotify \cup (\Dnotify - \Druall)$ contains an update node with key $x$.  Hence, $x$ is a candidate return value of $pOp$.

So suppose $pOp$ completes its traversal of the $\UALL$ before $uOp'$ completes the instance of \textsc{NotifyPredOps}.
Since no update operation with key $x$ notifies $pOp$, no latest update node with key $x$ is removed from $\UALL$ before $pOp$ completes its traversal of the $\UALL$.
It follows that $pOp$ encounters an activated update node $\uNode'$ with key $x$ during its traversal of the $\UALL$, where  $\uNode'$ is either $\iNode$ or $\dNode$.
By Lemma~\ref{lemma:firstActivated_False}, \textsc{FirstActivated}$(\uNode')$ returns \textsc{True} on line~\ref{ln:traverseUALL:findLatest}.
It follow from the code of \textsc{TraverseUall} that $pOp$ adds  $\uNode'$ to $\Iuall \cup \Duall$ on line~\ref{ln:pred:uall2}.
Since
$\uNode' \in \Iuall \cup (\Duall - \Druall)$, $x$ is a candidate return value of $pOp$.
\end{proof}

\begin{lemma}\label{lemma:ins_after_CT}
Let $iOp$ be an $S$-modifying \textsc{Insert}$(x)$ operation, for some key $x < y$ that is linearized sometime after $C_{\leq x}$, but before $pOp$ encounters any update nodes with key at least $x$ during its instance of $\textsc{TraverseUall}(y)$. Then the INS node created by $iOp$ is in $\Iuall \cup \Inotify$ and $x$ is a candidate return value of $pOp$.
\end{lemma}

\begin{proof}
Suppose $iOp$ (or a helper of $iOp$) completes an instance of \textsc{NotifyPredOps}$(\iNode)$ before $pOp$ completes its traversal of the $\UALL$. 
Note that $\pNode$ is inserted into $\PALL$ before the start of $pOp$, and is not removed from $\PALL$ until sometime after $pOp$ completes its traversal of the $\UALL$. So $\pNode$ is in $\PALL$ throughout  $iOp$'s traversal of the $\PALL$ on line~\ref{ln:insert:notify}. Since $iOp$ is linearized after $C_{\leq x}$ 
it follows that $\pNode.\RuallPosition$ points to an update node with key less than or equal to $x$ throughout  $iOp$  traversal of the $\PALL$. 
When $iOp$ encounters $\pNode$ in the $\PALL$, it creates a notify node, $\nNode$, where $\nNode.\updateNode = \iNode$ and $\nNode.\notifyThreshold \leq x$.
By Lemma~\ref{lemma:ins_notify}, $iNode$ is the first activated update node in the $\latest[x]$ list throughout \textsc{NotifyPredOps}. So by Lemma~\ref{lemma:firstActivated_False}, $iOp$'s instance of \textsc{FirstActivated}$(\iNode)$ returns \textsc{True}, and $iOp$ successfully adds $\nNode$ to $\pNode.\notifyList$.
It follows by line~\ref{ln:pred:notifyThresholdINS} and \ref{ln:pred:I2} that $\iNode \in \Inotify$ when $pOp$ finishes traversing its notify list.  Hence, $x$ is a candidate return value of $pOp$.

So suppose $pOp$ completes its traversal of the $\UALL$ before $iOp$ (or a helper of $iOp$) completes an instance of \textsc{NotifyPredOps}$(\iNode)$.
It follows by Lemma~\ref{lemma:ins_notify} that $\iNode$ is the first activated update node in the $\latest[x]$ list throughout $pOp$'s traversal of the $\UALL$.
So $pOp$ encounters $\iNode$ during its traversal of the $\UALL$. Furthermore, by Lemma~\ref{lemma:firstActivated_False}, \textsc{FirstActivated}$(\iNode)$ returns \textsc{True} on line~\ref{ln:traverseUALL:findLatest}.
So
$\iNode \in \Iuall$ when \textsc{TraverseUall} returns on line~\ref{ln:pred:uall2}. Hence, $x$ is a candidate return value of $pOp$.
\end{proof}

The following lemma considers a scenario in which an \textsc{Insert}$(w)$ or \textsc{Delete}$(w)$ operation, $uOp$, notifies $pOp$ about an INS node with key $x$, where $x > w$. 
We show that $pOp$ has a candidate return value that is at least $x$.
During the proof, we use Lemma~\ref{lemma:ins_notify} to guarantee that 
either the \textsc{Insert}$(x)$ operation that created the INS node with key $x$ will notify $pOp$ about its INS node before $C_{\mathit{notify}}$, or $uOp$ will see this INS node when it traverses the $\UALL$ and, hence, include an INS node with key at least $x$ when it notifies $pOp$. 

\begin{lemma}\label{lemma:updateNodeMax_candidate}
Let $iOp$ be  an $S$-modifying \textsc{Insert}$(x)$ operation,  for some key $x < y$. 
Let $uOp$ be an $S$-modifying \textsc{Insert}$(w)$ or \textsc{Delete}$(w)$ operation linearized after $C_T$ that notifies $pOp$ before $C_{\mathit{notify}}$, for some key $w < x < y$.
Suppose $iOp$ is linearized after $C_T$, but before $uOp$ encounters any update nodes with key at least $x$ during its instance of $\textsc{TraverseUall}$ (on line~\ref{ln:notify:travuall}). 
Then $pOp$ has a candidate return value $x'$, where $w < x \leq x' < y$.
\end{lemma}

\begin{proof}
Let $\iNode$ be the INS node created by $iOp$, and let $\uNode$ be the update node created by $uOp$.
By Lemma~\ref{lemma:ins_notify}, $\iNode$ is the first activated update node in the $\latest[x]$ list in all configurations from when $iOp$ is linearized to when it first completes \textsc{NotifyPredOps}$(\iNode)$. It follows by Lemma~\ref{lemma:firstActivated_False} that all instances of \textsc{FirstActivated}$(\iNode)$ during this instance of \textsc{NotifyPredOps}$(\iNode)$ return \textsc{True}. 
Note that $\pNode$ is a predecessor node in $\PALL$ in all configurations from $C_T$ to $C_{\mathit{notify}}$. So it follows from the code of $\textsc{NotifyPredOps}$ that $iOp$ will successfully add a notify node, $\nNode$, to $\pNode.\notifyList$. 
Furthermore, since $iOp$ is linearized after $C_T$, $\nNode.\notifyThreshold$ is set to $-\infty$ on line~\ref{ln:notify:notifyThreshold}.
So $\iNode \in \Inotify$ after $pOp$ performs line~\ref{ln:pred:I2}. Then $pOp$ has a candidate return value $x$.

So suppose does not attempt to notify $pOp$ before $C_{\mathit{notify}}$. Then  Lemma~\ref{lemma:ins_notify} implies that $iOp$ is the first activated update node in the $\latest[x]$ list in all configurations from when $iOp$ is linearized to $C_{\mathit{notify}}$. Since $iOp$ is linearized before $uOp$ encounters any update nodes with key greater than or equal to $x$ during its instance of \textsc{TraverseUall} (on line~\ref{ln:notify:travuall}), it follows that $uOp$ encounters $\iNode$ during its traversal of the $\UALL$. Furthermore, by  Lemma~\ref{lemma:firstActivated_False}, $uOp$'s instance of \textsc{FirstActivated}$(\iNode)$ during \textsc{TraverseUall} (on line~\ref{ln:traverseUALL:findLatest}) returns \textsc{True}. This implies that when $iOp$ notifies $pOp$ by adding a notify node $\nNode$ into $\pNode.\notifyList$, $\nNode.\updateNodeMax$ contains a pointer to an INS node with key $x'$, where $w < x \leq x' < y$. 
Since $uOp$ is linearized after $C_T$, $\nNode.\notifyThreshold$ is set to $-\infty$ on line~\ref{ln:notify:notifyThreshold} and $pOp$ does not encounter $\uNode$ during its traversal of the $\RUALL$. So by line~\ref{ln:pred:updateNodeMax}, an INS node with key $x'$ is added to $\Inotify$. Then $pOp$ has a candidate return value $x'$.

\end{proof}

Recall that the keys of update nodes in $\Iuall \cup \Inotify \cup (\Duall-\Druall) \cup (\Dnotify - \Druall)$ are candidate return values. We show that Property~\ref{prop:configC} is satisfied by these values.
\begin{lemma}\label{lemma:prop:configC}
Let $w$ be the key of an update node in $\Iuall \cup \Inotify \cup (\Duall-\Druall) \cup (\Dnotify - \Druall)$. 
Then there is a configuration $C$ after $pOp$ is announced but before the end of $pOp$ such that
\begin{enumerate}
	\item[(a)] 
	$w \in S$ in $C$,
	
	\item[(b)]
	if $C$ occurs before $C_T$ and
	there exists an $S$-modifying \textsc{Delete}$(x)$ operation linearized between $C$ and $C_T$
	with $w < x < y$,  then $pOp$ has a candidate return value which is at least $x$, and
	
	
	\item[(c)]     
	if $C$ occurs at or after $C_T$ and 
	there exists an $S$-modifying \textsc{Insert}$(x)$ operation 
	linearized between $C_T$ and $C$ with  $w < x < y$, then $pOp$ has a candidate return value which is at least $x$.
	
\end{enumerate}
\end{lemma}

\begin{proof}

First suppose $w$ is the key of an INS node $\iNode \in \Iuall$. Let $iOp_w$ be the \textsc{Insert}$(w)$ operation that created $\iNode$.  
By Lemma~\ref{lemma:traverseUALL}, there is a configuration $C$ during $pOp$'s traversal of $\UALL$ in which $w \in S$. Furthermore, $C$ occurs after $C_T$ but before $pOp$ encounters any update nodes with key $x$ during its traversal of the $\UALL$.
Suppose, for contradiction, that there is an \textsc{Insert}$(x)$ operation, $iOp_x$, linearized between $C_T$ and $C$, where $w < x < y$.
By Lemma~\ref{lemma:ins_after_CT}, $x$ is a candidate return value of $pOp$.

Now consider an update node in $\uNode \in \Inotify \cup (\Duall-\Druall) \cup (\Dnotify - \Druall)$, where $\uNode.\key = w < y$.
Suppose there is a configuration $C$ between $C_{\leq w}$ and $C_T$ in which $w \in S$. 
Suppose, for contradiction, that there is a \textsc{Delete}$(x)$ operation, $dOp_x$, linearized between $C$ and $C_T$, where $w < x < y$. 
By Lemma~\ref{lemma:del_before_CT},  $x$ is a candidate return value of $pOp$.

So there is no configuration between $C_{\leq w}$ and $C_T$ in which $w \in S$. 
By Lemma~\ref{lemma:I2_1}, Lemma~\ref{lemma:I2_2}, Lemma~\ref{lemma:D1}, and Lemma~\ref{lemma:D2}, there is a configuration $C$ during $pOp$ after $C_{\leq w}$ in which $w \in S$. So this configuration occurs after $C_T$.
\begin{itemize}
	\item Suppose $\uNode \in \Inotify$. Let $iOp_w$ be the \textsc{Insert}$(w)$ operation that created $\uNode$. Let $C$ be the configuration immediately after $iOp_w$ is linearized. Note that Lemma~\ref{lemma:I2_1}, by Lemma~\ref{lemma:I2_2}, $\uNode$ is the first activated update node in the $\latest[w]$ list in $C$. Since $w \in S$ in all configurations from when $iOp_w$ is linearized to $C$,  $iOp_w$ is linearized after $C_T$ and $C$ occurs after $C_T$.
	
	First, suppose that $\uNode \in \Inotify$ because $iOp_w$ notified $pOp$.
	Suppose, for contradiction, that there is an \textsc{Insert}$(x)$ operation, $iOp_x$, linearized between $C_T$ and $C$, where $w < x < y$.
	So $iOp_x$ is linearized before $iOp_w$. Since $iOp_w$ notifies $pOp$, it follows from Lemma~\ref{lemma:updateNodeMax_candidate} that $pOp$ has a candidate return value whose value is at least $x$.
	
	Next, suppose  $\uNode \in \Inotify$ because some \textsc{Insert}$(w')$ or \textsc{Delete}$(w')$ operation, $uOp'$, included $\uNode$ in its notification to $pOp$, where $w' < w < y$. In other words,  $uOp'$ added a notify node, $\nNode$, to $\pNode.\notifyList$ where $\nNode.\updateNodeMax = \uNode$.
	Suppose, for contradiction, that there is an \textsc{Insert}$(x)$ operation, $iOp_x$, linearized between $C_T$ and $C$, where $w < x < y$. 
	Since $\nNode.\updateNodeMax = \uNode$, $\uNode$ was a node returned by $uOp'$'s instance of \textsc{TraverseUall} on line~\ref{ln:notify:travuall}.
	By the code of \textsc{TraverseUall},  \textsc{FirstActivated}$(\uNode)$ returned \textsc{True} (on line~\ref{ln:traverseUALL:findLatest}). By Lemma~\ref{lemma:firstActivated}, there is a configuration $C'$ during this instance of \textsc{FirstActivated}$(\uNode)$ in which $\uNode$ is the first activated update node in the $\latest[w]$ list. Note that $C'$ occurs after $C$.
	Since $\UALL$ is sorted by increasing key, $C'$ occurs before $uOp$ encounters any update node with key greater than or equal to $x$ during its instance of \textsc{TraverseUall}. Since $iOp$ is linearized between $C_T$ and $C$, and hence before $C'$, it follows by Lemma~\ref{lemma:updateNodeMax_candidate} that $pOp$ has a candidate return value whose value is at least $x$.
	
	
	
	\item Suppose $\uNode \in (\Duall-\Druall)$.
	Let $dOp_w$ be the \textsc{Delete}$(w)$ operation that created $\uNode$. 
	Let $C$ be the configuration immediate before  $dOp_w$  is linearized. By Lemma~\ref{lemma:D0}, $C$ is after $C_{< w}$. Since $w \in S$ immediately before $C$, $C$ is after $C_T$.
	
	Suppose, for contradiction, that there is an \textsc{Insert}$(x)$ operation, $iOp_x$, linearized between $C_T$ and $C$, where $w < x < y$. Since $\uNode \in (\Duall-\Druall)$ and $pOp$ encounters $\uNode$ when it traverses theh $\UALL$, $C$ occurs before $pOp$ encounters any update nodes with key greater than $w$. It follows by Lemma~\ref{lemma:ins_after_CT} that  $x$ is a candidate return value of $pOp$.
	
	\item Suppose $\uNode \in (\Dnotify-\Druall)$. Let $dOp_w$ be the \textsc{Delete}$(x)$ operation that created $\uNode$. 
	Let $C$ be the configuration immediate before  $dOp_w$ is linearized. By Lemma~\ref{lemma:D0}, $C$ is after $C_{< w}$. Since $w \in S$ immediately before $C$, $C$ is after $C_T$.
	
	Suppose, for contradiction, that there is an \textsc{Insert}$(x)$ operation, $iOp_x$, linearized between $C_T$ and $C$, where $w < x < y$.  So $iOp_x$ is linearized before $iOp_w$. Since $\uNode \in (\Dnotify-\Druall)$ and $dOp_w$ notifies $pOp$, it follow from Lemma~\ref{lemma:updateNodeMax_candidate} that $pOp$ has a candidate return value whose value is at least $x$.
	
\end{itemize}
\end{proof}


\subsubsection{Our implementation satisfies Property~\ref{prop:configC} and Property~\ref{prop:trie} }\label{section_crv_3}

Besides the candidate return values considered in the previous section,
$pOp$ may compute
one additional candidate return value, which is stored in $pOp$'s local variable $r_0$.
We show that this value satisfies Property~\ref{prop:configC} and Property~\ref{prop:trie}. 

If $pOp$'s traversal of the relaxed binary trie returns a value $w \neq \bot$, then $r_0$ is set to $w$ and is a candidate return value. 
It this case,
it is straightforward to show that $w$ satisfies Property~\ref{prop:configC}.

\begin{lemma}
Suppose $pOp$'s instance of \textsc{RelaxedPredecessor}$(y)$ returns $w \neq \bot$.
Then there is a configuration $C$ during $pOp$ such that
\begin{enumerate}
	\item
	$w \in S$ in $C$,
	
	\item    
	$C$ occurs after $C_T$, and
	
	\item if there exists an \textsc{Insert}$(x)$ operation 
	linearized between $C_T$ and $C$ with  $w < x < y$, then $pOp$ has a candidate return value which is at least $x$.
	
\end{enumerate}
\end{lemma}

\begin{proof}
By Lemma~\ref{lemma:relaxed_trie_kxy}, $w \in S$ in a configuration $C$ during \textsc{RelaxedPredecessor}$(y)$. So $C$ occurs after $C_T$, but before $pOp$ begins its traversal of the $\UALL$. Suppose there is an \textsc{Insert}$(x)$ operation, $iOp_x$, linearized between $C_T$ and $C$, where $w < x < y$.
By Lemma~\ref{lemma:ins_after_CT}, $x$ is a candidate return value of $pOp$.
\end{proof}

If the traversal of the relaxed binary trie returns $\bot$ and $\Druall \neq \emptyset$, then $pOp$ assigns a value $x$ to $r_0$ on line~\ref{ln:pred:set_pred0}. If, $\Iuall \cup \Inotify \cup (\Duall-\Druall) \cup (\Dnotify - \Druall)$ only contains update nodes with key less than $x$,
$x$ is a candidate return value. 
Showing that Property~\ref{prop:configC} is satisfied by this value is more complicated.
We begin by proving a few facts about when $pOp$'s traversal of the relaxed binary trie returns $\bot$.

The next lemma  considers a case when there is an $S$-modifying update operation $uOp$ with key $x$ that is updating the relaxed binary trie while $pOp$ is traversing the relaxed binary trie. 
We show that $pOp$ will encounter an update node with key $x$ when $pOp$ traverses the $\UALL$ or when $pOp$ traverses its own notify list.
To do this, we show that $pOp$ will either traverse the $\UALL$ before $uOp$ can remove its update node from the $\UALL$, or $uOp$ will notify $pOp$ before $pOp$ removes its predecessor node from the $\PALL$.


\begin{lemma}\label{lemma:dangerOverlap}
Let $uOp$ be an $S$-modifying update operation with key $x$, where $x < y$. 
Suppose that there exists a configuration during $pOp$'s traversal of the relaxed binary trie in which $uOp$ is the latest $S$-modifying update operation with key $x$ that has been linearized and $uOp$ has not yet completed updating the relaxed binary trie.
Then $\Iuall \cup \Inotify \cup \Duall \cup \Dnotify$ contains an update node with key $x$ that was created by $uOp$ or an update operation with key $x$ linearized after $uOp$.
\end{lemma}

\begin{proof}
Since $uOp$ is an $S$-modifying update operation, the update node it created is the first activated update node in the $\latest[x]$ list when it was linearized. 
Lemma~\ref{lemma:ins_notify} and Lemma~\ref{lemma:del_notify} imply that, from when $uOp$ is linearized until some instance of $\textsc{NotifyPredOps}(\uNode)$ is completed, there is an update node in the $\UALL$ with key $x$ that is the first activated update node in the $\latest[x]$ list, where $\uNode$ is some update node with key $x$.
Note that $\uNode$ is either the update node created by $uOp$, or if $uOp$ is a \textsc{Delete}$(x)$ operation, it may be the INS node of the first $S$-modifying \textsc{Insert}$(x)$ operation linearized after $uOp$.
We consider two cases, depending on whether $pOp$ completes its traversal of the $\UALL$ on line~\ref{ln:pred:uall2} first, or if the instance of  $\textsc{NotifyPredOps}(\uNode)$ is completed first.

Suppose $pOp$ completes its traversal of the $\UALL$ before some update operation completes an instance of $\textsc{NotifyPredOps}(\uNode)$.
So there is an update node in the $\UALL$ with key $x$ that is the first activated update node in the $\latest[x]$ list throughout $pOp$'s traversal of the $\UALL$. 
So $pOp$ encounters one such update node, $\uNode'$, with key $x$ when traversing the $\UALL$ and performs \textsc{FirstActivated}$(\uNode')$ on line~\ref{ln:traverseUALL:findLatest}.
By Lemma~\ref{lemma:firstActivated_False}, \textsc{FirstActivated}$(\uNode')$  returns \textsc{True}.  So $\uNode' \in \Iuall \cup \Duall$.

So suppose some update operation completes $\textsc{NotifyPredOps}(\uNode)$ before  $pOp$ completes its traversal of the $\UALL$, and hence, before $pOp$ removes $\pNode$ from the $\PALL$. 
This instance of \textsc{NotifyPredOps}$(\uNode)$ occurs after $uOp$ is linearized, and hence after $\pNode$ is added to the $\PALL$ and after $\pNode$ sets $\pNode.\RuallPosition$ to point to the sentinel node in the $\RUALL$ with key $-\infty$. 
Throughout the instance of $\textsc{NotifyPredOps}(\uNode)$, $\uNode$ is the first activated update node in the $\latest[x]$ list. By Lemma~\ref{lemma:firstActivated_False}, any instance of \textsc{FirstActivated}$(\uNode)$ that is invoked returns \textsc{True}. It follows from the code of $\textsc{NotifyPredOps}(\uNode)$ that a successful CAS adds a notify node $\nNode$, to $\pNode.\notifyList$, where $\nNode.\updateNode = \uNode$ and $\nNode.\notifyThreshold = -\infty$.
When $pOp$ traverses its notify list (on line~\ref{ln:pred:newlist}), it will add $\uNode$ to $\Inotify \cup \Dnotify$ (either on line~\ref{ln:pred:I2} or \ref{ln:pred:D2}).
\end{proof}

For the remainder of this section, let $k$ be the largest key less than $y$ that is completely present throughout $pOp$'s traversal of the relaxed binary trie, and $-1$ if no such key exists.
The following lemma follows from Lemma~\ref{lemma:dangerOverlap} and the specification of the relaxed binary trie (Lemma~\ref{lemma_relaxed_trie_k}).

\begin{lemma}\label{lemma:d0_non_empty}
Suppose $pOp$'s instance of \textsc{RelaxedPredecessor}$(y)$ returns $\bot$. If $\Iuall \cup \Inotify \cup (\Duall - \Druall) \cup (\Dnotify - \Druall)$ 
only contains update nodes with key less than $k$,
then $\Druall$ contains an update node with key $x$, where $k < x < y$.
\end{lemma}

\begin{proof}
Since $pOp$'s instance of \textsc{RelaxedPredecessor}$(y)$ returns $\bot$,
Lemma~\ref{lemma_relaxed_trie_k} implies that there exists an $S$-modifying update operation $uOp$ with key $x$, where $k < x < y$, such that,
at some point during $pOp$'s traversal of the relaxed binary trie, $uOp$ is the last linearized $S$-modifying update operation with key $x$ and $uOp$ has not yet completed updating the relaxed binary trie.
By Lemma~\ref{lemma:dangerOverlap},  $\Iuall \cup \Inotify \cup \Duall \cup \Dnotify$ contains an update node with key $x$.
By assumption, $\Iuall \cup \Inotify \cup (\Duall - \Druall) \cup (\Dnotify - \Druall)$ does not contain an update node with key $x$, so $\Druall$ contains an update node with key $x$.






\end{proof}

Recall that $R$ is a set of keys and $L$ is a list of update nodes that are computed by $pOp$ on lines~\ref{ln:pred:fix_binary_trie_pred} to \ref{ln:pred:R_remove_Druall} in order to compute $r_0$
on \ref{ln:pred:set_pred0}.
We consider the values of these local variables at various points during this computation by $pOp$.
If $pOp$ assigns a value to $\pNode'$ on line~\ref{ln:pred:pNode'}, 
let $C_{ann}$ be the configuration immediately before the predecessor node $\pNode'$ was announced;
otherwise let $C_{ann}$ be the configuration immediately before $\pNode$ was announced. 

We next state and prove several facts about $R$, $L$, $\pNode'$, and $C_{ann}$.
From lines~\ref{ln:pred:travPALL} to \ref{ln:pred:travPALLend}, $pOp$ traverses the $\PALL$ starting from $\pNode$ to determine a list of predecessor nodes, $Q$. Predecessor nodes are only added to the head of $\PALL$. 
$Q$ is a sequence of predecessor nodes sorted in the order in which they were announced.
Furthermore, every predecessor node in $Q$ was announced before $\pNode$. 
By line~\ref{ln:pred:pNode'}, $\pNode'$ is a predecessor node in $Q$, so it was announced earlier than $\pNode$.

\begin{observation}\label{obs:pNode'_before_pNode}
If $pOp$ assigns a value to $\pNode'$ on line~\ref{ln:pred:pNode'}, then $\pNode'$ is announced before $\pNode$. 
\end{observation}

\begin{lemma}\label{obs:first_pred_Cann}
The first embedded predecessor operations of the DEL nodes in $\Druall$ are announced after $C_{ann}$ and are completed sometime before $C_T$.
\end{lemma}

\begin{proof}
Let $\pNode''$ be the predecessor node of the first embedded predecessor operation of a DEL node, $\dNode$, in $\Druall$. Recall that $Q$ is the sequence of predecessor nodes that $pOp$ encounters when it traverses the $\PALL$ starting from $\pNode$ (from lines~\ref{ln:pred:travPALL} to \ref{ln:pred:travPALLend})

Suppose $\pNode'' \notin Q$. Suppose, for contradiction, that $\pNode''$ is announced before $\pNode$.
Since $\dNode \in \Druall$, $pOp$ encounters $\dNode$ when it traverses the $\RUALL$.
Since $dOp$ removes $\dNode$ from the $\RUALL$ before $\pNode''$ is removed from the $\PALL$ (by  lines~\ref{ln:delete:remove_uall_ruall} and \ref{ln:delete:remove_pall}),
$\pNode''$ is still in the $\PALL$ when $pOp$  encounters $\dNode$ in the $\RUALL$.
This implies that when $pOp$ traverses the $\PALL$ starting from $\pNode$ (which occurs before its traversal of the $\RUALL$), it will encounter $\pNode''$, contradicting the fact that $\pNode'' \notin Q$.
So $\pNode''$ is announced after $\pNode$. If $C_{ann}$ is immediately before $\pNode$ is announced, then $\pNode''$ is announced after $C_{ann}$. If $C_{ann}$ is immediately before $\pNode'$ is announced, then Observation~\ref{obs:pNode'_before_pNode} implies $\pNode''$ is announced after $C_{ann}$.

So $\pNode'' \in Q$. Let $\mathit{predNodes}$ denote the set of predecessor nodes of the first embedded predecessor operations of DEL nodes in $\Druall$. This set is computed by $pOp$ on line~\ref{ln:pred:predNodes}.
Since $\pNode'' \in Q$, $pOp$ assigns a value to $\pNode'$ on line~\ref{ln:pred:pNode'}. 
By line~\ref{ln:pred:pNode'}, $\pNode'$ is the predecessor node in $\mathit{predNodes}$ that occurs the earliest in $Q$, and hence was announced the earliest. Since $C_{ann}$ is immediately before  $\pNode'$ was announced, $\pNode''$ is announced after $C_{ann}$.

From the code, this embedded predecessor operation completes before $\dNode$ is added to the $\RUALL$. Since $pOp$ added $\dNode$ to $\Druall$ while it traversed the $\RUALL$, $\dNode$ was added to the $\RUALL$ before $pOp$ began its traversal of the relaxed binary trie. Therefore, $dOp$'s first  embedded predecessor operation completes before $C_T$.
\end{proof}

\begin{observation}\label{obs:pNode'_always_pNode}
In all configurations after $C_{ann}$ but before $C_{\mathit{notify}}$, either $\pNode'$ or $\pNode$ is in the $\PALL$.
\end{observation}

\begin{proof}
By definition, $C_{ann}$ is immediately before $\pNode'$ is announced in the $\PALL$. By definition, $\pNode'$ is the predecessor node of an embedded predecessor operation performed by a \textsc{Delete} operation, $dOp$, whose DEL node, $\dNode$, is in $\Druall$.
Since $dOp$ adds $\pNode'$ into the $\PALL$ before adding $\dNode$ into the $\RUALL$, and does not remove $\pNode'$ from the $\PALL$ until $\dNode$ is removed from the $\RUALL$, $\pNode'$ is in the $\PALL$ in the configuration immediately after $pOp$ first encounters $\dNode$ in $\Druall$, and hence after $\pNode$ is announced. So $\pNode'$ is in the $\PALL$ from $C_{ann}$ to when $pOp$ is announced. Furthermore, $\pNode$ is in the $\PALL$ from when $pOp$ is announced until at least $C_{\mathit{notify}}$, because $pOp$ does not remove $\pNode$ from the $\PALL$ until after $C_{\mathit{notify}}$.
\end{proof}


From Lemma~\ref{obs:L_INS_inS} to \ref{lemma:orderL}, we prove properties of $L$ immediately after it is initialized on line~\ref{ln:pred:L-L'}. After this line of code, update nodes may only be removed from $L$ (on line~\ref{ln:pred:LremoveDEL}).

The key of an INS node, $\iNode$,  in $L$ is in $S$ sometime between $C_{ann}$ and $C_T$. 

\begin{lemma}\label{obs:L_INS_inS}
Consider any INS node, $\iNode$, in $L$, where $\iNode.\key = x$. There is a configuration $C$ between $C_{ann}$ and $C_T$ where
\begin{itemize}
	\item $x \in S$ in $C$, and
	\item $\iNode$ is the first activated update node in the $\latest[x]$ list in $C$.
\end{itemize} 
\end{lemma}

\begin{proof}
Let $iOp$ be the \textsc{Insert}$(x)$ operation with key $x$ that created $\iNode$.
Since $\iNode \in L$, $uOp$ successfully notifies either $\pNode$ or $\pNode'$. 
So $iOp$ encounters at least one of these predecessors nodes when it traverses the $\PALL$. 
The predecessor node among $\pNode$ and $\pNode'$ that $iOp$ encounters is announced in the $\PALL$. By the definition of $C_{ann}$ and Observation~\ref{obs:pNode'_before_pNode}, this predecessor node is announced in the $\PALL$ after $C_{ann}$.
Immediately before $uOp$ successfully adds a notify node containing a pointer to $\iNode$ to the notify list of $\pNode$ or $\pNode$ (on line~\ref{ln:sendNotification:CAS}), it performs \textsc{FirstActivated}$(\uNode)$ (on line~\ref{ln:sendNotification:firstActivated}), which returns \textsc{True}. By
Lemma~\ref{lemma:firstActivated}, $\iNode$ is the first activated update node in some configuration $C$ during this instance of \textsc{FirstActivated}$(\iNode)$. Note that $C$ occurs after $C_{ann}$.

Suppose $iOp$ is linearized before $\pNode$ is announced. Note that $x \in S$ in all configurations from when it is linearized to $C$.
Suppose $pOp$ assigns a value to $\pNode'$ on line~\ref{ln:pred:pNode'}. If $iOp$ is linearized before $C_{ann}$, then $x \in S$ in $C_{ann}$, which is before $C_T$. Otherwise $iOp$ is linearized after $C_{ann}$, and $x \in S$ immediately after $iOp$ is linearized, which is between $C_{ann}$ and $C_T$.

Suppose $iOp$ is linearized after $\pNode$ is announced.
Then $iOp$ encounters $\pNode$ (before encountering $\pNode'$) when traversing the $\PALL$ on line~\ref{ln:notify:for_loop}. Then $iOp$ must successfully add a notify node, $\nNode$, containing a pointer to $\iNode$ to $\pNode.\notifyList$ (otherwise $iOp$ will also not successfully notify $\pNode'$). It follows from line~\ref{ln:pred:removeL'1} that $\iNode \notin L_1$ and $\iNode \in L_2$. By line~\ref{ln:pred:thresholdL'2}, $\nNode.\notifyThreshold \geq x$. Since $\nNode.\notifyThreshold \neq -\infty$, $pOp$ has not completed its traversal of the $\RUALL$ (which is before $C_T$) when $iOp$ reads $\pNode.\RuallPosition$. This read is after $C$. So $C$ is before $C_T$.


\end{proof}	

We prove a similar lemma for a DEL node, $\dNode$, in $L$. Additionally, we prove that the second embedded predecessor operation of $\dNode$ is complete before $C_T$. 

\begin{lemma}\label{lemma:LcompletePred2}
Consider any DEL node, $\dNode$, in $L$, where $\dNode.\key = x$. There is a configuration $C$ between $C_{ann}$ and $C_T$ where
\begin{itemize}
	\item $x \notin S$ in $C$,
	\item $\dNode$ is the first activated update node in the $\latest[x]$ list in $C$, and
	\item the second embedded predecessor operation of $\dNode$ is complete before $C_T$.
\end{itemize} 
\end{lemma}

\begin{proof}
We consider the values of the lists $L_1$ and $L_2$ immediately before line~\ref{ln:pred:L-L'}, which is immediately before $L$ is initialized as the concatenation of $L_1$ and $L_2$. It follows from line~\ref{ln:pred:removeL'1} that update nodes added to $L_2$ are removed from $L_1$, so $\dNode$ is in exactly one of $L_1$ and $L_2$.
Since $\dNode \in L$, $dOp$ successfully notifies either $\pNode$ or $\pNode'$. 
So $dOp$ encounters at least one of these predecessors nodes when it traverses the $\PALL$. 
The predecessor node among $\pNode$ and $\pNode'$ that $dOp$ encounters is announced in the $\PALL$. By the definition of $C_{ann}$ and Observation~\ref{obs:pNode'_before_pNode}, this predecessor node is announced in the $\PALL$ after $C_{ann}$.

Suppose $\dNode \in L_2$. Prior to $dOp$ traversing the $\PALL$ to notify predecessor operations, $dOp$ completes its second embedded predecessor operation (on line~\ref{ln:delete:second_em_pred}). When $dOp$ encounters $\pNode$ in the $\PALL$, it must be after $\pNode$ is announced, which is after $C_{ann}$.
Immediately before $dOp$ creates a new notify node, $\nNode$, containing a pointer to $\dNode$, it performs \textsc{FirstActivated}$(\dNode)$ (on line~\ref{ln:notify:first_activated}), which returns \textsc{True}. By
Lemma~\ref{lemma:firstActivated}, there is a configuration $C$ during this instance of \textsc{FirstActivated}$(\dNode)$ in which $x \notin S$.
By line~\ref{ln:pred:thresholdL'2}, $\dNode$ is only added to $L_2$ when the notify threshold of $\nNode$ is greater than or equal to $x$. 
So
$dOp$ read that $\pNode.\RuallPosition$ points to an update node with key greater than or equal to $x$ prior to notifying $\pNode$, which is before $C_T$. 
So $C$ occurs before $C_T$
and $dOp$'s second embedded predecessor operation completes before $C_T$.

So suppose $\dNode \in L_1$, which means $dOp$ successfully notified $\pNode'$. Suppose, for contradiction, that $dOp$ encounters $\pNode$ prior to encountering $\pNode'$ when traversing the $\PALL$. Since $dOp$ successfully notified $\pNode'$ (on line~\ref{ln:sendNotification:CAS}), its instance of \textsc{FirstActivated}$(\dNode)$ returned \textsc{True} on line~\ref{ln:sendNotification:firstActivated}. By Lemma~\ref{lemma:firstActivated}, there is a configuration during this instance of \textsc{FirstActivated}$(\dNode)$ in which $\dNode$ is the first activated update node in the $\latest[x]$ list. This implies that $\dNode$ is the first activated update node in the $\latest[x]$ list throughout its attempt to notify $\pNode$. By Lemma~\ref{lemma:firstActivated_False}, prior to notifying $\pNode$, \textsc{FirstActivated}$(\dNode)$ returned \textsc{True} on line~\ref{ln:sendNotification:firstActivated}, and hence $dOp$ successfully notifies $\pNode$. By line~\ref{ln:pred:removeL'1}, $\dNode$ is removed from $L_1$, a contradiction.
So $dOp$ does not encounter $\pNode$ when traversing the $\PALL$. 
So $dOp$ began its traversal of the $\PALL$ before $\pNode$ is announced in the $\PALL$, and hence before $C_T$. So $dOp$ completes its second embedded predecessor operation before $C_T$.

After $dOp$ encounters $\pNode'$ but before $dOp$ creates a new notify node, $\nNode$, containing a pointer to $\dNode$, it performs \textsc{FirstActivated}$(\dNode)$ (on line~\ref{ln:notify:first_activated}), which returns \textsc{True}. By
Lemma~\ref{lemma:firstActivated}, there is a configuration $C'$ during this instance of \textsc{FirstActivated}$(\dNode)$ in which $\dNode$ is the first activated update node in the $\latest[x]$ list. Between when $dOp$ is linearized and $C'$, $x \notin S$. In particular, in the configuration $C$ immediately after $dOp$ begins its traversal of the $\PALL$ (which is between $C_{ann}$ and $C_T$), $x \notin S$. 
\end{proof}

The next lemma describes a case when the INS nodes of \textsc{Insert}$(x)$ operations are in $L$. 

\begin{lemma}\label{lemma:L}
Let $iOp$ be an $S$-modifying \textsc{Insert}$(x)$ operation, where $x < y$, that is linearized between $C_{ann}$ and $C_T$.
If $\Iuall \cup \Inotify$ does not contain an INS node with key $x$, 
then $L$ contains the INS node created by $iOp$ immediately after line~\ref{ln:pred:L-L'}.
\end{lemma}

\begin{proof}
Let $\iNode$ be the update node created by $iOp$. 
If $iOp$ does not notify $\pNode'$ or  $\pNode$ by the time $pOp$ completes its traversal of the $\UALL$, then it follows from Lemma~\ref{lemma:ins_notify} that $\iNode$ is in the $\UALL$ and is the first activated update node in the $\latest[x]$ list throughout $pOp$'s traversal of the $\UALL$.
So when $pOp$ encounters $\iNode$ during \textsc{TraverseUall}, Lemma~\ref{lemma:firstActivated_False} implies that $pOp$'s instance of \textsc{FirstActivated}$(\iNode)$ on line~\ref{ln:traverseUALL:findLatest} return \textsc{True}, so
$\iNode \in \Iuall$. 

So $iOp$ performs a complete instance of  \textsc{NotifyPredOps}$(\iNode)$ before $pOp$ completes a traversal of the $\UALL$. By Lemma~\ref{lemma:ins_notify}, $\iNode$ is the first activated update node in the $\latest[x]$ list throughout  \textsc{NotifyPredOps}$(\iNode)$. So Lemma~\ref{lemma:firstActivated_False}, all instances of \textsc{FirstActivated}$(\iNode)$ during \textsc{NotifyPredOps}$(\iNode)$ (on line~\ref{ln:notify:first_activated} or line~\ref{ln:sendNotification:firstActivated}) return \textsc{True}. 

By Observation~\ref{obs:pNode'_always_pNode},
either $\pNode'$ or  $\pNode$ is announced in the $\PALL$ from $C_{ann}$ to $C_{\mathit{notify}}$.
So $iOp$ must notify $\pNode'$ or  $\pNode$ before $pOp$ completes its traversal of the $\UALL$. Suppose $iOp$ encounters $\pNode$ when traversing the $\PALL$, and hence notifies $\pNode$.
Let $\nNode$ be the notify node that $iOp$ adds to $\pNode.\notifyList$.
Since $\iNode \notin \Inotify$, $\iNode$ is not added to $\Inotify$ on line~\ref{ln:pred:I2}.
So by line~\ref{ln:pred:notifyThresholdINS}, $\nNode.\notifyThreshold > x$.
It follows by line~\ref{ln:pred:prependL'2} that $\iNode$ is added to $L_2$.
Otherwise $iOp$ does not notify $\pNode$. So it encounters $\pNode'$ when traversing the $\PALL$ and notifies $\pNode'$. Then $\iNode$ is added to $L_1$ on line~\ref{ln:pred:prependL'1}. In either case, $L$ contains $\iNode$ immediately after line~\ref{ln:pred:L-L'}.
\end{proof}

The next lemma describes a case when the DEL nodes of \textsc{Delete}$(x)$ operations are in $L$. 

\begin{lemma}\label{lemma:L_DEL}
Let $dOp$ be an $S$-modifying \textsc{Delete}$(x)$ operation, where $x < y$, that is linearized between $C_{ann}$ and $C_T$, and let $\dNode$ be the DEL node it created.
If
\begin{itemize}
	\item no $S$-modifying \textsc{Insert}$(x)$ operation is linearized after $dOp$ but before $C_T$,
	\item $\dNode \notin \Druall$, and
	\item $\Iuall \cup \Inotify \cup (\Duall - \Druall) \cup (\Dnotify - \Druall)$ does not contain an update node with key $x$,
\end{itemize}
then $\dNode \in L$ immediately after line~\ref{ln:pred:L-L'}.
\end{lemma}

\begin{proof}
By Lemma~\ref{lemma:del_before_CT}, $dOp$ is linearized before $C_{< x}$, but after $C_{ann}$. Since no $S$-modifying \textsc{Insert}$(x)$ operation is linearized after $dOp$ but before $C_T$, $\dNode$ is the first activated update node in the $\latest[x]$ list from when it is linearized to $C_T$.
By Lemma~\ref{lemma:firstActivated_False}, any instance of \textsc{FirstActivated}$(\dNode)$ contained between when $dOp$ is linearized to $C_T$ returns \textsc{True}.
It follows that $pOp$ does not encounter $\dNode$ when it traverses the $\RUALL$ (otherwise 
\textsc{FirstActivated}$(\dNode)$ returns \textsc{True} on line~\ref{ln:traverseRevUALL:findLatest}
and $\dNode$ would be added to $\Druall$ on line~\ref{ln:TraverseRUall:addD}).
Since $\dNode \notin \Druall$, $\dNode$ is removed from the $\RUALL$ before $pOp$ encounters $dNode$ in the $\RUALL$. So $dOp$'s instance of \textsc{NotifyPredOps}$(\dNode)$  (on line~\ref{ln:delete:notifyPredOps}) is completed before $pOp$ first encounters an update node with key less than $x$ in the $\RUALL$. 
All instances of \textsc{FirstActivated}$(\dNode)$ during \textsc{NotifyPredOps}$(\dNode)$ returns \textsc{True}.
By Observation~\ref{obs:pNode'_always_pNode}, either $\pNode'$ or $\pNode$ is in the $\PALL$ from $C_{ann}$ to $C_{\mathit{notify}}$. 
So $dOp$ must notify $\pNode'$ or  $\pNode$ during its instance of \textsc{NotifyPredOps}$(\dNode)$.
Suppose $dOp$ encounters $\pNode$ when traversing the $\PALL$ on line~\ref{ln:notify:for_loop}, and hence notifies $\pNode$. Consider the notify node, $\nNode$, that $dOp$ adds to $\pNode.\notifyList$.
Since $\pNode.\RuallPosition$ points to an update node with key greater than or equal to $x$ throughout \textsc{NotifyPredOps}$(\dNode)$, $\nNode.\notifyThreshold \geq x$. It follows that $\dNode$ is added to $L_2$ on line~\ref{ln:pred:prependL'2}.
Otherwise $dOp$ does not notify $\pNode$. So it encounters $\pNode'$ when traversing the $\PALL$ and notifies $\pNode'$. Then $\dNode$ is added to $L_1$ on line~\ref{ln:pred:prependL'1}. In either case,  $\dNode \in L$ immediately after line~\ref{ln:pred:L-L'}.

\end{proof}

In the next lemma, we consider the order in which update operations with the same key notify predecessor operations. 
We guarantee that once an update operation adds a notify node (containing a pointer to its own update node) to the notify list of a predecessor node, no update operation with the same key that was linearized earlier can add a notify node (containing a pointer to its own update node) to this notify list.
Note that a notify list may contain several notify nodes containing a pointer to the same INS node due to helping \textsc{Delete} operations.

\begin{lemma}\label{lemma:notifyList}
Let $\uNode_1$ and $\uNode_2$ be two update nodes with the same key $x$, and suppose the update operation that created $\uNode_1$ is linearized before the update operation that created $\uNode_2$.
Then any notify node that contains a pointer to $\uNode_2$ is added to $\pNode.\notifyList$ after any notify node that contains a pointer to $\uNode_1$.
\end{lemma}

\begin{proof}
Let $uOp_1$ and $uOp_2$ be the update operations that created $\uNode_1$ and $\uNode_2$, respectively. 
If no notify node in $\pNode.\notifyList$ contains a pointer to $\uNode_1$, the lemma is vacuously true. 
So at least one notify node in $\pNode.\notifyList$ contains a pointer to $\uNode_1$.
Consider the last notify node, $\nNode_1$, added to $\pNode.\notifyList$ that contains a pointer to $\uNode_1$.
Then $uOp_1$ (or a helper of $uOp_1$) performs a successful CAS on line~\ref{ln:sendNotification:CAS}, which updates $\pNode.\notifyList.\head$ to point to $\nNode_1$.
In the line of code immediately before this successful CAS (on line~\ref{ln:sendNotification:firstActivated}), \textsc{FirstActivated}$(\uNode_1)$ returned \textsc{True}. By Lemma~\ref{lemma:firstActivated}, there is a configuration $C$ during this instance of \textsc{FirstActivated}$(\uNode_1)$ in which $\uNode_1$ is the first activated update node in the $\latest[x]$ list. Therefore, $uOp_2$ must be linearized after $C$. 
Since $uOp_1$'s CAS is successful, $\pNode.\notifyList.\head$ is not changed by any step after $C$ until $\nNode_1$ is successfully added to $\pNode.\notifyList$.
Since $uOp_2$ does not attempt to add any notify nodes containing $\uNode_2$ into $\pNode.\notifyList$ until after it is linearized, this occurs after $\nNode_1$ is added to $\pNode.\notifyList$.
\end{proof}

We use the previous lemma to argue that update nodes with the same key in $L$ appear in the order in which the update operations that created them were linearized.

\begin{lemma}\label{lemma:orderL}


Let $\uNode_1$ and $\uNode_2$ be two update nodes in $L$ with the same key $x$.
If the update operation that created $\uNode_1$ is linearized before the update operation that created $\uNode_2$, then $\uNode_1$ appears before $\uNode_2$ in $L$.
\end{lemma}

\begin{proof}
Let $uOp_1$ and $uOp_2$ be the update operations that created $\uNode_1$ and $\uNode_2$, respectively. 
We consider the values of the lists $L_1$ and $L_2$ immediately before line~\ref{ln:pred:L-L'}, which is immediately before $L$ is initialized as the concatenation of $L_1$ and $L_2$. It follows from line~\ref{ln:pred:removeL'1} that update nodes added to $L_2$ are removed from $L_1$, so the update nodes contained in $L_1$ and $L_2$ are disjoint.

Suppose $\uNode_1$ is in $L_2$ and, hence, from the code, $uOp_1$ added a notify node, $\nNode_1$, containing a pointer to $\uNode_1$ to $\pNode.\notifyList$. Then $uOp_1$ read that $\pNode \in \PALL$ (on line~\ref{ln:notify:for_loop}) when traversing the $\PALL$. Upon encountering $\pNode$, $uOp_1$ performs a successful CAS on line~\ref{ln:sendNotification:CAS}, which updates $\pNode.\notifyList.\head$ to point to $\nNode_1$. In the line of code immediately before this successful CAS (on line~\ref{ln:sendNotification:firstActivated}), \textsc{FirstActivated}$(\uNode_1)$ returned \textsc{True}. By Lemma~\ref{lemma:firstActivated}, there is a configuration $C$ during this instance of \textsc{FirstActivated}$(\uNode_1)$ in which $\uNode_1$ is the first activated update node in the $\latest[x]$ list. Therefore, $uOp_2$ must be linearized after $C$.
So $uOp_2$ traverses the $\PALL$ after $C$.
When $uOp_2$ traverses the $\PALL$, $uOp_2$ must encounter $\pNode$ before encountering $\pNode'$ because $\pNode$ is announced later than $\pNode'$. 

Suppose $uOp_2$ does not add a notify node into $\pNode.\notifyList$, which happens when \textsc{FirstActivated}$(\uNode_2)$ returns \textsc{False} on  line~\ref{ln:sendNotification:firstActivated}. It follows from Lemma~\ref{lemma:firstActivated_False} that there is a configuration during this instance of \textsc{FirstActivated}$(\uNode_2)$ in which $\uNode_2$ is not the first activated update node in the $\latest[x]$ list. After this configuration,  $uOp_2$ does not add any notify nodes to the notify lists of any predecessor nodes, and hence will not notify $\pNode'$. Since $uOp$ does not add a notify node to $\pNode.\notifyList$ or $\pNode'.\notifyList$, this contradicts the fact that $\uNode_2 \in L$. 

So $uOp_2$ does  add a notify node into $\pNode.\notifyList$. It follows from  line~\ref{ln:pred:removeL'1} that $\uNode_2 \notin L_1$. Since $\uNode_2$ is in $L$, it must be that $\uNode_2 \in L_2$. Then by Lemma~\ref{lemma:notifyList}, any notify node containing a pointer to $\uNode_1$ is added to $\pNode'.\notifyList$ before  any notify node containing a pointer to $\uNode_2$. Since the update nodes in $L_1$ are in the order in which they were added to $\pNode'.\notifyList$, $\uNode_1$ appears before $\uNode_2$ in $L$.

Suppose $\uNode_1$ is in $L_1$, and hence, from the code, was added to $\pNode'.\notifyList$. 
Suppose $\uNode_2$ is in $L_1$. Then by Lemma~\ref{lemma:notifyList}, any notify node containing a pointer to $\uNode_1$ is added to $\pNode'.\notifyList$ before  any notify node containing a pointer to $\uNode_2$. Since the update nodes in $L_1$ are in the order in which they were added to $\pNode'.\notifyList$, $\uNode_1$ appears before $\uNode_2$ in $L$.
So suppose $\uNode_2$ is in $L_2$. Since $L$ is the concatenation of $L_1$ followed by $L_2$, it follows that $\uNode_1$ appears before $\uNode_2$ in $L$.
\end{proof}


The next lemma is about the set of keys, $X$, determined on lines~\ref{ln:pred:X_init} and line~\ref{ln:pred:X_init_INS}. After line~\ref{ln:pred:X_init_INS}, $X$ does not change.

\begin{lemma}\label{lemma:X_in_S}
Consider the set of keys in $X$ immediately after line~\ref{ln:pred:X_init_INS}. 
For any $x \in X$, there is a configuration between $C_{ann}$ and $C_T$ in which
$x \in S$.
\end{lemma}

\begin{proof}
Suppose $x \in X$ immediately after line~\ref{ln:pred:X_init}. From this line, there is a DEL node, $\dNode$, in $\Druall$ where $\dNode.\delpred = x$. 
Recall that $\dNode.\delpred$ is the return value of the  first embedded predecessor operation of $\dNode$. It follows from Lemma~\ref{obs:first_pred_Cann} that this first embedded predecessor operation occurs entirely between $C_{ann}$ and $C_T$. Since the predecessor node of this operation is only in the $\PALL$  this operation,
it follows from Theorem~\ref{thm:predecessor} that $x\in S$ sometime  between $C_{ann}$ and $C_T$.

Now suppose $x$ is added to $X$ on line~\ref{ln:pred:X_init_INS}. From this line, there is an INS node, $\iNode$, in $L$ where $\iNode.\key = x$. By Lemma~\ref{obs:L_INS_inS}, $x \in S$ in some configuration between $C_{ann}$ and $C_T$.
\end{proof}

In the following three lemmas, we prove basic facts about the vertices of the directed graph $T_L$ that is constructed on line~\ref{ln:pred:T_L}. 

\begin{lemma}\label{lemma:w_leaf}
Suppose $w \in S$ in all configurations from $C_{ann}$ to $C_T$. If $w$ is a vertex of $T_L$, then $w$ is a sink.
\end{lemma}

\begin{proof}
Suppose that $w$ has an outgoing edge $(w,w')$. By Definition~\ref{definition:T_L}, there is a $\dNode \in L$ where $\dNode.\key = w$ and $\dNode.\delpredsecond = w'$. By Lemma~\ref{lemma:LcompletePred2}, 
$w \notin S$ in some configuration between $C_{ann}$ to $C_T$.	
\end{proof}

\begin{lemma}\label{lemma:T_L_ins}
Consider an INS node, $\iNode$, in $L$ that was created by an \textsc{Insert}$(w)$ operation, $iOp$. 
Suppose no \textsc{Delete}$(w)$ operation is linearized after  $iOp$ but before $C_T$. Then $w$ is a sink in $T_L$.
\end{lemma}

\begin{proof}
Note that by Definition~\ref{definition:T_L}, $w$ is a vertex of $T_L$.
Suppose, for contradiction, that $w$ has an outgoing edge.
Then there is a DEL node, $\dNode$, in $L$ with key $w$.  
By line~\ref{ln:pred:LremoveDEL}, $\dNode$ is the last update 
node in $L$ with key $w$, so it
appears after $\iNode$ in $L$.  By Lemma~\ref{lemma:orderL}, the \textsc{Delete}$(w)$ operation, $dOp$, that created $\dNode$ is linearized after $iOp$. By Lemma~\ref{lemma:LcompletePred2}, there is a configuration $C$ between $C_{ann}$ and $C_T$ in which $\dNode$ is the first activated update node in the $\latest[x]$ list. This implies that $dOp$ is linearized between $iOp$ and $C_T$, a contradiction.
\end{proof}

\begin{lemma}\label{lemma:T_L_DEL}
Consider an DEL node, $\dNode$, in $L$ that was created by a \textsc{Delete}$(w)$ operation, $dOp$. 
Suppose no \textsc{Insert}$(w)$ operation is linearized after $dOp$ but before $C_T$. Then
$w$ is not a sink in $T_L$.
\end{lemma}

\begin{proof}
We first prove that there is no INS node with key $w$ that appears after $\dNode$ in $L$. Suppose there is such an INS node, $\iNode$. Let $iOp$ be the \textsc{Insert}$(w)$ operation that created $\iNode$. 
By Lemma~\ref{lemma:orderL}, $iOp$ is linearized after $dOp$. 
By Lemma~\ref{obs:L_INS_inS}, there is a configuration before $C_T$ in which  $\iNode$ is the first activated update node in the $\latest[w]$ list. Since $\iNode$ is activated when it is linearized, there is a configuration 
between when $iOp$ is linearized and $C_T$ in which $w \in S$.
This contradicts the fact that $w \notin S$ in all configurations from when $dOp$ is linearized and $C_T$.  

Let $\dNode.\delpredsecond = w'$. It follows by Definition~\ref{definition:T_L} of $T_L$ that there is an edge $(w,w')$ in $T_L$. So $w$ is not a sink. 
\end{proof}




The following lemma gives a scenario in which the endpoint of an edge in $T_L$ is a key that is in $S$ in some configuration between $C_{ann}$ and $C_T$.

\begin{lemma}\label{lemma:path_in_TL}
Let $(x,x')$ be an edge in $T_L$. 
Suppose that $x \in S$ in some configuration $C$ between $C_{ann}$ and $C_T$ and $\Iuall \cup \Inotify \cup (\Duall - \Druall) \cup (\Dnotify - \Druall)$  does not contain update nodes with key $x$.
\begin{itemize}
	\item If $x' \neq -1$, then $x' \in S$ in some configuration between $C$ and $C_T$.
	\item If $w < x$ and $w \in S$ in all configurations between $C_{ann}$ and $C_T$, then $w \leq x'$.
\end{itemize}
\end{lemma}

\begin{proof}
Since $(x,x')$ is an edge in $T_L$, there is a DEL node, $\dNode$, in $L$ immediately after line~\ref{ln:pred:LremoveDEL} such that $\dNode.\key = x$ and $\dNode.\delpredsecond = x'$.

Let $C$ be the last configuration between $C_{ann}$ and $C_T$ in which $x \in S$.
Suppose, for contradiction, that $dOp$ is linearized before $C$. Since $x \in S$ in $C$,
there is an $S$-modifying \textsc{Insert}$(x)$ operation, $iOp$, linearized after $dOp$ but before $C$. Let $\iNode$ be the INS node created by $iOp$. 

\begin{itemize}
	\item Suppose $iOp$ is linearized before $C_{ann}$. 
	Since $\dNode \in L$, it follows from the code (from lines 217, 218, 220, and 223) that a notify node containing a pointer to $\dNode$ is in the notify list of $\pNode$ or $\pNode'$. 
	A \textsc{Delete} operation only adds a notify node to the notify list of a predecessor node it encounters when traversing the $\PALL$ in an instance of  \textsc{NotifyPredOps}$(\dNode)$ on line~185.
	By the definition of $C_{ann}$, $\pNode'$ and $\pNode$ are only in the $\PALL$ after $C_{ann}$. So $dOp$ encounters at least one of $\pNode'$ or $\pNode$ after $C_{ann}$. 
	Since $iOp$ is linearized after $dOp$ but before $C_{ann}$, $\dNode$ is no longer the first activated update node in the $\latest[x]$ list after $C_{ann}$. It follows from Lemma~\ref{lemma:firstActivated} that any instance of \textsc{FirstActivated}$(\dNode)$ that is invoked after $C_{ann}$ returns \textsc{False}. 
	So $dOp$'s instance of \textsc{FirstActivated}$(\dNode)$ on line~\ref{ln:notify:first_activated} returns \textsc{False} and $dOp$ does not notify either $\pNode'$ or $\pNode$, a contradiction.
	
	\item Otherwise, $iOp$ is linearized after $C_{ann}$. 
	By assumption, $\Iuall \cup \Inotify$ only contains update nodes with key smaller than $x$. In particular, it does not contain an INS node with key $x$. 
	Since $iOp$ is linearized between $C_{ann}$ and $C_T$, Lemma~\ref{lemma:L} implies that $\iNode$ is in $L$.
	Since $iOp$ is linearized after $dOp$, it follows from Lemma~\ref{lemma:orderL} that $\iNode$ appears after $\dNode$ in $L$.
	By line~\ref{ln:pred:LremoveDEL} of the code, $\dNode$ is removed from $L$. This contradicts the fact that $\dNode \in L$ immediately after line~\ref{ln:pred:LremoveDEL}.
\end{itemize}

Since update operations are linearized at steps,
$dOp$ is linearized after $C$. Hence, its second embedded predecessor operation occurs after $C$. 
By Lemma~\ref{lemma:LcompletePred2}, its second embedded predecessor operation (which returns $x'$) is completed before $C_T$, and, hence, before $pOp$.
By the induction hypothesis, this embedded predecessor operation satisfies Theorem~\ref{thm:predecessor}.
Therefore, $x'$ is the predecessor of $x$ sometime during this embedded predecessor operation.
Hence, if $x' \neq -1$, then $x' \in S$ in some configuration between $C$ and $C_T$.
Furthermore, if $w < x$ and $w \in S$ in all configurations between $C_{ann}$ and $C_T$, then the predecessor of $x$ in all configurations between $C_{ann}$ and $C_T$ is at least $w$, so 
$w \leq x'$.
\end{proof}


Recall that on line~\ref{ln:pred:R_from_TL}, the set of keys, $R$, is the set of sinks in $T_L$ that are reachable from some key in $X$.
The following lemma follows from inductively applying Lemma~\ref{lemma:path_in_TL} to each edge on certain paths from keys in $X$ to a sink.

\begin{lemma} \label{lemma:R_at_least_w}
Suppose $w \in S$ in all configurations between $C_{ann}$ and $C_T$, and $\Iuall \cup \Inotify \cup (\Duall - \Druall) \cup (\Dnotify - \Druall)$ only contains update nodes with key smaller than $w$.
If $X$ contains a key $x \geq w$ immediately after line~\ref{ln:pred:X_init_INS}, then $R$ contains a key $x'$, where $w \leq x' \leq x$, immediately after line~\ref{ln:pred:R_from_TL}.
\end{lemma}

\begin{proof}
Recall that each vertex in $T_L$ has at most 1 outgoing edge.
If $x$ is a sink in $T_L$, then by the construction of $R$ (on line~\ref{ln:pred:R_from_TL}), $x \in R$.
Otherwise, consider the path  $(x = x_0, x_1, \dots, x_\ell)$ in $T_L$, where $x_\ell$ is a sink in $T_L$. By the construction of $R$ (on line~\ref{ln:pred:R_from_TL}), $x_\ell \in R$. By Definition~\ref{definition:T_L}, the keys on this path are strictly decreasing, so $x_\ell < x_0$. By Lemma~\ref{lemma:w_leaf}, if $w$ is a vertex in $T_L$, it is a sink, so $w \notin \{x_0, x_1, \dots, x_{\ell-1}\}$.
Because \textsc{Delete} operations do not have key $-1$, if $-1$ is a vertex in $T_L$, it is a sink, and, hence, $-1 \notin \{x_0, x_1, \dots, x_{\ell-1}\}$.

We prove by induction on the path  $x_0, x_1, \dots, x_{\ell-1}$ that $x_{\ell-1} \in S$ in some configuration between $C_{ann}$ and $C_T$ and $x_{\ell-1}  > w$. 
In the base case, it follows from that Lemma~\ref{lemma:X_in_S}, $x_0 \in S$ in some configuration between $C_{ann}$ and $C_T$. Since $x_0  \neq w$ and $x_0 \geq w$, it follows that $x_0 > w$.

Now suppose, for $0 \leq i \leq \ell -2$, that $x_i \in S$ in some configuration between $C_{ann}$ and $C_T$ and $x_i > w$. We show the claim is true for $x_{i+1}$. 
By Lemma~\ref{lemma:path_in_TL} and the fact that $x_{i+1} \neq -1$, $x_{i+1} \in S$ in some configuration between $C_{ann}$ and $C_T$. Furthermore, since $w < x_i$, the lemma also states that $w \leq x_{i+1}$. Since $x_{i+1} \neq w$, it follows that $w < x_{i+1}$.

Now consider the edge $(x_{\ell-1}, x_\ell)$. By Lemma~\ref{lemma:path_in_TL}, $w \leq  x_\ell$. Since  $x_\ell \in R$, $R$ contains a key $x_\ell$, where $w \leq x_\ell < x$.
\end{proof}

After $R$ is initialized on line~\ref{ln:pred:R_from_TL}, some keys may be removed from $R$ on line~\ref{ln:pred:R_remove_Druall}. But there are some keys that are never removed from $R$.

\begin{lemma}\label{lemma:k_never_remove}
Suppose $w \in S$ in all configurations from $C_{ann}$ to $C_T$. Then $w$ is not removed from $R$ (on line~\ref{ln:pred:R_remove_Druall}).
\end{lemma}

\begin{proof}

Suppose that $w$ is removed from $R$ on line~\ref{ln:pred:R_remove_Druall}. Then there is a DEL node, $\dNode$, in $\Druall$ where $\dNode.\key = w$. Immediately before $pOp$ added $\dNode$ to $\Druall$ on line~\ref{ln:TraverseRUall:addD}, it checked that $\dNode$ is activated and performed \textsc{FirstActivated}$(\dNode)$, which returned \textsc{True}. It follows from Lemma~\ref{lemma:firstActivated} that $w \notin S$ in some configuration $C$ during this instance of \textsc{FirstActivated}$(\dNode)$. So $C$ occurs sometime during $pOp$'s traversal of the $\RUALL$, which by definition is between $C_{ann}$ and $C_T$.
\end{proof}

The following lemma is the main lemma that proves Property~\ref{prop:trie} is satisfied by $pOp$.
The conclusion of the lemma is that $pOp$ has a candidate return value which is at least $w$.

\begin{lemma}\label{lemma:pred_0}

Suppose an $S$-modifying \textsc{Insert}$(w)$ operation $iOp$ is linearized before $C_T$, $w < y$, 
and there are no $S$-modifying \textsc{Delete}$(w)$ operations linearized after $iOp$ and before $C_T$.
Then either
\begin{itemize}
	\item $pOp$'s instance of \textsc{RelaxedPredecessor}$(y)$ returns a value at least $w$, 
	\item $\Iuall \cup \Inotify \cup (\Duall - \Druall) \cup (\Dnotify - \Druall)$ contains an update node with key at least $w$, or
	\item $r_0$ is set to a value at least $w$ on line~\ref{ln:pred:set_pred0}. 
	
\end{itemize}
\end{lemma}

\begin{proof}
Let $\iNode$ be the INS node created by $iOp$.
If $\Iuall \cup \Inotify \cup (\Duall - \Druall) \cup (\Dnotify - \Druall)$ contains an update node with a key which is at least $w$, the lemma holds. 
So suppose it only contains update nodes with keys smaller than $w$.


If $iOp$ did not complete updating the relaxed binary trie before 
$C_T$, then,
since it is the latest $S$-modifying update operation with key $w$ linearized prior to $C_T$,
then the preconditions of Lemma~\ref{lemma:dangerOverlap} hold at $C_T$.
Similarly, if some $S$-modifying update operation with key $w$ is linearized during $pOp$'s traversal of the relaxed binary trie, then the preconditions of Lemma~\ref{lemma:dangerOverlap} hold in the configuration immediately after this operation is linearized.
In either case,  Lemma~\ref{lemma:dangerOverlap} implies that $\Iuall \cup \Inotify \cup \Duall \cup \Dnotify$ contains an update node, $\uNode$, with key $w$ that was created by $iOp$ or by an $S$-modifying update operation with key $w$ linearized after $iOp$.
Since $\Iuall \cup \Inotify \cup (\Duall - \Druall) \cup (\Dnotify - \Druall)$ does not contain update nodes with key $w$,
$\Druall$ must contain $\uNode$. 
Since $\Druall$ only contains DEL nodes, the creator of $\uNode$ is not $iOp$.
There are no $S$-modifying update operations with key $w$ linearized after $iOp$ but before $C_T$, so 
the creator of $\uNode$ is linearized after $C_T$.
Since the creator of $\uNode$ is linearized when $\uNode$ is activated, 
and
by lines~\ref{ln:pred:uall1} and \ref{ln:pred:travtrie}, $pOp$ traverses the $\RUALL$ before $C_T$,
$\uNode$ is inactive throughout $pOp$'s traversal of the $\RUALL$.
By Lemma~\ref{lemma:firstActivated}, any instance of \textsc{FirstActtivated}$(\uNode)$ (on line~\ref{ln:traverseRevUALL:findLatest}) during $pOp$'s traversal of the $\RUALL$ returns \textsc{False}.
By lines~\ref{ln:traverseRevUALL:findLatest} to \ref{ln:TraverseRUall:addD}, 
$\uNode$ is never added to $\Druall$. This is a contradiction. 
Hence, $iOp$ completed updating the relaxed binary trie before $C_T$ and no $S$-modifying update operation with key $w$ is linearized during $pOp$'s traversal of the relaxed binary trie. 

Thus, $w$ is completely present throughout $pOp$'s traversal of the relaxed binary trie. 
Recall that $k$ is the largest value less than $y$ that is  completely present throughout $pOp$'s traversal of the relaxed binary trie, so $w \leq k$.
By the specification of the relaxed binary trie (Lemma~\ref{lemma:relaxed_trie_kxy}), if $pOp$'s instance of \textsc{RelaxedPredecessor}$(y)$ returns a key in $U$, it returns a value at least $k$.

So, suppose $pOp$'s instance of \textsc{RelaxedPredecessor}$(y)$ returns $\bot$. 
Since $\Iuall \cup \Inotify \cup (\Duall - \Druall) \cup (\Dnotify - \Druall)$ only contains update nodes with key smaller than $w$ (and, hence, smaller than $k$), it follows from Lemma~\ref{lemma:d0_non_empty} that $\Druall$ contains a DEL node, $\dNode$, with key $x$, where $k < x < y$.
So the test on line~\ref{ln:pred:fix_binary_trie_pred} evaluates to \textsc{True}. It remains to show that $r_0$ will be set to a value at least $w$ by $pOp$ on line~\ref{ln:pred:set_pred0}.

\begin{itemize}
	
	\item Suppose $iOp$ is linearized after $C_{ann}$. By Lemma~\ref{lemma:L}, $\iNode \in L$ immediately after line~224. 
	Note that
	$pOp$ does not remove $\iNode$ from $L$ on line~\ref{ln:pred:LremoveDEL}.
	So, by line~\ref{ln:pred:X_init_INS}, $w \in X$.
	By Lemma~\ref{lemma:T_L_ins}, $w$ is a sink in $T_L$. Therefore, from the code, $w \in R$ on line~\ref{ln:pred:R_from_TL}.
	
	
	It remains to show that $w$ is not removed from $R$ on line~\ref{ln:pred:R_remove_Druall}. 
	Suppose, for contradiction, there is a DEL node, $\dNode'$, in  $\Druall$ with key $w$.
	Then $pOp$ reads a pointer to $\dNode'$ (on lines~242 to 243) when traversing the $\RUALL$.
	Recall that $C_{\leq w}$ is the configuration immediately after $pOp$ first reads a pointer to an update node with key less than or equal to $w$ during its traversal of the $\RUALL$. So 
	$C_{\leq w}$ occurs before or immediately after $pOp$ first reads a pointer to $\dNode'$.
	Therefore, after $C_{\leq w}$, $pOp$ reads that $\dNode'$ is active and performs \textsc{FirstActivated}$(\dNode')$ on line~\ref{ln:traverseRevUALL:findLatest}, which returns \textsc{True}. 
	It follows from Lemma~\ref{lemma:firstActivated} that there is a configuration $C$ during this instance of \textsc{FirstActivated}$(\dNode')$ in which $w \notin S$.
	Since $w \in S$ in all configurations between when $iOp$ is linearized and $C_T$, $iOp$ is linearized after $C$, and, hence, after $C_{\leq w}$.
	Since $C_T$ is before $pOp$'s traversal of the $\UALL$, 
	it follows from Lemma~\ref{lemma:ins_after_CT} that $\iNode \in \Iuall \cup \Inotify$, a contradiction.
	Thus, $w$ is not removed from $R$ on line~\ref{ln:pred:R_remove_Druall} and $r_0$ is set to a value at least $w$ on line~\ref{ln:pred:set_pred0}.
	
	\item Now suppose $iOp$ is linearized before $C_{ann}$. So $w \in S$ in all configurations between $C_{ann}$ and $C_T$.
	Since $\dNode \in \Druall$, Lemma~\ref{obs:first_pred_Cann} implies that the first embedded predecessor operation of $\dNode$ is announced after $C_{ann}$ and is completed before $C_T$. 
	Since  $w \in S$ throughout the execution interval of $\dNode$'s first embedded predecessor operation which has key $x$, it follows from Theorem~\ref{thm:predecessor} that this embedded predecessor operation returns a value $w'$ such that $w \leq w' < x$.
	This value $w'$ is added to $X$ on line~\ref{ln:pred:X_init}. 
	By Lemma~\ref{lemma:R_at_least_w}, $R$ contains a key which is at least $w$ immediately after line~\ref{ln:pred:R_from_TL}.
	
	Let $w''$ be the smallest value greater than or equal to $w$ that is in $R$ immediately after line~\ref{ln:pred:R_from_TL}.
	If $w = w''$, then by Lemma~\ref{lemma:k_never_remove}, $w''$ is never removed from $R$. In this case, $r_0$ is set to a value at least $w$ on line~\ref{ln:pred:set_pred0}.
	So, suppose $w'' > w$. 
	If $w''$ is removed from $R$ on line~\ref{ln:pred:R_remove_Druall},
	then there exists a DEL node, $\dNode'$, in $\Druall$ such that $\dNode'.\key = w''$. 
	By Lemma~\ref{obs:first_pred_Cann}, the first embedded predecessor operation of $\dNode'$ occurs entirely between $C_{ann}$ and $C_T$.
	Since $w \in S$ throughout this embedded predecessor operation which has key $w''$, it follows from Theorem~\ref{thm:predecessor} that it returns a key $w'''$, where $w \leq w''' < w''$. So $X$ contains a key $w'''$ immediately after line~\ref{ln:pred:X_init}.
	Lemma~\ref{lemma:R_at_least_w} implies that, immediately after line~\ref{ln:pred:R_from_TL},  $R$ contains a key $w^{*}$, where $w \leq w^{*} \leq w''' < w''$. This contradicts the definition of $w''$.
	Therefore, $r_0$ is set to a value at least $w$ on line~\ref{ln:pred:set_pred0}.

	
\end{itemize}
\end{proof}

Recall that the value $w$ assigned to $r_0$ on line~\ref{ln:pred:set_pred0} is a candidate return value when $\Iuall \cup \Inotify \cup (\Duall - \Druall) \cup (\Dnotify - \Druall)$ only contains update nodes with key less than $w$.
If $w \neq -1$, we prove that $w \in S$ in $C_T$. 
Property~\ref{prop:configC}(b) and Property~\ref{prop:configC}(c) are vacuously true when $C = C_T$.
Thus, $w$ satisfies Property~\ref{prop:configC} with $C = C_T$. 


\begin{lemma}\label{lemma:propC_r0}
Suppose $\Iuall \cup \Inotify \cup (\Duall - \Druall) \cup (\Dnotify - \Druall)$ 
only contains update nodes with key less than $w$.
Let $w$ be the value of $r_0$ immediately after line~\ref{ln:pred:set_pred0}. 
If $w \neq -1$, then $w \in S$ in $C_T$.
\end{lemma}

\begin{proof}
We first prove that there is a configuration $C$ between $C_{ann}$ and $C_T$ in which $w \in S$.
If $w \in X$ (immediately after line~\ref{ln:pred:X_init_INS}), then by Lemma~\ref{lemma:X_in_S} $w \in S$ in some configuration between $C_{ann}$ and $C_T$. So suppose $w \notin X$.
Since $r_0$ is set to the largest key in $R$,
it follows from line~\ref{ln:pred:R_from_TL} such that there is a path $(x_0,x_1,\dots,x_\ell)$, where $x_0 \in X$ and $x_\ell = w$ is a sink.
Recall that, for each edge $(x_i,x_{i+1})$ in $T_L$, $x_{i+1} < x_i$. So for $0 \leq i < \ell$, $x_i > w$. By Lemma~\ref{lemma:X_in_S}, $x_0 \in S$ in some configuration between $C_{ann}$ and $C_T$. Now suppose, for $0 \leq i < \ell$, that $x_i \in S$ in some configuration between  $C_{ann}$ and $C_T$. By Lemma~\ref{lemma:path_in_TL}, $x_{i+1} \in S$ in some configuration between  $C_{ann}$ and $C_T$. It follows by induction that $w \in S$ in some configuration $C$ between $C_{ann}$ and $C_T$.

Suppose, for contradiction, that $w \notin S$ in $C_T$. 
Since $w \in S$ in $C$, there exists an $S$-modifying \textsc{Delete}$(w)$ operation that removed $w$ from $S$ between $C$ and $C_T$.
Let $dOp$ be the last such $S$-modifying \textsc{Delete}$(w)$ operation, and let $\dNode$ be the DEL node created by $dOp$.
Note that 
$dOp$ is linearized after $C_{ann}$ and no $S$-modifying \textsc{Insert}$(w)$ operation is linearized after $dOp$ but before $C_T$. 

If $\dNode \in \Druall$, then $w$ is removed from $R$ on line~\ref{ln:pred:R_remove_Druall} and $r_0$ is not set to $w$ on line~\ref{ln:pred:set_pred0}. So $\dNode \notin \Druall$.
By assumption, there is no update node with key $w$ in $\Iuall \cup \Inotify \cup (\Duall - \Druall) \cup (\Dnotify - \Druall)$, 
so Lemma~\ref{lemma:L_DEL} states that $\dNode \in L$. 	
By Lemma~\ref{lemma:T_L_DEL}, $w$ is not a sink in $T_L$.
By line~\ref{ln:pred:R_from_TL}, $R$ only contains sinks of $T_L$, so $w \notin R$.
By line~\ref{ln:pred:set_pred0}, $r_0$ is set to the largest value in $R$, which cannot be $w$, a contradiction.

\end{proof}

Therefore, $pOp$'s candidate return values satisfy all properties. It follows by 	Theorem~\ref{thm:predecessor} that when  $pOp$ returns $w \in U \cup \{-1\}$, there exists a configuration during $pOp$ in which $w$ is the predecessor of $y$. So our implementation of a lock-free binary trie is linearizable with respect to all operations.





\subsection{Amortized Analysis}\label{section_analysis}

In this section, we give the amortized analysis of our lock-free binary trie implementation. The amortized analysis is simple because the algorithms to update and traverse the relaxed binary trie are wait-free. Most other steps involve traversing and updating lock-free linked lists, which has already been analyzed.

\begin{lemma}
\textsc{Search}$(x)$ operations have $O(1)$ worst-case step complexity.
\end{lemma}

\begin{proof}
A \textsc{Search}$(x)$ operation finds the first activated update node in the $\latest[x]$ list. From the code, it always completes in a constant number of reads.
\end{proof}

Recall that inserting and deleting from a lock-free linked list can be done with $O(\dot{c}(op) + L(op))$ amortized cost~\cite{FomitchevR04}, where $L(op)$ is the number of nodes in the linked list at the start of $op$.

\begin{lemma}\label{lemma:amortizedList}
In any configuration $C$, the number of nodes in the $\UALL$, $\PALL$, and $\RUALL$ is $O(\dot{c}(C))$, where $\dot{c}(C)$ denotes the number of \textsc{Insert}, \textsc{Delete} or \textsc{Predecessor} operations active during $C$. For any operation $op$, the amortized cost to update or traverse these linked lists is $O(\dot{c}(op))$.
\end{lemma}

\begin{proof}
Each \textsc{Insert}, \textsc{Delete}, or \textsc{Predecessor} operation only adds a constant number of nodes they created into the $\UALL$, $\RUALL$, and $\PALL$.
Furthermore, these nodes are always removed these lists before the operation completes. So for any configuration $C$, there are $O(\dot{c}(C))$ nodes in the $\UALL$, $\RUALL$, and $\PALL$.

Inserting and deleting from a lock-free linked list takes $O(\dot{c}(op) + L(op))$ amortized steps.
The number of nodes in the $\UALL$, $\RUALL$, and $\PALL$ at the start of $op$ is $O(\dot{c}(op))$.  So updating these lists costs $O(\dot{c}(op))$.

By inspection of the code, $op$ performs a constant number of steps for each node it encounters when traversing the $\UALL$ or $\PALL$.
When traversing the $\RUALL$, $op$ performs an atomic copy for each node it encounters, as well as a constant number of reads and writes. Single-writer atomic copy can be implemented from CAS with $O(1)$ worst-case step complexity~\cite{BlellochW20}.
In any case, $op$ performs a constant number of steps for each node it encounters when traversing the $\UALL$, $\RUALL$, or $\PALL$
There are $O(\bar{c}(op))$ nodes added to the $\UALL$, $\RUALL$, and $\PALL$ after the start of $op$.
For any execution $\alpha$, $\sum_{op \in \alpha} \bar{c}(op) \leq \sum_{op \in \alpha} 2\dot{c}(op)$.
So $op$ can traverse the $\UALL$, $\RUALL$, and $\PALL$ in $O(\dot{c}(op))$ amortized steps.
\end{proof}

We next consider the number of steps taken during instances of \textsc{NotifyPredOps}, which are invoked by \textsc{Insert} and \textsc{Delete} operations. Recall that it involves adding notify nodes into the notify lists of every predecessor node in $\PALL$. 


\begin{lemma}\label{lemma:amortized_notifyPred}
In any execution $\alpha$, the total number of steps taken by instances of \textsc{NotifyPredOps} is $\sum_{op \in \alpha} \dot{c}(op)^2$,  where $op \in \alpha$ denotes an \textsc{Insert}, \textsc{Delete}, or \textsc{Predecessor} operation, $op$, invoked during $\alpha$.
\end{lemma}

\begin{proof}
Let $uOp$ be an update operation invoking \textsc{NotifyPredOps}$(\uNode)$ (on line~\ref{ln:insert:notify} for \textsc{Insert} operations, or on lines~\ref{ln:delete:help_notify} and \ref{ln:delete:notifyPredOps} for \textsc{Delete} operations). 
On line~\ref{ln:notify:travuall}, the $\UALL$ is traversed. By Lemma~\ref{lemma:amortizedList}, this takes $O(\dot{c}(uOp))$ amortized steps.

Consider a successful CAS, $s$, performed by $uOp$ (on line~\ref{ln:sendNotification:CAS}) that changes the head of $\pNode.\notifyList$, where $\pNode$ is a predecessor node in the $\PALL$. Let $pOp$ be the operation 
that invoked the instance of \textsc{PredHelper} that created $\pNode$.
Note that a constant number of steps are performed from lines~\ref{ln:notify:first_activated} to \ref{ln:sendNotification:firstActivated} leading up to this CAS.
This successful CAS causes at most $\dot{c}(C)-1$ other operations to later perform an unsuccessful CAS on the same head pointer, where $C$ is the configuration immediately after this step. 
In particular, these are operations that read the head pointer on line~\ref{ln:sendNotification:readHead}, but have not yet performed the following CAS on line~\ref{ln:sendNotification:CAS}.
These unsuccessful CASs, as well as the constant number of steps leading up to the unsuccessful CAS, are charged to the successful CAS. 
So $O(\dot{c}(C))$ steps are charged to this successful CAS.

If $\pNode$ is in the $\PALL$ at the start of $uOp$, we let $uOp$ pay for the $O(\dot{c}(C))$ steps charged to the successful CAS. Since $C$ is a configuration sometime during 
$uOp$, $uOp$ is charged $O(\dot{c}(uOp))$ for this CAS.
There are $O(\dot{c}(uOp))$ predecessor nodes in the $\PALL$ at the start of $uOp$. So $uOp$ is charged $O(\dot{c}(uOp)^2)$ amortized steps in this manner.

Otherwise $\pNode$ is added to the $\PALL$ sometime during $uOp$. 
We let $pOp$ pay for $O(\dot{c}(C))$ steps charged to the successful CAS. 
If $C$ is a configuration sometime during $pOp$, $O(\dot{c}(pOp))$ steps are charged to $pOp$ for this CAS. 
So suppose $C$ occurs sometime after $pOp$ has completed.  
Every update operation that performs an unsuccessful CAS as a result of $s$ must have read $\pNode$ in the $\PALL$, but not yet performed a CAS by $C$. All such update operations are then active during $C$.
So $O(\dot{c}(pOp))$ steps are charged to $pOp$ for this CAS.
There are $O(\dot{c}(pOp))$ update operations active when $pOp$ adds $\pNode$ to the $\PALL$. So a total of $O(\dot{c}(pOp)^2)$ amortized  steps are charged to $pOp$.

Therefore, for any execution $\alpha$, the total number of steps taken by instances of \textsc{NotifyPredOps} is $\sum_{op \in \alpha} \dot{c}(op)^2$.
\end{proof}

We next consider instances of $pOp$ of \textsc{PredHelper}. Note that \textsc{Predecessor} operations invoke exactly one instance of \textsc{PredHelper}, while \textsc{Delete} operations invoke at most two instances of \textsc{PredHelper}. 

\begin{lemma}\label{lemma:amortized_predHelper}
The amortized number of steps taken by instances $pOp$ of \textsc{PredHelper} is $O(\dot{c}(pOp) + \tilde{c}(pOp) + \log u)$. 
\end{lemma}

\begin{proof}
Let $\pNode$ be the predecessor node created by $pOp$.
By Lemma~\ref{lemma:amortizedList}, adding $\pNode$ into the $\PALL$ (on line~\ref{ln:pred:insertPALL}) takes $O(\dot{c}(pOp))$ amortized steps. Additionally, traversing the $\PALL$ from lines~\ref{ln:pred:travPALL} to \ref{ln:pred:travPALLend} takes $O(\dot{c}(pOp))$ amortized steps.
By Lemma~\ref{lemma:amortizedList}, traversing the $\UALL$ and $\RUALL$ (on lines~\ref{ln:pred:uall1} and \ref{ln:pred:uall2}) takes $O(\dot{c}(pOp))$ amortized steps.
Performing \textsc{RelaxedPredecessor}$(y)$ (on line~\ref{ln:pred:travtrie}) takes $O(\log u)$ steps in the worst-case.

Only update operations that are active during $pOp$ add notify nodes into the notify list of $\pNode$ and they only add a constant number of notify nodes into it.
So the amortized number of steps taken for $pOp$ to traverse its own notify list is $O(\bar{c}(pOp)) = O(\dot{c}(pOp))$.

The steps in the if-block on line~\ref{ln:pred:fix_binary_trie_pred} to line~\ref{ln:pred:set_pred0} involves traversing the notify list of $\pNode$, as well as possibly of the notify list of a predecessor node $\pNode'$ created by some other operation, $pOp'$.
The total number of notify nodes added to the notify list of $\pNode'$ is $O(\bar{c}(pOp'))$.
So $pOp$ takes $O(\tilde{c}(pOp))$ steps to traverse the notify list of $O(pOp')$. 
By inspection of the code, all other steps take constant time.
In summary, the total amortized cost is $O(\dot{c}(pOp) + \tilde{c}(pOp) + \log u)$.
\end{proof}

We next give the amortized step complexity of \textsc{Insert}, \textsc{Delete}, and \textsc{Predecessor} operations by combining the previous lemmas.

\begin{lemma}
An \textsc{Insert} operation, $iOp$, has $O(\dot{c}(iOp)^2 + \log u)$ amortized cost.
A \textsc{Delete} operation, $dOp$, has $O(\dot{c}(dOp)^2 + \tilde{c}(dOp) + \log u)$ amortized cost.
\end{lemma}

\begin{proof}
Let $uOp$ be an \textsc{Insert} or \textsc{Delete} operation. 
The relaxed binary trie operations, \textsc{TrieInsert} and \textsc{TrieDelete} (invoked on lines~\ref{ln:insert:binaryTrie} and \ref{ln:delete:binaryTrie}), have $O(\log u)$  worst-case step complexity. This is because they perform a single traversal up the relaxed binary trie from a leaf to the root, which has length $O(\log u)$, performing a constant number of steps at each node of this path.

By inspection of the code, $uOp$ performs a constant number of updates and traversals to the $\UALL$, $\RUALL$, and $\PALL$ (outside of instances of \textsc{NotifyPredOps} and \textsc{PredHelper}).
By Lemma~\ref{lemma:amortizedList}, this takes $\dot{c}(uOp)$ amortized steps.
By Lemma~\ref{lemma:amortized_notifyPred}, $O(\dot{c}(uOp)^2)$ amortized steps are charged to $uOp$ for instances of \textsc{NotifyPredOps}.
If $uOp$ is a \textsc{Delete} operation, it invokes \textsc{PredHelper} twice to perform its embedded predecessor operations. By Lemma~\ref{lemma:amortized_predHelper}, these instances of \textsc{PredHelper} take  $O(\dot{c}(uOp) + \tilde{c}(uOp) + \log u)$ amortized steps.
It follows by inspection of the code that all other parts of \textsc{Insert} and \textsc{Delete} operations take a constant number of steps.

\end{proof}

\begin{lemma}
A \textsc{Predecessor} operation, $pOp$, has $O(\dot{c}(pOp)^2 + \tilde{c}(pOp) + \log u)$ amortized cost.
\end{lemma}

\begin{proof}
A \textsc{Predecessor} operation, $pOp$, invokes a single instance of \textsc{PredHelper} (on line~\ref{ln:pred:PredHelper}). By Lemma~\ref{lemma:amortized_predHelper}, this takes  $O(\dot{c}(pOp) + \tilde{c}(pOp) + \log u)$ amortized steps.
The only other steps performed by $pOp$ are to remove its predecessor node from the $\PALL$ (on line~\ref{ln:pred:removePall}). By Lemma~\ref{lemma:amortizedList}, this has an amortized cost of $O(\dot{c}(pOp))$. 
\end{proof}

The following theorem summarizes the results of this section.

\begin{theorem}
We give a linearizable, lock-free implementation of a binary trie for a universe of size $u$ supporting \textsc{Search} with $O(1)$ worst-case step complexity, \textsc{Delete} and \textsc{Predecessor} with $O(\dot{c}(op)^2 + \tilde{c}(op) + \log u)$ amortized cost, and \textsc{Insert} with $O(\dot{c}(op)^2 + \log u)$ amortized cost.
\end{theorem}

\section{Conclusion}\label{section_conclusion}

The main contribution of this chapter is a deterministic, lock-free implementation of a binary trie using read, write, CAS, and AND operations. 
We prove that the implementation is linearizable. We show that it supports \textsc{Search} with $O(1)$ worst-case step complexity, \textsc{Delete} and \textsc{Predecessor} with $O(\dot{c}(op)^2 + \tilde{c}(op) + \log u)$ amortized cost, and \textsc{Insert} with $O(\dot{c}(op)^2 + \log u)$ amortized cost.

The implementation
uses a relaxed binary trie as one of its components.
All update operations on the relaxed binary trie  take $O(\log u)$ steps in the worst-case.
Each predecessor operation on the relaxed binary trie takes $O(\log u)$ steps in the worst-case,
since it can complete without having to help concurrent update operations.

It is possible extend our lock-free binary trie to support \textsc{Max}, which returns the largest key in $S$. This can be done by extending the binary trie to represent an additional key $\infty$ that is larger than all keys in $U$, and then performing  \textsc{Predecessor}$(\infty)$. By symmetry, \textsc{Successor} and \textsc{Min} can also be supported.

Our lock-free binary trie has been implemented \cite{JMalek24}.
The implementation uses a version of epoch-based memory reclamation based on DEBRA~\cite{Brown15} to avoid ABA problems when accessing dynamically allocated objects.
Its performance was compared to those of other lock-free data structures supporting \textsc{Predecessor}.

In our lock-free binary trie,
predecessor operations get information about update operations that announce themselves in the update announcement linked list.
Predecessor operations also announce themselves in the predecessor announcement linked list, so that update update operations
can give them information.
There is an amortized cost of $O(\dot{c}(op)^2)$ for an update operation, $op$, to give information to all predecessor operations.
We would like to obtain a more efficient algorithm to do this, which will result in a more efficient implementation of a lock-free binary trie. 


A sequential van Emde Boas trie supports \textsc{Search}, \textsc{Insert}, \textsc{Delete}, and  \textsc{Predecessor} in $O(\log\log u)$ worst-case time. We conjecture that there is a lock-free implementation supporting operations with  $O(\dot{c}(op)^2 + \tilde{c}(op) + \log\log u)$ amortized step complexity. 
Since the challenges are similar, we believe our techniques for implementing a lock-free binary trie will be useful for implementing a lock-free van Emde Boas trie. In particular, using an implementation of a relaxed van Emde Boas trie satisfying the specification from Section~\ref{section_relaxed_binary_trie_specification} should be a good approach.

\bibliography{sources}

\begin{thebibliography}{10}

\bibitem{AfekDT95}
Yehuda Afek, Dalia Dauber, and Dan Touitou.
\newblock Wait-free made fast (extended abstract).
\newblock In {\em Proceedings of the Twenty-Seventh Annual {ACM} Symposium on
  Theory of Computing}, pages 538--547, 1995.

\bibitem{AndersonM95}
James~H. Anderson and Mark Moir.
\newblock Universal constructions for multi-object operations.
\newblock In {\em Proceedings of the Fourteenth Annual {ACM} Symposium on
  Principles of Distributed Computing}, pages 184--193, 1995.

\bibitem{Arbel-Raviv017}
Maya Arbel{-}Raviv and Trevor Brown.
\newblock Reuse, don't recycle: Transforming lock-free algorithms that throw
  away descriptors.
\newblock In {\em 31st International Symposium on Distributed Computing},
  volume~91 of {\em LIPIcs}, pages 4:1--4:16, 2017.

\bibitem{Barnes93}
Greg Barnes.
\newblock A method for implementing lock-free shared-data structures.
\newblock In Lawrence Snyder, editor, {\em Proceedings of the 5th Annual {ACM}
  Symposium on Parallel Algorithms and Architectures, {SPAA} '93}, pages
  261--270. {ACM}, 1993.

\bibitem{BedinLMPP21}
Denis B{\'{e}}din, Fran{\c{c}}ois L{\'{e}}pine, Achour Most{\'{e}}faoui, Damien
  Perez, and Matthieu Perrin.
\newblock Wait-free cas-based algorithms: The burden of the past.
\newblock In {\em 35th International Symposium on Distributed Computing, {DISC}
  2021}, volume 209, pages 11:1--11:15, 2021.

\bibitem{BieganskiRCR94}
Paul Bieganski, John Riedl, John~V. Carlis, and Ernest~F. Retzel.
\newblock Generalized suffix trees for biological sequence data: Applications
  and implementation.
\newblock In {\em 27th Annual Hawaii International Conference on System
  Sciences (HICSS-27)}, pages 35--44. {IEEE} Computer Society, 1994.

\bibitem{BlellochW20}
Guy~E. Blelloch and Yuanhao Wei.
\newblock {LL/SC} and atomic copy: Constant time, space efficient
  implementations using only pointer-width {CAS}.
\newblock In {\em 34th International Symposium on Distributed Computing, {DISC}
  2020}, pages 5:1--5:17, 2020.

\bibitem{BraginskyP12}
Anastasia Braginsky and Erez Petrank.
\newblock A lock-free b+tree.
\newblock In {\em 24th {ACM} Symposium on Parallelism in Algorithms and
  Architectures, {SPAA} '12}, pages 58--67. {ACM}, 2012.

\bibitem{Brown14}
Trevor Brown.
\newblock B-slack trees: Space efficient b-trees.
\newblock In {\em Algorithm Theory - {SWAT} 2014 - 14th Scandinavian Symposium
  and Workshops}, volume 8503 of {\em Lecture Notes in Computer Science}, pages
  122--133, 2014.

\bibitem{Brown15}
Trevor Brown.
\newblock Reclaiming memory for lock-free data structures: There has to be a
  better way.
\newblock In {\em Proceedings of the 2015 {ACM} Symposium on Principles of
  Distributed Computing, {PODC} 2015}, pages 261--270, 2015.

\bibitem{BrownThesis17}
Trevor Brown.
\newblock {\em Techniques for Constructing Efficient Lock-Free Data
  Structures}.
\newblock PhD thesis, Department of Computer Science, University of Toronto,
  2017.

\bibitem{BrownA12}
Trevor Brown and Hillel Avni.
\newblock Range queries in non-blocking \emph{k}-ary search trees.
\newblock In {\em Principles of Distributed Systems, 16th International
  Conference, {OPODIS} 2012}, volume 7702 of {\em Lecture Notes in Computer
  Science}, pages 31--45. Springer, 2012.

\bibitem{BrownER13}
Trevor Brown, Faith Ellen, and Eric Ruppert.
\newblock Pragmatic primitives for non-blocking data structures.
\newblock In {\em {ACM} Symposium on Principles of Distributed Computing,
  {PODC} '13}, pages 13--22. {ACM}, 2013.

\bibitem{BrownER14}
Trevor Brown, Faith Ellen, and Eric Ruppert.
\newblock A general technique for non-blocking trees.
\newblock In {\em Proceedings of the Symposium on Principles and Practice of
  Parallel Programming (PPoPP)}, pages 329--342, 2014.

\bibitem{BrownPA20}
Trevor Brown, Aleksandar Prokopec, and Dan Alistarh.
\newblock Non-blocking interpolation search trees with doubly-logarithmic
  running time.
\newblock In {\em PPoPP '20: 25th {ACM} {SIGPLAN} Symposium on Principles and
  Practice of Parallel Programming}, pages 276--291. {ACM}, 2020.

\bibitem{Chatterjee17}
Bapi Chatterjee.
\newblock Lock-free linearizable 1-dimensional range queries.
\newblock In {\em Proceedings of the 18th International Conference on
  Distributed Computing and Networking}, page~9. {ACM}, 2017.

\bibitem{ChatterjeeDT13}
Bapi Chatterjee, Nhan~Nguyen Dang, and Philippas Tsigas.
\newblock Efficient lock-free binary search trees.
\newblock In {\em Proceedings of the {ACM} Symposium on Principles of
  Distributed Computing (PODC)}, pages 322--331, 2014.

\bibitem{ChuongER10}
Phong Chuong, Faith Ellen, and Vijaya Ramachandran.
\newblock A universal construction for wait-free transaction friendly data
  structures.
\newblock In {\em {SPAA} 2010: Proceedings of the 22nd Annual {ACM} Symposium
  on Parallelism in Algorithms and Architectures}, pages 335--344, 2010.

\bibitem{DegermarkBCP97}
Mikael Degermark, Andrej Brodnik, Svante Carlsson, and Stephen Pink.
\newblock Small forwarding tables for fast routing lookups.
\newblock In {\em Proceedings of the {ACM} {SIGCOMM} 1997 Conference on
  Applications, Technologies, Architectures, and Protocols for Computer
  Communication}, pages 3--14. {ACM}, 1997.

\bibitem{DrachslerVY14}
Dana Drachsler, Martin~T. Vechev, and Eran Yahav.
\newblock Practical concurrent binary search trees via logical ordering.
\newblock In {\em {ACM} {SIGPLAN} Symposium on Principles and Practice of
  Parallel Programming, PPoPP '14}, pages 343--356. {ACM}, 2014.

\bibitem{EllenFHR13}
Faith Ellen, Panagiota Fatourou, Joanna Helga, and Eric Ruppert.
\newblock The amortized complexity of non-blocking binary search trees.
\newblock In {\em Proceedings of the Symposium on Principles of Distributed
  Computing (PODC)}, pages 332--340, 2014.

\bibitem{EllenFKMT16}
Faith Ellen, Panagiota Fatourou, Eleftherios Kosmas, Alessia Milani, and
  Corentin Travers.
\newblock Universal constructions that ensure disjoint-access parallelism and
  wait-freedom.
\newblock {\em Distributed Comput.}, 29(4):251--277, 2016.

\bibitem{EllenFRB10}
Faith Ellen, Panagiota Fatourou, Eric Ruppert, and Franck van Breugel.
\newblock Non-blocking binary search trees.
\newblock In {\em Proceedings of the 29th Annual {ACM} Symposium on Principles
  of Distributed Computing (PODC)}, pages 131--140, 2010.

\bibitem{FatourouK11}
Panagiota Fatourou and Nikolaos~D. Kallimanis.
\newblock A highly-efficient wait-free universal construction.
\newblock In {\em {SPAA} 2011: Proceedings of the 23rd Annual {ACM} Symposium
  on Parallelism}, pages 325--334, 2011.

\bibitem{FatourouKK19}
Panagiota Fatourou, Nikolaos~D. Kallimanis, and Eleni Kanellou.
\newblock An efficient universal construction for large objects.
\newblock In {\em 23rd International Conference on Principles of Distributed
  Systems, {OPODIS} 2019}, volume 153 of {\em LIPIcs}, pages 18:1--18:15, 2019.

\bibitem{FatourouPR19}
Panagiota Fatourou, Elias Papavasileiou, and Eric Ruppert.
\newblock Persistent non-blocking binary search trees supporting wait-free
  range queries.
\newblock In {\em The 31st {ACM} on Symposium on Parallelism in Algorithms and
  Architectures, {SPAA} 2019}, pages 275--286. {ACM}, 2019.

\bibitem{FatourouR24}
Panagiota Fatourou and Eric Ruppert.
\newblock Lock-free augmented trees.
\newblock In {\em 38th International Symposium on Distributed Computing, {DISC}
  2024}, volume 319 of {\em LIPIcs}, pages 23:1--23:24, 2024.

\bibitem{FomitchevR04}
Mikhail Fomitchev and Eric Ruppert.
\newblock Lock-free linked lists and skip lists.
\newblock In {\em Proceedings of the Twenty-Third Annual {ACM} Symposium on
  Principles of Distributed Computing (PODC)}, pages 50--59, 2004.

\bibitem{GiakkoupisGW21}
George Giakkoupis, Mehrdad~Jafari Giv, and Philipp Woelfel.
\newblock Efficient randomized {DCAS}.
\newblock In {\em {STOC} '21: 53rd Annual {ACM} {SIGACT} Symposium on Theory of
  Computing}, pages 1221--1234, 2021.

\bibitem{GolabHW11}
Wojciech~M. Golab, Lisa Higham, and Philipp Woelfel.
\newblock Linearizable implementations do not suffice for randomized
  distributed computation.
\newblock In {\em Proceedings of the 43rd {ACM} Symposium on Theory of
  Computing, {STOC} 2011}, pages 373--382, 2011.

\bibitem{Harris01}
Timothy~L. Harris.
\newblock A pragmatic implementation of non-blocking linked-lists.
\newblock In {\em Distributed Computing, 15th International Conference, {DISC}
  2001}, pages 300--314, 2001.

\bibitem{HarrisFP02}
Timothy~L. Harris, Keir Fraser, and Ian~A. Pratt.
\newblock A practical multi-word compare-and-swap operation.
\newblock In {\em Distributed Computing, 16th International Conference, {DISC}
  2002}, volume 2508 of {\em Lecture Notes in Computer Science}, pages
  265--279. Springer, 2002.

\bibitem{Herlihy91}
Maurice Herlihy.
\newblock Wait-free synchronization.
\newblock {\em {ACM} Trans. Program. Lang. Syst.}, 13(1):124--149, 1991.

\bibitem{Herlihy93}
Maurice Herlihy.
\newblock A methodology for implementing highly concurrent objects.
\newblock {\em {ACM} Trans. Program. Lang. Syst.}, 15(5):745--770, 1993.

\bibitem{HerlihyW90}
Maurice Herlihy and Jeannette~M. Wing.
\newblock Linearizability: {A} correctness condition for concurrent objects.
\newblock {\em Trans. Program. Lang. Syst.}, 12(3):463--492, 1990.

\bibitem{HowleyJ12}
Shane~V. Howley and Jeremy Jones.
\newblock A non-blocking internal binary search tree.
\newblock In {\em 24th {ACM} Symposium on Parallelism in Algorithms and
  Architectures, {SPAA} '12}, pages 161--171. {ACM}, 2012.

\bibitem{Ko18}
Jeremy Ko.
\newblock {The Amortized Analysis of a Non-blocking Chromatic Tree}.
\newblock In {\em 22nd International Conference on Principles of Distributed
  Systems (OPODIS 2018)}, volume 125, pages 8:1--8:17, 2018.

\bibitem{JMalek24}
Jordan Malek.
\newblock An implementation and experimental comparison of dynamic ordered
  sets.
\newblock Master's thesis, Department of Computer Science, University of
  Toronto, 2024.
\newblock URL: \url{https://arxiv.org/abs/2411.01090}.

\bibitem{Martinez-Prieto16}
Miguel~A. Mart{\'{\i}}nez{-}Prieto, Nieves~R. Brisaboa, Rodrigo C{\'{a}}novas,
  Francisco Claude, and Gonzalo Navarro.
\newblock Practical compressed string dictionaries.
\newblock {\em Inf. Syst.}, 56:73--108, 2016.

\bibitem{NatarajanM14}
Aravind Natarajan and Neeraj Mittal.
\newblock Fast concurrent lock-free binary search trees.
\newblock In {\em {ACM} {SIGPLAN} Symposium on Principles and Practice of
  Parallel Programming, PPoPP '14}, pages 317--328. {ACM}, 2014.

\bibitem{OshmanS13}
Rotem Oshman and Nir Shavit.
\newblock The skiptrie: low-depth concurrent search without rebalancing.
\newblock In {\em {ACM} Symposium on Principles of Distributed Computing,
  {PODC} '13}, pages 23--32, 2013.

\bibitem{PetrankT13}
Erez Petrank and Shahar Timnat.
\newblock Lock-free data-structure iterators.
\newblock In {\em Distributed Computing - 27th International Symposium, {DISC}
  2013}, volume 8205 of {\em Lecture Notes in Computer Science}, pages
  224--238. Springer, 2013.

\bibitem{ProkopecBBO12}
Aleksandar Prokopec, Nathan~Grasso Bronson, Phil Bagwell, and Martin Odersky.
\newblock Concurrent tries with efficient non-blocking snapshots.
\newblock In {\em Proceedings of the 17th {ACM} {SIGPLAN} Symposium on
  Principles and Practice of Parallel Programming, {PPOPP} 2012}, pages
  151--160. {ACM}, 2012.

\bibitem{Pugh89}
William~W. Pugh.
\newblock Skip lists: {A} probabilistic alternative to balanced trees.
\newblock In {\em Algorithms and Data Structures, Workshop {WADS}}, volume 382
  of {\em Lecture Notes in Computer Science}, pages 437--449. Springer, 1989.

\bibitem{Shafiei13}
Niloufar Shafiei.
\newblock Non-blocking patricia tries with replace operations.
\newblock In {\em {IEEE} 33rd International Conference on Distributed Computing
  Systems, {ICDCS} 2013}, pages 216--225. {IEEE} Computer Society, 2013.

\bibitem{Shafiei15}
Niloufar Shafiei.
\newblock Non-blocking doubly-linked lists with good amortized complexity.
\newblock In {\em Proceedings of the 19th International Conference on
  Principles of Distributed Systems (OPODIS)}, pages 35:1--35:17, 2015.

\bibitem{ShalevS03}
Ori Shalev and Nir Shavit.
\newblock Split-ordered lists: lock-free extensible hash tables.
\newblock In {\em Proceedings of the Twenty-Second {ACM} Symposium on
  Principles of Distributed Computing, {PODC} 2003}, pages 102--111. {ACM},
  2003.

\bibitem{Valois95}
John~D. Valois.
\newblock Lock-free linked lists using compare-and-swap.
\newblock In {\em Proceedings of the Fourteenth Annual {ACM} Symposium on
  Principles of Distributed Computing}, pages 214--222, 1995.

\bibitem{Boas77}
Peter van Emde~Boas.
\newblock Preserving order in a forest in less than logarithmic time and linear
  space.
\newblock {\em Inf. Process. Lett.}, 6(3):80--82, 1977.

\bibitem{BoasKZ77}
Peter van Emde~Boas, R.~Kaas, and E.~Zijlstra.
\newblock Design and implementation of an efficient priority queue.
\newblock {\em Math. Syst. Theory}, 10:99--127, 1977.

\bibitem{WeiBBFR021}
Yuanhao Wei, Naama Ben{-}David, Guy~E. Blelloch, Panagiota Fatourou, Eric
  Ruppert, and Yihan Sun.
\newblock Constant-time snapshots with applications to concurrent data
  structures.
\newblock In {\em PPoPP '21: 26th {ACM} {SIGPLAN} Symposium on Principles and
  Practice of Parallel Programming}, pages 31--46, 2021.

\bibitem{Willard83}
Dan~E. Willard.
\newblock Log-logarithmic worst-case range queries are possible in space
  theta(n).
\newblock {\em Inf. Process. Lett.}, 17(2):81--84, 1983.

\end{thebibliography}

\end{document}